\documentclass[pra,longbibliography,twocolumn,showpacs,nofootinbib,superscriptaddress,notitlepage]{revtex4-2}
\usepackage{amsmath}
\usepackage{amssymb,bm}
\usepackage{amsthm}

\newtheorem{innercustomgeneric}{\customgenericname}
\providecommand{\customgenericname}{}
\newcommand{\newcustomtheorem}[2]{%
  \newenvironment{#1}[1]
  {%
   \renewcommand\customgenericname{#2}%
   \renewcommand\theinnercustomgeneric{##1}%
   \innercustomgeneric
  }
  {\endinnercustomgeneric}
}

\newcustomtheorem{customthm}{Theorem}
\newcustomtheorem{definitionBob}{Definition}

\usepackage{color,dsfont} 
\usepackage{graphicx}
\usepackage{subcaption}
\usepackage{ragged2e}
\DeclareCaptionJustification{justified}{\justifying}
\usepackage[colorlinks=true, hyperindex, breaklinks, linkcolor=blue, urlcolor=blue, citecolor=blue]{hyperref} 
\usepackage[normalem]{ulem}
\usepackage[capitalise]{cleveref}
\usepackage{mathrsfs}

\usepackage{mathtools}
\usepackage{float}
\usepackage{verbatim}
\usepackage{latexsym}
\usepackage{amsmath}
\usepackage{amssymb}
\usepackage{setspace}
\usepackage{amsfonts}
\usepackage{stmaryrd}
\usepackage{xcolor}
\usepackage{enumitem}
\usepackage{hhline}
\usepackage[normalem]{ulem}

\newtheorem{theorem}{Theorem} \newtheorem{lemma}{Lemma}
 
\newtheorem{claim}{Claim} \newtheorem{definition}{Definition}
\newtheorem{proposition}{Proposition}

\newtheorem{innercustomthm}{Condition}
\newenvironment{condition}[1]
{\renewcommand\theinnercustomthm{#1}\innercustomthm}
{\endinnercustomthm}

\newcommand{\codepar}[1]{\ensuremath{[\![#1]\!]}}

\newcommand{\HR}[1]{\textcolor{red} {#1}}

\newcommand{\DL}[1]{\textcolor{brown}{#1}}

\crefname{equation}{Eq.\!}{Eqs.\!}
\crefname{figure}{Fig.\!}{Figs.\!}
\usepackage{multirow}
\mathchardef\mhyphen="2D

\begin{document}

\title{Achieving fault tolerance on capped color codes with few ancillas}


\author{Theerapat Tansuwannont}
\email{t.tansuwannont@duke.edu}
\affiliation{
    Institute for Quantum Computing and Department of Physics and Astronomy,
    University of Waterloo,
    Waterloo, Ontario, N2L 3G1, Canada
    }
\affiliation{
	Department of Electrical and Computer Engineering, Duke University, Durham, NC 27708, USA
}
    
\author{Debbie Leung}
\email{wcleung@uwaterloo.ca}
\affiliation{
   Institute for Quantum Computing and Department of Combinatorics and Optimization,
    University of Waterloo,
    Waterloo, Ontario, N2L 3G1, Canada
    }
\affiliation{
	Perimeter Institute for Theoretical Physics, 
	Waterloo, Ontario, N2L 2Y5, Canada
}

\begin{abstract}
Attaining fault tolerance while maintaining low overhead is one of the main challenges in a practical implementation of quantum circuits. 
One major technique that can overcome this problem is the flag technique, in which high-weight errors arising from a few faults can be detected by a few ancillas and distinguished using subsequent syndrome measurements. The technique can be further improved using the fact that for some families of codes, errors of any weight are logically equivalent if they have the same syndrome and weight parity, as previously shown in \cite{TL20}. 
In this work, we develop a notion of distinguishable fault set which captures both concepts of flags and weight parities, and extend the use of weight parities in error correction from \cite{TL20} to families of capped and recursive capped color codes. We also develop fault-tolerant protocols for error correction, measurement, state preparation, and logical $T$ gate implementation via code switching, which are sufficient for performing fault-tolerant Clifford computation on a capped color code, and performing fault-tolerant universal quantum computation on a recursive capped color code. Our protocols for a capped or a recursive capped color code of any distance require only 2 ancillas, assuming that the ancillas can be reused. The concept of distinguishable fault set also leads to a generalization of the definitions of fault-tolerant gadgets proposed by Aliferis, Gottesman, and Preskill.
\end{abstract}


\maketitle

\section{Introduction}
\label{sec:Intro}%


Fault-tolerant error correction (FTEC), a procedure which suppresses error propagation in a quantum circuit, is one of the most important components for building large-scale quantum computers. Given that the physical error rate is below some constant threshold value, an FTEC scheme along with other schemes for fault-tolerant quantum computation (FTQC) allow us to fault-tolerantly simulate any quantum circuit with arbitrarily low logical error rates \cite{Shor96,AB08,Kitaev97,KLZ96,Preskill98,TB05,ND05,AL06,AGP06}. However, lower logical error rate requires more overhead (e.g., quantum gates and ancilla qubits) \cite{Steane03,PR12,CJL16b,TYC17}. Therefore, fault-tolerant protocols which require a small number of ancillas and give high threshold value are very desirable for practical implementation.


Traditional FTEC schemes require substantial number of ancillas for error syndrome measurements. For example, the Shor error correction (EC) scheme \cite{Shor96,DA07} which is applicable to any stabilizer code requires as many ancillas as the maximum weight of the stabilizer generators. The Knill EC scheme \cite{Knill05a}, which is also applicable to any stabilizer code, requires two code block of ancillas. Meanwhile, the Steane EC scheme \cite{Steane97,Steane02} which is applicable to any CSS code requires one code block of ancillas. (The Shor scheme also requires repeated syndrome measurement, while the Knill and the Steane schemes do not.) There are several recently proposed schemes which require fewer ancillas. Yoder and Kim proposed an FTEC scheme for the \codepar{7,1,3} code which requires only 2 ancillas \cite{YK17}, and their scheme is further developed into a well-known flag FTEC scheme for the \codepar{5,1,3} code and the \codepar{7,1,3} code which also require only 2 ancillas \cite{CR17a} (where an \codepar{n,k,d} stabilizer code encodes $k$ logical qubits into $n$ physical qubits and has distance $d$). In general, a flag FTEC scheme for any stabilizer code requires as few as $d+1$ ancillas where $d$ is the code distance \cite{CR20}, with further reduction known for certain families of codes \cite{CR17a,CB18,TCL20,CKYZ20,CZYHC20}. The flag technique can also be applied to other schemes for FTQC \cite{CR17b,CC19,SCC19,BCC19,BXG19,Vui18,GMB19,LA19,DB20,RBMS21}.


How errors spread during the protocols depends on several factors such as the order of quantum gates in the circuits for syndrome measurement and the choice of stabilizer generators being measured. The idea behind the flag technique is that a few ancillas are added to the circuits in order to detect errors of high weight arising from a few faults, and the errors will be distinguished by their syndromes obtained from subsequent syndrome measurements. Note that some possible errors may be logically equivalent and need not be distinguished, and for some families of codes, we can tell whether the errors are logically equivalent using their syndromes and error weight parities. Reference \cite{TL20} combines the ideas of flags and weight parities to construct an FTEC scheme for a \codepar{49,1,9} concatenated Steane code, which can correct up to 3 faults and requires only 2 ancillas. In such a scheme, the weight parity of the error in each subblock, which is the lower-level \codepar{7,1,3} code, are determined by the results from measuring the generators of the higher-level \codepar{7,1,3} code. The scheme in \cite{TL20} uses very few ancillas compared to conventional schemes for a concatenated code (which is constructed by replacing physical qubit by a code block) and is expected to be applicable to concatenated codes other than the \codepar{49,1,9} code.  


There are families of codes that attain high distance without code concatenation. Topological codes in which the code distance can be made arbitrarily large by increasing the lattice size are good candidates for practical implementation of quantum computers since fault-tolerant protocols for these codes typically give very high accuracy thresholds \cite{DKLP02,DP10,BH13,DP13,ABCB14,BSV14,Delfosse14,BLPSW16,BNB16,CR18_deep,DBT18,DP18,KD19,KP19,MKJ19,NB19,VBK21}. Examples of two-dimensional (2D) topological stabilizer codes are 2D toric codes \cite{Kitaev97,BK98} and 2D color codes \cite{BM06}. These codes are suitable for physical implementations using superconducting qubits \cite{FMMC12,CZYHC20,CKYZ20} and qubits realized by Majorana zero modes \cite{KarzigScalable17,CBDH20} since qubits can be arranged on a 2D plane and only quantum gates involving neighboring qubits are required. Toric codes and color codes can be transformed to one another using the techniques developed in \cite{KYP15} (see also \cite{VB19}).


The simplest way to perform FTQC on a topological stabilizer code is to implement logical gates by applying physical gates transversally since doing so does not spread errors (therefore fault tolerant). Unfortunately, it is known by the Eastin-Knill theorem that a universal set of quantum operations cannot be achieved using only transversal gates \cite{EK09}. Moreover, logical gates which can be implemented transversally on a 2D topological stabilizer code are in the Clifford group \cite{BK13} (see also \cite{PY15}).
The Clifford group can be generated by the Hadamard gate ($H$), the $\frac{\pi}{4}$-gate ($S$), and the CNOT gate \cite{CRSS97,Gottesman98b}.
A transversal CNOT gate is achievable by both 2D toric codes and 2D color codes since these codes are in the CSS code family \cite{CS96,Steane96b}.  In addition, the 2D color codes have transversal $H$ and $S$ gates \cite{BM06}, so, any Clifford operation can be implemented transversally on any 2D color code.


Implementing only Clifford gates on a 2D color code is not particularly interesting since Clifford operation can be efficiently simulated by a classical computer (the result is known as Gottesman-Knill theorem) \cite{Gottesman97,NC00}. However, universality can be achieved by Clifford gates together with any gate not in the Clifford group \cite{NRS01}.
There are two compelling approaches for implementing a non-Clifford gate on a 2D color code: magic state distillation \cite{BK05} and code switching \cite{PR13,ADP14,Bombin15,KB15}. The former approach focuses on producing high-fidelity $T$ states from noisy $T$ states and Clifford operations, where $|T\rangle = (|0\rangle+\sqrt{i}|1\rangle)/\sqrt{2}$ is the state that can be used to implement non-Clifford $T = \bigl( \begin{smallmatrix}1 & 0\\ 0 & \sqrt{i} \end{smallmatrix}\bigr)$ operation. By replacing any physical gates and qubits with logical gates and blocks of code, a logical $T$ gate can be implemented using a method similar to that proposed in \cite{BK05}. The latter approach uses the gauge fixing method to switch between a 2D color code (in which Clifford gates are transversal) and a 3D color code (in which the $T$ gate is transversal). A recent study \cite{BKS21} which compares the overhead required for these two approaches shows that code switching does not outperform magic state distillation when certain FT schemes are used, except for some small values of physical error rate. Nevertheless, their results do not rule out the possibilities of FT schemes have yet to be discovered, in which the authors are hopeful that such schemes could reduce the overhead required for either of the aforementioned approaches.


The EC technique using weight parities introduced in \cite{TL20} was originally developed for the [[49,1,9]] code obtained from concatenating the \codepar{7,1,3} codes. The \codepar{7,1,3} code is also the smallest 2D color code. Surprisingly, we find that 2D color codes of any distance have certain properties which make similar technique applicable, under appropriate modifications of the original code to be described in this paper.
In order to obtain the weight parity of an error on a 2D color code, we need to make measurements of stabilizer generators of a bigger code which contains the 2D color code as a subcode.
In contrast to \cite{TL20}, the bigger code in this work is not obtained from code concatenation. Our development for FTEC protocols leads to a family of capped color codes, which are CSS subsystem codes \cite{Poulin05,Bacon06}. We study two stabilizer codes obtained from a (subsystem) capped color code through gauge fixing, namely capped color codes in H form and T form.
The code in H form which contains a 2D color code as a subcode has transversal Clifford gates, while the code in T form has transversal CNOT and transversal $T$ gates. In fact, our capped color codes bear similarities to the subsystem codes presented in \cite{JB16,BC15,JBH16}, in which qubits can be arranged on a 2D plane. In this work, we focus mainly on the construction of circuits for measuring generators of a capped color code in H form, and the construction of an FTEC scheme as well as other fault-tolerant schemes for measurement, state preparation, and Clifford operation. We also prove that our fault-tolerant schemes for capped color codes in H form of \emph{any distance} require only 2 ancillas (assuming that the ancillas can be reused). In addition, we construct a family of recursive capped color codes by recursively encoding the top qubit of capped color codes. Circuits for measuring generators of capped color codes in H form also work for recursive capped color codes, so fault-tolerant Clifford computation on a recursive capped color code of any distance using only 2 ancillas is possible. We also show that a logical $T$ gate can be fault-tolerantly implemented on a recursive capped color code of any distance using only 2 ancillas via code switching, leading to a complete set of operations for fault-tolerant universal quantum computation.




This paper is organized as follows: In \cref{sec:flag_n_WPEC}, we provide a brief review on EC technique using flags and error weight parities. We also develop the notion of distinguishable fault set in \cref{def:distinguishable}, which is the central idea of this work. In \cref{sec:3D_code}, we review basic properties of the 3D color code of distance 3 (which is defined as a subsystem code). We then provide a construction of circuits for measuring the stabilizer generators of the 3D color code in H form which give a distinguishable fault set. In \cref{sec:CCC}, we define families of capped and recursive capped color codes, whose properties are very similar to those of the 3D color code of distance 3. Afterwards, circuits for measuring the stabilizer generators of the capped color code in H form are constructed using ideas from the previous section. 
We prove \cref{thm:main} which states sufficient conditions for the circuits that can give a distinguishable fault set, then prove \cref{thm:main2,thm:main3} which state that for a capped color code in H form of any distance, a distinguishable fault set can be obtained if the circuits for measuring generators are flag circuits of a particular form.
The circuits which work for capped color codes are also applicable to recursive capped color codes.
In \cref{sec:FT_protocol}, we discuss an alternative version of fault-tolerant gadgets whose definitions are modified so that they are compatible with the notion of distinguishable fault set. 
Afterwards, we construct fault-tolerant protocols for capped and recursive capped color codes in H form. Some protocols described in this work are also applicable to other stabilizer codes whose generator measurement circuits give a distinguishable fault set.
Last, we discuss our results and provide directions for future work in \cref{sec:discussions}.

\section{Flags and error weight parities in error correction}
\label{sec:flag_n_WPEC}

In this section, we start by providing a brief review on the flag EC technique applied to the case of one fault in \cref{subsec:flag_ana}. Next, we extend the idea to the case of multiple faults in \cref{subsec:fault_set} and introduce the notion of distinguishable fault set in \cref{def:distinguishable}. Afterwards, we explain how weight parities can be used in error correction in \cref{subsec:WPEC_ana}. The equivalence of Pauli errors with the same syndrome and weight parity proved for the \codepar{7,1,3} Steane code in \cite{TL20} is also extended to a bigger family of codes in \cref{lem:err_equivalence}.

\subsection{Flag error correction}
\label{subsec:flag_ana}





Quantum computation is prone to noise, and an error on a few qubits can spread and cause a big problem in the computation if the error is not treated properly. One way to protect quantum data against noise is to use a quantum error correcting code (QECC) to encode a small number of logical qubits into a larger number of physical qubits. A quantum \codepar{n,k,d} stabilizer code \cite{Gottesman96,Gottesman97} encodes $k$ logical qubits into $n$ physical qubits and can correct errors up to weight $\tau = \lfloor(d-1)/2\rfloor$. Quantum error correction (QEC) is a process that aims to undo the corruption that happens to a codeword. 

A stabilizer code is a simultaneous $+1$ eigenspace of a list of commuting independent Pauli operators; they generates the stabilizer group for the code. For a stabilizer code, the error correction (EC) procedure involves measurements of stabilizer generators, which results in an error syndrome. The QEC is designed so that the more likely Pauli errors are either logically equivalent or have distinguishable syndrome. 
If the weight of the Pauli error $E$ occurred to a codeword is no bigger than $\tau$, $E$ can be identified by the error syndrome $\vec{s}(E)$ obtained from the generator measurements, and be corrected by applying $E^\dagger$ to the codeword.

The above working principle for a stabilizer code assumes that the syndrome measurements are perfect. In practice, every step in a quantum computation, including those in the syndrome measurements, is subject to error. An initial error can lead to a complex overall effect in the circuit. We adhere to the following terminologies and noise model in our discussion.

\begin{definition}{Location, noise model, and fault} \cite{AGP06}
	
	A circuit consists of a number of time steps and a number of qubits and is specified by operations to the qubits in each time step.  The operations can be single qubit state preparation, 1- or 2-qubit gates, or single qubit measurement.  (When nothing happens to a qubit, it goes through the 1-qubit gate of identity.)  
	A \emph{location} is labeled by a time step and the index (or indices) of a qubit (or pair of qubits) involved in an operation.  
	
	We consider the circuit-level noise in which every location is followed by \emph{depolarizing noise}: every one-qubit operation is followed by a single-qubit Pauli error $I, X, Y,$ or $Z$, and every two-qubit operation is followed by a two-qubit Pauli error of the form $P_1\otimes P_2$ where $P_1,P_2 \in \{I,X,Y,Z\}$. For a single qubit measurement (which outputs a classical bit of information), the operation is followed by either no error or a bit-flip error; this is equivalent to having a single-qubit $X$ (or $Z$) error before a measurement in $Z$ (or $X$) basis.
	
	A \emph{fault} is specified by a location and a nontrivial 1- or 2-qubit Pauli operation which describes a deviation from the ideal operation on the location. This Pauli operation is called the ``Pauli error due to the fault''.
	\label{def:noise_model}
\end{definition}

A small number of faults during the measurements can lead to an error of weight higher than $\tau$ which may cause the EC protocol to fail. To see this, first, we describe how an error of weight 1 or 2 arising from a faulty operation can propagate through a circuit and become an error of higher weight. Specifically, a Hadamard gate and a CNOT gate will transform $X$-type and $Z$-type errors as follows: 
\begingroup
\setlength\arraycolsep{1pt}	
\begin{equation}
	\begin{matrix}
		H: \quad &X &\mapsto &Z, \quad &Z &\mapsto &X, \\
		\mathrm{CNOT}: \quad &XI &\mapsto &XX, \quad &ZI &\mapsto &ZI, \\
		&IX &\mapsto &IX, \quad &IZ &\mapsto &ZZ.
	\end{matrix} \nonumber
\end{equation}%
\endgroup 

To see how errors from a few faults can cause an EC protocol to fail, let us consider a circuit for measuring a stabilizer generator of the Steane code as an example. The \codepar{7,1,3} Steane code \cite{Steane96b} is a stabilizer code which can be described by the following generators:
\begingroup
\setlength\arraycolsep{1pt}	
\begin{equation}
	\begin{matrix}
		g^x_1: &I &I &I &X &X &X &X, & \quad & g^z_1: &I &I &I &Z &Z &Z &Z,\\
		g^x_2: &I &X &X &I &I &X &X, & \quad & g^z_2: &I &Z &Z &I &I &Z &Z,\\
		g^x_3: &X &I &X &I &X &I &X, & \quad & g^z_3: &Z &I &Z &I &Z &I &Z.
	\end{matrix}
\end{equation}%
\endgroup
Logical $X$ and logical $Z$ operators of the Steane code are $X^{\otimes 7} M$ and $Z^{\otimes 7} N$ for any stabilizers $M,N$. The syndrome is a 6-bit string of the form ($\vec{s}_x|\vec{s}_z$), with the $i$-th bit being 0 (or 1) if measuring the $i$-th generator (ordered as $g^x_{1}$, $g^x_{2}$, $g^x_{3}$, then $g^z_{1}$, $g^z_{2}$, $g^z_{3}$) gives $+1$ (or $-1$) eigenvalue.




Suppose that during the syndrome measurement, all circuits for measuring stabilizer generators are perfect except for a circuit for measuring $g^z_1$ which has at most 1 fault.
Consider a circuit for measuring $g^z_1$ and storing the syndrome using one ancilla qubit (called the \emph{syndrome ancilla}) as in \cref{subfig:nonflag_circuit}. Also, assume that at most one CNOT gate causes either $II,IZ,ZI,$ or $ZZ$ error. Because of error propagation, a $Z$ error occurred to the syndrome ancilla can propagate back to one or more data qubit(s). As a result, we find that possible errors on data qubits arising from at most 1 CNOT fault (up to multiplication of $g^z_1$) are,
\begin{equation}
	I,Z_4,Z_5,Z_6,Z_7,Z_6Z_7.
\end{equation}
A circuit fault may also cause the syndrome bit to flip. In order to obtain the syndrome exactly corresponding to the data error, one can perform full syndrome measurements until the outcomes are repeated two times in a row, then do the error correction using the repeated syndrome.
However, note that the Steane code which can correct any error up to weight 1 must be able to correct the following errors as well:
\begin{equation}
	I,Z_1,Z_2,Z_3,Z_4,Z_5,Z_6,Z_7.
\end{equation}

\begin{figure}[tbp]
	\centering
	\begin{subfigure}{0.23\textwidth}
		\includegraphics[width=\textwidth]{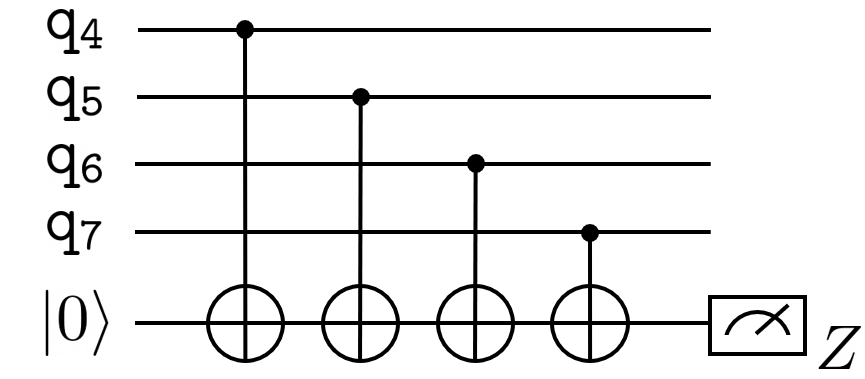}
		\captionsetup{justification=centering}
		\caption{}
		\label{subfig:nonflag_circuit}
	\end{subfigure}	
	\begin{subfigure}{0.29\textwidth}
		\includegraphics[width=\textwidth]{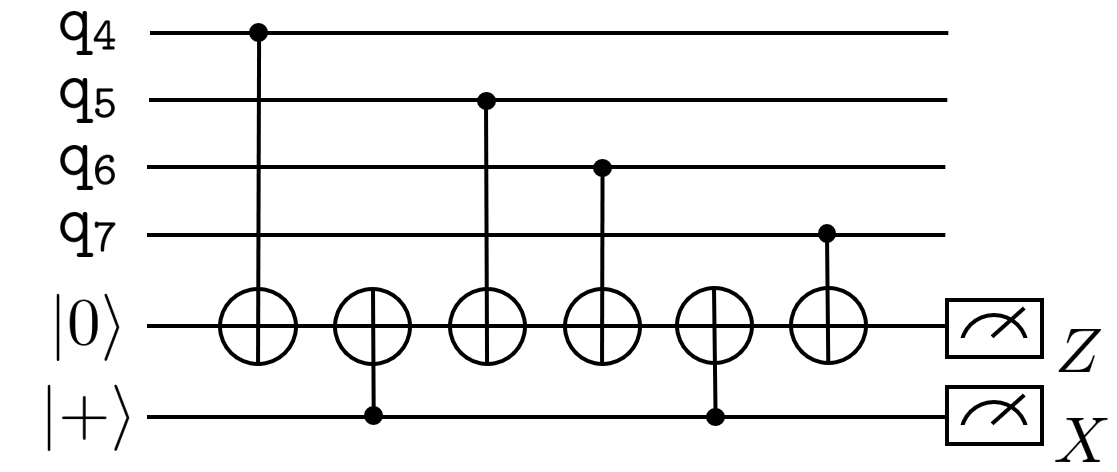}
		\captionsetup{justification=centering}
		\caption{}
		\label{subfig:flag_circuit}
	\end{subfigure}
	\begin{subfigure}{0.21\textwidth}
		\includegraphics[width=\textwidth]{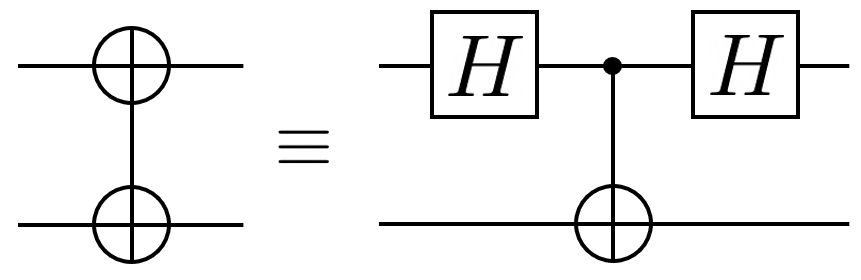}
		\captionsetup{justification=centering}
		\caption{}
		\label{subfig:XNOT}
	\end{subfigure}
	\caption{(a) An example of non-flag circuit for measuring generator $g_1^z$ of the \codepar{7,1,3} code. Only qubits on which the operator acts are displayed. The measurement result 0 and 1 obtained from the syndrome ancilla correspond to the $+1$ and $-1$ eigenvalues of $g_1^z$. (b) An example of flag circuit for measuring $g_1^z$. The state of the flag ancilla can flip from $\left|+\right\rangle$ to $\left|-\right\rangle$ if some fault occurs in between two flag CNOT gates. A circuit for measuring $X$-type generator can be obtained by replacing each CNOT gate with the gate shown in (c).}
	\label{fig:flag_n_nonflag}
\end{figure}

\noindent Errors $Z_1$ and $Z_6Z_7$ have the same syndrome $(0,0,1|0,0,0)$ but are not logically equivalent, and subsequent syndrome measurements cannot distinguish between these two cases. This means that if a CNOT fault leads to the $Z_6Z_7$ error, a correction step for the syndrome $(0,0,1|0,0,0)$ that applies $Z_1^\dagger$ to the data qubits will result in a logical error $Z_1Z_6Z_7$ on the data qubits, causing the EC protocol to fail.

The goal of this work is to design an EC protocol which is \emph{fault tolerant}; that is, we want to make sure that any subsequent error arising from a small number of faults will still be correctable by the protocol regardless of its weight (the formal definitions of fault tolerance will be discussed in \cref{subsec:FT_def}).


One way to solve the error distinguishing issue is to use traditional FTEC schemes such as the ones proposed by Shor \cite{Shor96,DA07}, Steane \cite{Steane97,Steane02}, or Knill \cite{Knill05a}. However, these schemes require a large number of ancillas. 
An alternative way to solve the problem is to add an additional ancilla qubit in a circuit for measuring $g^z_1$ as shown in \cref{subfig:flag_circuit}. A circuit of this form is called \emph{flag circuit} \cite{CR17a} (in contrast to the circuit in \cref{subfig:nonflag_circuit}, which is called \emph{non-flag circuit}). The additional ancilla qubit is called \emph{flag ancilla}, which is initially prepared in the state $|+\rangle$. There are two types of CNOT gates in a flag circuit: a \emph{data CNOT} which couples one of the data qubits and the syndrome ancilla, and a \emph{flag CNOT} which couples the flag ancilla and the syndrome ancilla. Whenever a data CNOT in between two flag CNOTs causes either $IZ$ or $ZZ$ error, a $Z$ error will propagate from the syndrome ancilla to the flag ancilla, causing the state of the flag ancilla to flip to $|-\rangle$. In general, a flag circuit may have more than one flag ancilla, and data and flag CNOTs may be arranged in a complicated way so that a certain number of faults can be caught by the flag ancillas.

By using the circuit in \cref{subfig:flag_circuit} for measuring $g^z_1$, we find that possible errors on the data qubits arising from at most 1 CNOT fault corresponding to each flag measurement outcome are,
\begin{equation}
	\begin{matrix*}[l]
		0: &I,Z_4,Z_5,Z_6,Z_7,\\
		1: &I,Z_4,Z_6Z_7,Z_7,
	\end{matrix*} \label{eq:flag_ex}
\end{equation}
where the outcome 0 and 1 correspond to $|+\rangle$ and $|-\rangle$ states, respectively. We can see that the flag measurement outcome is 1 whenever $Z_6Z_7$ occurs. In contrast, an input error $Z_1$ will not flip the state of the flag ancilla, so it always corresponds to the flag measurement outcome 0. Therefore, $Z_1$ and $Z_6Z_7$ can be distinguished using the flag measurement outcome, and an appropriate error correction for each case can be applied to correct such an error. The main advantage of the flag technique is that the number of ancillas required for the flag FTEC protocol is relatively small compared to that required for the traditional FTEC protocols (assuming that ancilla preparation and measurement are fast and the ancillas can be reused).

\subsection{Distinguishable fault set}
\label{subsec:fault_set}



For a general stabilizer code which can correct errors up to weight $\tau=\lfloor(d-1)/2\rfloor$, we would like to construct circuits for syndrome measurement in a way that all possible errors arising from up to $t$ faults (where $t \leq \tau$) can be corrected, and $t$ is as close to $\tau$ as possible. Note that these errors include any single-qubit errors and errors arising from any fault in any circuit involved in the syndrome measurement. For simplicity, this work will focus mainly on a stabilizer code in the Calderbank-Shor-Steane (CSS) code family \cite{CS96,Steane96b}, in which $X$-type and $Z$-type errors can be detected and corrected separately.  


For a given CSS code, a circuit for measuring $Z$-type generator will look similar to a circuit in \cref{subfig:nonflag_circuit} or \cref{subfig:flag_circuit}, except that there will be $w$ data CNOT gates for a $Z$-type generator of weight $w$. A circuit can have any number of flag ancillas (or have no flag ancillas). There are several factors that can determine the ability to distinguish possible errors; for example, the number of flag ancillas, the ordering of data and flag CNOT gates, and the choice of generators being used for the syndrome measurement \cite{CR17a}. A circuit for measuring $X$-type generator is similar to a circuit for measuring $Z$-type generator, except that each CNOT gate is replaced by the gate displayed in \cref{subfig:XNOT}.

For a given $t$, finding all possible combinations of faults up to $t$ faults can be laborious since there are many circuits involved in the syndrome measurement, and each circuit have many gates. To simplify our analysis, we will first consider the case that there is only one CNOT fault in one of the circuits for measuring $Z$-type generators (similar to \cref{subfig:nonflag_circuit} or \cref{subfig:flag_circuit}). Suppose that there are a total of $c$ flag ancillas involved in a single round of full syndrome measurement (counted from all circuits). We define a \emph{flag vector} $\in \mathbb{Z}^c_2$ to be a bitstring wherein each bit is the measurement outcome of each flag ancilla. There are two mathematical objects associated with each fault: a data error arising from the fault, and a flag vector corresponding to the fault.

Recall that a faulty CNOT gate can cause a two-qubit error of the form $P_1\otimes P_2$ where $P_1,P_2 \in \{I,X,Y,Z\}$. However, there are many cases of a single fault which are equivalent, meaning that they can give rise to the same data error and the same flag vector. We find that all possible cases in which a single fault can lead to a purely $Z$-type error on the data qubits can be obtained by considering only (1) the cases that a faulty CNOT gate in a circuit for measuring $Z$-type generator causes $IZ$ error, and (2) the cases that a $Z$ error occurs to any data or ancilla qubit.
%
This follows from the following facts \cite{TCL20}:
\begin{enumerate}
	\item The case that a faulty CNOT gate causes $ZZ$ error is equivalent to the case that the preceding CNOT gate causes error $IZ$ (while the case that the first CNOT gate in a circuit causes $ZZ$ error is equivalent to the case that a $Z$ error occurs to an ancilla qubit).
	\item The case that a faulty CNOT gate causes $XZ,YZ$ error is equivalent to the case that an $X$ error occurs to a data qubit and a faulty CNOT gate causes $IZ$ or $ZZ$ error.
	\item The case that a faulty CNOT gate causes $XI,YI,ZI,IX,XX,YX$ or $ZX$ error can be considered as the case that a single-qubit error occurs to a data qubit since an $X$ error occurred to the syndrome ancilla will not propagate back to any data qubit.
	\item The case that a faulty CNOT gate causes $IY,XY,YY$ or $ZY$ error is similar to the case that a faulty CNOT gate causes $IZ,XZ,YZ$ or $ZZ$ error,
	\item An ancilla preparation or measurement fault can be considered as the case that $X$ or $Z$ error occurred to an ancilla qubit (either syndrome or flag ancilla).
	\item A CSS code can detect and correct $X$-type and $Z$-type errors separately, and a single fault in a circuit for measuring $X$-type generator cannot cause an $Z$-type error of weight greater than 1 (and vice versa).
\end{enumerate}
Moreover, if $X$-type and $Z$-type generators have similar forms and the gate permutations in the measuring circuits are the same, then all possible faults that can lead to $X$-type errors on the data qubits are of similar form.

If there are many faults during the protocol, the data errors and the flag vectors caused by each fault can be combined \cite{TL20}. In particular, a fault combination can be defined as follows:
\begin{definition}{Fault combination}
	
	A \emph{fault combination} $\Lambda =\{\lambda_{1},\lambda_{2},\dots,\lambda_{r}\}$ is a set of $r$ faults $\lambda_{1}, \lambda_{2}, \cdots, \lambda_{r}$. 
	Suppose that the Pauli error due to the fault $\lambda_{i}$ can propagate through the circuit and lead to \emph{data error} $E_{i}$ and \emph{flag vector} $\vec{f}_{i}$. The \emph{combined data error} $\mathbf{E}$ and \emph{cumulative flag vector} $\vec{\mathbf{f}}$ corresponding to $\Lambda$ are defined as follows:
	\begin{align}
		\mathbf{E}&=\prod_{i=1}^r E_{i}, \label{eq:combined_E}\\
		\vec{\mathbf{f}}&=\sum_{i=1}^r \vec{f}_{i}\;(\mathrm{mod}\;2). \label{eq:cumulative_f}
	\end{align}
	\label{def:fault_combi}%
\end{definition}
\noindent Note that the error syndrome of the combined data error is $\vec{s}(\mathbf{E})=\sum_{i=1}^r \vec{s}(E_{i})\;(\mathrm{mod}\;2)$. For example, suppose that a fault combination $\Lambda$ arises from two faults $\lambda_1$ and $\lambda_2$ which can lead to data errors $E_1$ and $E_2$, and cumulative flag vectors $\vec{f}_1$ and $\vec{f}_2$. Then, the combined data error $\mathbf{E}$ and the cumulative flag vector $\vec{\mathbf{f}}$ of $\Lambda$ are $\mathbf{E}=E_1 \cdot E_2$ and $\vec{\mathbf{f}}=\vec{f}_1+\vec{f}_2\;(\mathrm{mod}\;2)$.



When faults occur in an actual protocol, the faulty locations and the combined data error are not known. In order to determine the combined data error so that the error correction can be done, we will try to measure the error syndrome of the combined data error, and calculate the cumulative flag vector from the flag measurement results obtained since the beginning of the protocol. These measurements, in turn, are subject to errors. The full syndrome measurements will be performed until the syndromes and the cumulative flag vectors are repeated for a certain number of times (similar to the Shor FTEC scheme); the full details of the protocol will be described in \cref{subsec:FTEC_ana}. 
(Note that by defining the cumulative flag vector as a sum of flag vectors, we lose the information of the ordering in which each fault occurs. However, we find that fault-tolerant protocols presented in this work can still be constructed without such information.)

As previously explained, error correction can fail if there are different faults that lead to non-equivalent errors but there is no way to distinguish them using their error syndromes or flag measurement results. 
To avoid this, all possible fault combinations must satisfy some conditions so that they can be distinguished. In particular, for a given set of circuits for measuring stabilizer generators, all possible fault combinations can be found, and their corresponding combined data error and cumulative flag vector can be calculated. Let the fault set $\mathcal{F}_t$ be the set of all possible fault combinations arising from up to $t$ faults. We will be able to distinguish all fault combinations if the fault set satisfies the conditions in the following definition:

\begin{definition}{Distinguishable fault set}
	
	Let the \emph{fault set} $\mathcal{F}_t$ denote the set of all possible fault combinations arising from up to $t$ faults and let $S$ be the stabilizer group of the quantum error correcting code used to encode the data. We say that $\mathcal{F}_t$ is \emph{distinguishable} if for any pair of fault combinations $\Lambda_p,\Lambda_q \in \mathcal{F}_t$, at least one of the following conditions is satisfied:
	\begin{enumerate}
		\item $\vec{s}(\mathbf{E}_p) \neq \vec{s}(\mathbf{E}_q)$, or
		\item $\vec{\mathbf{f}}_p \neq \vec{\mathbf{f}}_q$, or
		\item $\mathbf{E}_p = \mathbf{E}_q\cdot M$ for some stabilizer $M \in S$,
	\end{enumerate}
	where $\mathbf{E}_p,\vec{\mathbf{f}}_p$ correspond to $\Lambda_p$, and $\mathbf{E}_q,\vec{\mathbf{f}}_q$ correspond to $\Lambda_q$.
	Otherwise, we say that $\mathcal{F}_t$ is \emph{indistinguishable}.
	\label{def:distinguishable}%
\end{definition}

An example of a distinguishable fault set with $t=1$ is the fault set corresponding to \cref{eq:flag_ex} (assuming that a fault occurs in a circuit for measuring $g_1^z$ only). In that case, we can see that for any pair of faults, either the syndromes of the data errors or the flag measurement outcomes are different. 


The following proposition states the relationship between `correctable' and `detectable' faults. This is similar to the fact that a stabilizer code of distance $d$ can detect errors up to weight $d-1$ and can correct errors up to weight $\tau=\lfloor(d-1)/2\rfloor$ \cite{Gottesman97}.


\begin{proposition}
	$\mathcal{F}_t$ is distinguishable if and only if a fault combination corresponding to a nontrivial logical operator and the zero cumulative flag vector is not in $\mathcal{F}_{2t}$.
	\label{prop:2t}%
\end{proposition}
\begin{proof}
	($\Rightarrow$) Let $\Lambda_p,\Lambda_q \in \mathcal{F}_t$ be fault combinations arising from up to $t$ faults, let $\tilde{\Lambda}_r \in \mathcal{F}_{2t}$ be a fault combination arising from up to $2t$ faults, and let $S$ be the stabilizer group. First, observe that for any $\tilde{\Lambda}_r \in \mathcal{F}_{2t}$, there exist $\Lambda_p,\Lambda_q \in \mathcal{F}_t$ such that $\tilde{\Lambda}_r = \Lambda_p \cup \Lambda_q$ (where the union of two fault combinations is similar to the union of two sets). Now suppose that $\mathcal{F}_t$ is distinguishable. Then, for each pair of $\Lambda_p,\Lambda_q$ in $\mathcal{F}_t$, $\vec{s}(\mathbf{E}_p)\neq \vec{s}(\mathbf{E}_q)$ or $\vec{\mathbf{f}}_p \neq \vec{\mathbf{f}}_q$ or $\mathbf{E}_p = \mathbf{E}_q\cdot M$ for some stabilizer $M \in S$. 
	We find that $\tilde{\Lambda}_r=\Lambda_p \cup \Lambda_q$ corresponds to $\mathbf{E}_r$ and $\vec{\mathbf{f}}_r$ such that $\vec{s}(\mathbf{E}_r)=\vec{s}(\mathbf{E}_p)+\vec{s}(\mathbf{E}_q)\neq 0$ or  $\vec{\mathbf{f}}_r = \vec{\mathbf{f}}_p+\vec{\mathbf{f}}_q\neq 0$ or $\mathbf{E}_r = \mathbf{E}_p \cdot \mathbf{E}_q = M$ for some stabilizer $M \in S$. This is true for any $\tilde{\Lambda}_r \in \mathcal{F}_{2t}$, meaning that there is no fault combination in $\mathcal{F}_{2t}$ which corresponds to a nontrivial logical operator and the zero cumulative flag vector.
	
	($\Leftarrow$) As before, we know that for any $\tilde{\Lambda}_r \in \mathcal{F}_{2t}$, there exist $\Lambda_p,\Lambda_q \in \mathcal{F}_t$ such that $\tilde{\Lambda}_r = \Lambda_p \cup \Lambda_q$. Now suppose that $\mathcal{F}_t$ is indistinguishable. Then, there are some pair of $\Lambda_p,\Lambda_q$ in $\mathcal{F}_t$ such that $\vec{s}(\mathbf{E}_p)= \vec{s}(\mathbf{E}_q)$, $\vec{\mathbf{f}}_p = \vec{\mathbf{f}}_q$, and $\mathbf{E}_p \cdot \mathbf{E}_q$ is not a stabilizer in $S$. For such pair, we find that $\tilde{\Lambda}_r=\Lambda_p \cup \Lambda_q$ corresponds to $\mathbf{E}_r$ and $\vec{\mathbf{f}}_r$ such that $\vec{s}(\mathbf{E}_r)=\vec{s}(\mathbf{E}_p)+\vec{s}(\mathbf{E}_q)= 0$, $\vec{\mathbf{f}}_r = \vec{\mathbf{f}}_p+\vec{\mathbf{f}}_q = 0$, and $\mathbf{E}_r = \mathbf{E}_p \cdot \mathbf{E}_q$ is not a stabilizer in $S$. Therefore, there is a fault combination corresponding to a nontrivial logical operator and the zero cumulative flag vector in $\mathcal{F}_{2t}$. 
\end{proof}



Finding a circuit configuration which gives a distinguishable fault set is one of the main goals of this work. We claim that for a given set of circuits for measuring generators of a stabilizer code, if the fault set is distinguishable, an FTEC protocol for such a code can be constructed. However, we will defer the proof of this claim until \cref{subsec:FTEC_ana}.




\subsection{Finding equivalent errors using error weight parities}
\label{subsec:WPEC_ana}

One goal of this work is to find a good combination of stabilizer code and a set of circuits for measuring the code generators in which the corresponding fault set is distinguishable. As we see in \cref{def:distinguishable}, whether each pair of fault combinations can be distinguished depends on the syndrome of the combined data error and the cumulative flag vector corresponding to each fault combination, and these features heavily depend on the structure of the circuits. However, we should note that there is no need to distinguish a pair of fault combinations whose combined data errors are logically equivalent. Therefore, if the circuits for a particular code are designed in a way that large portions of fault combinations can give equivalent errors, the fault set arising from the circuits will be more likely distinguishable.

For a general stabilizer code, it is not obvious to see whether two Pauli errors with the same syndrome are logically equivalent or off by a multiplication of some nontrivial logical operator. Fortunately, for some CSS codes, it is possible to check whether two Pauli errors with the same syndrome are logically equivalent by comparing their error weight parities, defined as follows:
\begin{definition}
	The \emph{weight parity} of Pauli error $E$, denoted by $\mathrm{wp}(E)$, is 0 if $E$ has even weight, or is 1 if $E$ has odd weight.
\end{definition}

In \cite{TL20}, we prove that for the \codepar{7,1,3} Steane code and the \codepar{23,1,7} Golay code, errors with the same syndrome and weight parity are logically equivalent. In this work, the idea is further extended to a family of \codepar{n,k,d} CSS codes in which $n$ is odd, $k$ is 1, all stabilizer generators have even weight, and $X^{\otimes n}$ and $Z^{\otimes n}$ are logical $X$ and logical $Z$ operators, respectively. The lemma (adapted from Claim 1 in \cite{TL20}) is as follows:

\begin{lemma}
	Let $C$ be an \codepar{n,k,d} CSS code in which $n$ is odd, $k=1$, all stabilizer generators have even weight, and $X^{\otimes n}$ and $Z^{\otimes n}$ are logical $X$ and logical $Z$ operators. Also, let $S_x,S_z$ be subgroups generated by $X$-type and $Z$-type generators of $C$, respectively. Suppose $E_1,E_2$ are Pauli errors of any weights with the same syndrome. 
	\begin{enumerate}
		\item If $E_1,E_2$ are $Z$-type errors, then $E_1,E_2$ have the same weight parity if and only if $E_1 = E_2 \cdot M$ for some $M \in S_z$.
		\item If $E_1,E_2$ are $X$-type errors, then $E_1,E_2$ have the same weight parity if and only if $E_1 = E_2 \cdot M$ for some $M \in S_x$.
	\end{enumerate}
	\label{lem:err_equivalence}%
\end{lemma}

\begin{proof}
	
	We focus on the first case when $E_1,E_2$ are $Z$-type errors and omit the similar proof for the second case. First, recall that the normalizer group of the stabilizer group (the subgroup of Pauli operators that commute with all stabilizers) is generated by the stabilizer generators together with the logical $X$ and the logical $Z$. Since $E_1,E_2$ have the same syndrome, their product $N = E_1 E_2$ has trivial syndrome, and is thus in the normalizer group. So we can express $N$ as a product of the stabilizer generators and the logical $X$ and $Z$'s.  
	But there is no $X$-type factors (since $N$ is $Z$-type). Therefore, 
	$N = M (Z^{\otimes n})^a$ where $M \in S_z$ and $a \in \{0,1\}$. 
	
	Next, we make an observation.  Let $M_1, M_2$ be two $Z$-type operators, with respective weights $w_1,w_2$.  
	The weight of the product $M_1 M_2$ is $w_1+w_2-2c$, where $c$ is the number of qubits supported on both $M_1$ and $M_2$.
	From this observation, and the fact that all generators have even weight, we know $M$ has even weight.   
	Also, from the same observation, and the hypothesis that $E_1, E_2$ have the same weight parity, $N$ also has even weight. 
	If $a=1$, $N = M (Z^{\otimes n})^a$ will contradict the observation, so, $a=0$, $N=M$, and  
	$E_1 E_2 = M \in S_z$ as claimed. On the other hand, if we assume that $E_1, E_2$ have different weight parities, then $N$ has odd weight and $a=1$, which implies that $E_1E_2 = M (Z^{\otimes n})$ for some $M \in S_z$.
\end{proof}

\cref{lem:err_equivalence} provides a possible way to perform error correction using syndromes and weight parities, and it can help us find a good code and circuits in which the fault set is distinguishable. In particular, for a given CSS code satisfying \cref{lem:err_equivalence}, if the error syndrome and the weight parity of the data error can be measured perfectly, then an EC operator which can map the erroneous codeword back to the original codeword can be determined without failure. The EC operator can be any Pauli operator that has the same syndrome and the same weight parity as those of the data error. For example, if the \codepar{7,1,3} Steane code is being used and the data error is $Z_1Z_3Z_6Z_7$, we can use $Z_1 Z_2$ as an EC operator to do the error correction.

However, measuring the weight parity should not be done directly on the codeword; measuring weight parities of $Z$-type and $X$-type errors correspond to measuring $X^{\otimes n}$ and $Z^{\otimes n}$, respectively, which may destroy the superposition of the encoded state. Moreover, $X^{\otimes n}$ and $Z^{\otimes n}$ do not commute. Fortunately, if we have two codes $C_1,C_2$ such that $C_1$ is a subcode of $C_2$, then the weight parity of an error on $C_1$ can sometimes be determined by the measurement results of the generators of $C_2$.


In \cite{TL20} in which an FTEC protocol for a \codepar{49,1,9} concatenated Steane code is developed, we consider the case that $C_1$ is the \codepar{7,1,3} Steane code and $C_2$ is the \codepar{49,1,9} concatenated code. The error weight parities for each subblock of the 7-qubit code are determined by the syndrome obtained from the measurement of the \codepar{49,1,9} code generators. Afterwards, error correction is performed blockwisely using the weight parity of the error in each subblock, together with the syndrome obtained from the measurement of the 7-qubit code generators for such a subblock. We also find some evidences suggesting that a similar error correction technique may be applicable to other concatenated codes such as the concatenated Golay code and a concatenated Steane code with more than 2 levels of concatenation.



In this work, we will use a different approach; we will consider a case that $C_2$ is not constructed from concatenating $C_1$'s. In \cref{sec:3D_code}, we will consider the 3D color code of distance 3 in the form that has a 2D color code of distance 3 as a subcode, and we will try to construct circuits for measuring its generators which give a distinguishable fault set. We will extend the construction ideas to families of capped and recursive color codes in \cref{sec:CCC}. Fault-tolerant protocols for the code and circuits which gives a distinguishable fault set will be discussed in \cref{sec:FT_protocol}.

\section{Syndrome measurement circuits for the 3D color code of distance 3}
\label{sec:3D_code}

In this section, we will try to find circuits for measuring generators of the 3D color code of distance 3 which gives a distinguishable fault set. We will first define a 3D color code of distance 3 as a CSS subsystem code and observe some of its properties which is useful for fault tolerant quantum computation. Afterwards, we will give the CNOT orderings for the circuits which can make the fault set become distinguishable. 

\subsection{The 3D color code of distance 3}
\label{subsec:3D_code_def}




First, let us consider the qubit arrangement as displayed in \cref{fig:3D_code}\hyperlink{target:3D}{a}. A 3D color code of distance 3 \cite{Bombin15} is a \codepar{15,1,3} CSS subsystem code \cite{Poulin05,Bacon06} which can be described by the stabilizer group $S_\mathrm{3D} = \langle v_i^x, v_i^z\rangle$ and the gauge group $G_\mathrm{3D} = \langle v_i^x,v_i^z,f_j^x,f_j^z\rangle$, $i=0,1,2,3$ and $j=1,2,...,6$, where $v_i^x$'s and $f_j^x$'s (or $v_i^z$'s and $f_j^z$'s) are $X$-type (or $Z$-type) operators defined on the following set of qubits: 
\begin{itemize}
	\item $v_0^x$ (or $v_0^z$) is defined on $\mathtt{q_0},\mathtt{q_1},\mathtt{q_2},\mathtt{q_3},\mathtt{q_4},\mathtt{q_5},\mathtt{q_6},\mathtt{q_7}$
	\item $v_1^x$ (or $v_1^z$) is defined on $\mathtt{q_1},\mathtt{q_2},\mathtt{q_3},\mathtt{q_5},\mathtt{q_8},\mathtt{q_9},\mathtt{q_{10}},\mathtt{q_{12}}$
	\item $v_2^x$ (or $v_2^z$) is defined on $\mathtt{q_1},\mathtt{q_3},\mathtt{q_4},\mathtt{q_6},\mathtt{q_8},\mathtt{q_{10}},\mathtt{q_{11}},\mathtt{q_{13}}$
	\item $v_3^x$ (or $v_3^z$) is defined on $\mathtt{q_1},\mathtt{q_2},\mathtt{q_4},\mathtt{q_7},\mathtt{q_8},\mathtt{q_9},\mathtt{q_{11}},\mathtt{q_{14}}$
	\item $f_1^x$ (or $f_4^z$) is defined on $\mathtt{q_1},\mathtt{q_2},\mathtt{q_3},\mathtt{q_5}$
	\item $f_2^x$ (or $f_5^z$) is defined on $\mathtt{q_1},\mathtt{q_3},\mathtt{q_4},\mathtt{q_6}$
	\item $f_3^x$ (or $f_6^z$) is defined on $\mathtt{q_1},\mathtt{q_2},\mathtt{q_4},\mathtt{q_7}$
	\item $f_4^x$ (or $f_1^z$) is defined on $\mathtt{q_1},\mathtt{q_4},\mathtt{q_8},\mathtt{q_{11}}$
	\item $f_5^x$ (or $f_2^z$) is defined on $\mathtt{q_1},\mathtt{q_2},\mathtt{q_8},\mathtt{q_9}$
	\item $f_6^x$ (or $f_3^z$) is defined on $\mathtt{q_1},\mathtt{q_3},\mathtt{q_8},\mathtt{q_{10}}$
\end{itemize}
where qubit $i$ in \cref{fig:3D_code}\hyperlink{target:3D}{a} is denoted by $\mathtt{q}_i$. Graphically, $v_i^x$'s and $v_i^z$'s are 8-body volumes shown in \cref{fig:3D_code}\hyperlink{target:3D}{b}, and $f_j^x$'s and $f_j^z$'s are 4-body faces shown in \cref{fig:3D_code}\hyperlink{target:3D}{c}. Note that $f_j^x$ and $f_k^z$ anticommute when $j=k$, and they commute when $j \neq k$. The dual lattice of the 3D color code of distance 3 is illustrated in \cref{fig:3D_code}\hyperlink{target:3D}{d}, where each vertex represents each stabilizer generator.

\begin{figure}[tbp]
	\centering
	\hypertarget{target:3D}{}
	\includegraphics[width=0.28\textwidth]{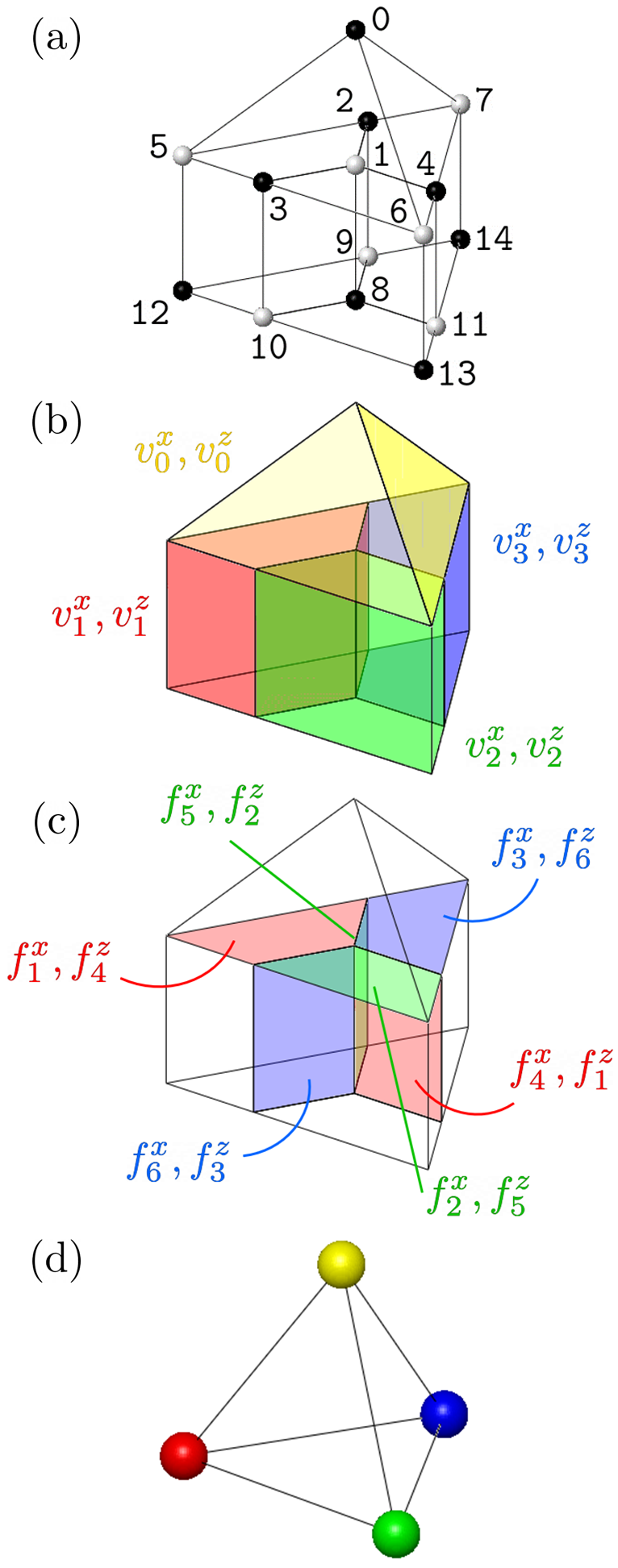}
	\caption{The 3D color code of distance 3. In (a), qubits are represented by vertices. Note that the set of qubits are bipartite, as displayed by black and white colors. Stabilizer generators and gauge generators of the code are illustrated by volume operators in (b) and face operators in (c), respectively. The dual lattice of the code is shown in (d).}
	\label{fig:3D_code}
\end{figure}

The 3D color code of distance 3 can be viewed as the \codepar{15,7,3} Hamming code in which 6 out of 7 logical qubits become gauge qubits. From the subsystem code previously described, a \codepar{15,1,3} stabilizer code can be constructed by fixing some gauge qubits; i.e., choosing some gauge operators which commute with one another and including them in the stabilizer group. In this work, we will discuss two possible ways to construct a stabilizer code from the 3D color code of distance 3. The resulting codes will be called the 3D color code in H form and the 3D color code in T form.\\





\noindent\textbf{The 3D color code of distance 3 in H form}

Let us consider the center plane of the code shown in \cref{fig:3D_code}\hyperlink{target:3D}{a} which covers $\mathtt{q_1}$ to $\mathtt{q_7}$. We can see that the plane looks exactly like the 2D color code of distance 3 \cite{BM06}, whose stabilizer group is $S_\mathrm{2D}=\langle f_1^x,f_2^x,f_3^x,f_4^z,f_5^z,f_6^z\rangle$ (the 2D color code of distance 3 is equivalent to the \codepar{7,1,3} Steane code). The 3D color code in H form is constructed by adding the stabilizer generators of the 2D color code to the old generating set of the 3D color code; the stabilizer group of the 3D color code of distance 3 in H form is
\begin{align}
	S_{\mathrm{H}}=\langle &v_0^x,v_1^x,v_2^x,v_3^x,f_1^x,f_2^x,f_3^x,\nonumber\\ 
	&v_0^z,v_1^z,v_2^z,v_3^z,f_4^z,f_5^z,f_6^z\rangle. \label{eq:S_H}
\end{align}
We can choose logical $X$ and logical $Z$ operators of this code to be $X^{\otimes n}M$ and $Z^{\otimes n}N$ for some stabilizers $M,N \in S_{\mathrm{H}}$. One important property of the code in H form for fault-tolerant quantum computation is that the logical Hadamard, $S$, and CNOT gates are transversal; i.e., $\bar{H}=H^{\otimes n}$ is a logical Hadamard gate, $\bar{S} = {(S^\dagger)^{\otimes n}}$ is a logical $S$ gate, and  $\overline{\mathrm{CNOT}}=\mathrm{CNOT}^{\otimes n}$ is a logical CNOT gate, where $H = \frac{1}{\sqrt{2}}\bigl( \begin{smallmatrix}1 & 1\\ 1 & -1\end{smallmatrix}\bigr)$ and $S = \bigl( \begin{smallmatrix}1 & 0\\ 0 & i\end{smallmatrix}\bigr)$. 

Note that the choice of stabilizer generators for $S_{\mathrm{H}}$ is not unique. However, the choice of generators determines how the error syndrome will be measured, and different choices of generators can give different fault sets. The circuits for measuring generators discussed later in \cref{subsec:3D_code_config} only correspond to the choice of generators in \cref{eq:S_H}.
\\

\noindent\textbf{The 3D color code of distance 3 in T form}

Compared to the code in H form, the 3D color code of distance 3 in T form is constructed from different gauge operators of the \codepar{15,1,3} subsystem code. In particular, the generators of the code in T form consist of the generators of the \codepar{15,1,3} subsystem code and all $Z$-type 4-body face generators; i.e., the stabilizer group of the code in T form is
\begin{align}
	S_{\mathrm{T}}=\langle &v_0^x,v_1^x,v_2^x,v_3^x,f_1^z,f_2^z,f_3^z,\nonumber\\ 
	&v_0^z,v_1^z,v_2^z,v_3^z,f_4^z,f_5^z,f_6^z\rangle.
\end{align}

Similar to the code in H form, we can choose logical $X$ and logical $Z$ operators of this code to be $X^{\otimes n}M$ and $Z^{\otimes n}N$ for some stabilizers $M,N \in S_{\mathrm{T}}$. Also, CNOT gate is transversal in the code of T form. However, one major difference from the code in H form is that Hadamard and $S$ gates are not transversal in this code. Instead, a $T$ gate is transversal; a logical $T$ gate can be implemented by applying $T$ gates on all qubits represented by black vertices in \cref{fig:3D_code}\hyperlink{target:3D}{a} and applying $T^\dagger$ gates on all qubits represented by white vertices, where $T = \bigl( \begin{smallmatrix}1 & 0\\ 0 & \sqrt{i} \end{smallmatrix}\bigr)$.

In fact, the code in T form is equivalent to the \codepar{15,1,3} quantum Reed-Muller code. Note that \cref{lem:err_equivalence} is applicable to both codes in H form and T form since they have all code properties required by the lemma, even though $X$-type and $Z$-type generators are not similar in the case of the code in T form.
\\

\noindent\textbf{Code switching}


It is possible to transform between the code in H form and the code in T form using the technique called \emph{code switching} \cite{PR13,ADP14,Bombin15,KB15}. The process involves measurements of gauge operators of the \codepar{15,1,3} subsystem code, which can be done as follows: Suppose that we start from the code in H form. We can switch to the code in T form by first measuring $f_1^z,f_2^z$ and $f_3^z$. Afterwards, we must apply an $X$-type Pauli operator that 
\begin{enumerate}
	\item commutes with all $v_i^x$'s and $v_i^z$'s ($i=0,1,2,3$), and
	\item commutes with $f_4^z,f_5^z,f_6^z$, and
	\item for each $j=1,2,3$, commutes with $f_j^z$ if the outcome from measuring such an operator is 0 (the eigenvalue is $+1$) or anticommutes with $f_j^z$ if the outcome is 1 (the eigenvalue is $-1$).
\end{enumerate}
Switching from the code in T form to the code in H form can be done similarly, except that $f_1^x,f_2^x$ and $f_3^x$ will be measured and the operator to be applied must be a $Z$-type Pauli operator that commutes or anticommutes with $f_1^x,f_2^x$ and $f_3^x$ (depending on the measurement outcomes).


Transversal gates satisfy the conditions for fault-tolerant gate gadgets proposed in \cite{AGP06} (see \cref{subsec:FT_def}), thus they are very useful for fault-tolerant quantum computation. It is known that universal quantum computation can be performed using only $H,S,$ CNOT, and $T$ gates \cite{CRSS97,Gottesman98b,NRS01}. However, for any QECC, universal quantum computation cannot be achieved using only transversal gates due to the Eastin-Knill theorem \cite{EK09}. Fortunately, the code switching technique allows us to perform universal quantum computation using both codes in H form and T form; any logical Clifford gate can be performed transversally on the code in H form since the Clifford group can be generated by $\{H,S,\mathrm{CNOT}\}$, and a logical $T$ gate can be performed transversally on the code in T form.
For the 3D color code of distance 3, code switching can be done fault-tolerantly using the above method \cite{Bombin15,KB15} or a method presented in \cite{BKS21} which involves a logical Einstein-Podolsky-Rosen (EPR) state.


\subsection{Circuit configuration for the 3D color code of distance 3}
\label{subsec:3D_code_config}



In this section, circuits for measuring the generators of the 3D color code of distance 3 in H form will be developed. Here we will try to find CNOT orderings for the circuits which make fault set $\mathcal{F}_1$ distinguishable (where $\mathcal{F}_1$ is the set of all fault combinations arising from up to 1 fault as defined in \cref{def:distinguishable}). The ideas used for the circuit construction in this section will be later adapted to the circuits for measuring generators of a capped or a recursive color code (capped and recursive capped codes will be defined in \cref{subsec:CCC_def,subsec:RCCC_def}, and the circuit construction will be discussed in \cref{subsec:CCC_config}). Fault-tolerant protocols for the 3D color code of distance 3 is similar to fault-tolerant protocols for capped color codes, which will be later discussed in \cref{sec:FT_protocol}.


For simplicity, since $X$-type and $Z$-type data errors can be corrected separately and $X$-type and $Z$-type generators of our choice have the same form, we will only discuss the case that a single fault can give rise to a $Z$-type data error. Similar analysis will also be applicable to the case of $X$-type errors. We start by observing that the 2D color code of distance 3 is a subcode of the the 3D color code of distance 3 in H form, where the 2D color code lies on the center plane of the code illustrated in \cref{fig:3D_code}\hyperlink{target:3D}{a}. The 2D color code is a code to which \cref{lem:err_equivalence} is applicable, meaning that if we can measure the syndrome and the weight parity of any $Z$-type Pauli error that occurred on the center plane, we can always find a Pauli operator logically equivalent to such an error. Moreover, we can see that the generator $v_0^x$ has support on all qubits on the center plane ($\mathtt{q_1}$ to $\mathtt{q_7}$). This means that the weight parity of a $Z$-type error on the center plane can be obtained by measuring $v_0^x$. For these reasons, we can always find an error correction operator for any $Z$-type error that occurred on the center plane using the measurement outcomes of $f_1^x,f_2^x,f_3^x$ (which give the syndrome of the error evaluated on the 2D color code) and the measurement outcome of $v_0^x$ (which gives the weight parity of the error).





All circuits for measuring generators of the 3D color code in H form used in this section are non-flag circuits. Each circuit has $w$ data CNOTs where $w$ is the weight of the operator being measured. The circuit for each generator looks similar to the circuit in \cref{fig:circuit_3D}, but the ordering of data CNOTs has yet to be determined. 

\begin{figure}[tbp]
	\centering
	\includegraphics[width=0.28\textwidth]{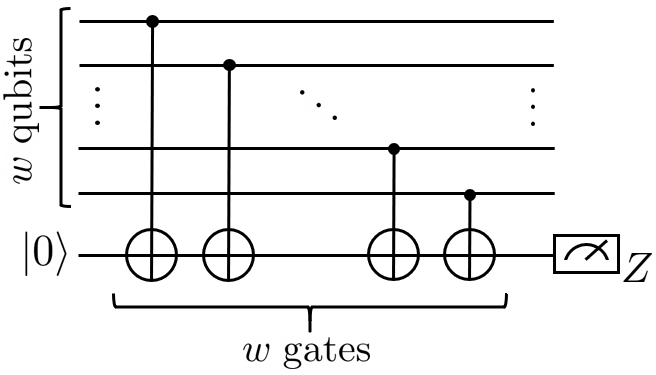}
	\caption{A non-flag circuit for measuring a $Z$-type generator of weight $w$ for the 3D color code of distance 3. The ordering of the CNOT gates for each generator has yet to be determined.}
	\label{fig:circuit_3D}
\end{figure}

Our goal is to find CNOT orderings for all circuits involved in the syndrome measurement so that $\mathcal{F}_1$ is distinguishable. Thus, we have to consider all possible errors arising from a single fault, not only the errors occurred on the center plane. Let us first consider an arbitrary single fault which can lead to a purely $Z$-type error. Since the 3D color code in H form has distance 3, all $Z$-type errors of weight 1 correspond to different syndromes. All we have to worry about are single faults which can lead to a $Z$-type error of weight $>1$ that has the same syndrome as some error of weight 1 but is not logically equivalent to such an error. Note that a $Z$-type error of weight $>1$ arising from a single fault can only be caused by a faulty CNOT gate in some circuit for measuring a $Z$-type generator. 


We can divide the generators of the 3D color code in H form into 3 categories: 
\begin{enumerate}
	\item $\mathtt{cap}$ generators, consisting of $v_0^x$ and $v_0^z$,
	\item $\mathtt{f}$ generators, consisting of $f_1^x,f_2^x,f_3^x,f_4^z,f_5^z,f_6^z$,
	\item $\mathtt{v}$ generators, consisting of $v_1^x,v_2^x,v_3^x,v_1^z,v_2^z,v_3^z$.
\end{enumerate}
($v_0^x$ and $v_0^z$ are considered separately from other $\mathtt{v}$ generators because they cover all qubits on the center plane.) Here we will analyze the pattern of $Z$-type errors arising from the measurement of $Z$-type generators of each category. The syndrome of each $Z$-type error will be represented in the form $(u,\vec{v},\vec{w})$, where $u,\vec{v},\vec{w}$ are syndromes obtained from the measurement of $\mathtt{cap}$, $\mathtt{f}$, and $\mathtt{v}$ generators of $X$ type, respectively. Note that for each $\mathtt{v}$ generator, there will be only one $\mathtt{f}$ generator such that the set of supporting qubits of the $\mathtt{v}$ generator contains all supporting qubits of the $\mathtt{f}$ generator (for example, $v_1^x$ and $f_1^x$, or $v_1^z$ and $f_4^z$).



Let us start by observing the syndromes of any $Z$-type error of weight 1. An error on the following qubits gives the syndrome of the following form:
\begin{itemize}
	\item an error on $\mathtt{q_0}$ gives syndrome $(1,\vec{0},\vec{0})$,
	\item an error on $\mathtt{q}_i$ ($i=1,\dots,7$) gives syndrome of the form $(1,\vec{q}_i,\vec{q}_i)$,
	\item an error on $\mathtt{q}_{\mathtt{7}+i}$ ($i=1,\dots,7$) gives syndrome of the form $(0,\vec{0},\vec{q}_i)$,
\end{itemize}
where $\vec{q}_i \in \mathbb{Z}_2^3$ is not zero (see \cref{tab:err_list_d3} as an example). We can see that all $Z$-type errors of weight 1 give different syndromes as expected. Next, let us consider a $Z$-type error $E$ of any weight which occurs only on the center plane. Suppose that the weight parity of $E$ is $\mathrm{wp}$ ($\mathrm{wp}$ is 0 or 1), and the syndrome of $E$ obtained from measuring $f_1^x,f_2^x,f_3^x$ is $\vec{p}$. Then, the syndrome of $E$ obtained from measuring all $X$-type generators is as follows:
\begin{itemize}
	\item an error $E$ on the center plane gives syndrome of the form $(\mathrm{wp},\vec{p},\vec{p})$.
\end{itemize}

\noindent We find that:
\begin{enumerate}
	\item $E$ and the error on $\mathtt{q_0}$ will have the same syndrome if $E$ has odd weight and $\vec{p}$ is trivial, which means that $E$ is equivalent to $Z^{\otimes 7}$ on the center plane. In this case, $E$ and $Z_0$ are logically equivalent up to a multiplication of $v_0^z$ and some stabilizer.
	\item $E$ and an error on $\mathtt{q}_i$ ($i=1,2,\dots,7$) will have the same syndrome if $E$ has odd weight and $\vec{p} = \vec{q}_i$ for some $i$. In this case, $E$ and $Z_i$ have the same weight parity and the same syndrome (evaluated by the generators of the 2D color code), meaning that $E$ and $Z_i$ are logically equivalent by \cref{lem:err_equivalence}.
	\item $E$ and an error on $\mathtt{q}_i$ ($i=7,8,\dots,14$) cannot have the same syndrome since $\vec{q}_i \neq \vec{0}$.
\end{enumerate}
Therefore, a $Z$-type error of any weight occurred only on the center plane either has syndrome different from those of $Z$-type errors of weight 1, or is logically equivalent to some $Z$-type error of weight 1.


Because of the aforementioned properties of a $Z$-type error on the center plane, we will try to design circuits for measuring $Z$-type generators so that most of the possible $Z$-type errors arising from a single fault are on the center plane. Finding a circuit for any $\mathtt{f}$ generator is easy since for the 3D color code in H form, any $\mathtt{f}$ generator lies on the center plane, so any CNOT ordering will work. Finding a circuit for a $\mathtt{cap}$ generator is also easy; if the first data CNOT in the circuit is the one that couples $\mathtt{q_0}$ with the syndrome ancilla, we can make sure that all possible $Z$-type errors arising from a faulty CNOT in this circuit are on the center plane (up to a multiplication of $v_0^z$ or $v_0^x$). 


Finding a circuit for measuring a $\mathtt{v}$ generator is not obvious. Since some parts of any $\mathtt{v}$ generator of $Z$ type are on the center plane and some parts are off the plane, some $Z$-type errors from a faulty data CNOT have support on some qubits which are not on the center plane. We want to make sure that in such cases, the error will not cause any problem; i.e., its syndrome must be different from those of other $Z$-type errors, or it must be logically equivalent to some $Z$-type error. In particular, we will try to avoid the case that a CNOT fault can cause a $Z$ error of weight $>1$ which is totally off-plane. This is because such a high-weight error and some $Z_i$ with $i=8,9,...,14$ may have the same syndrome but they are not logically equivalent (for example, $Z_{10}Z_{12}$ and $Z_{13}$ have the same syndrome but they are not logically equivalent).

One possible way to avoid such an error is to arrange the data CNOTs so that the qubits on which they act are alternated between on-plane and off-plane qubits. An ordering of data CNOTs used in the circuit for any $\mathtt{v}$ generator will be referenced by the ordering of data CNOTs used in the circuit for its corresponding $\mathtt{f}$ generator. For example, if the ordering of data CNOTs used for $f_4^z$ is (2,5,3,1), then the ordering of data CNOTs used for $v_1^z$ will be (2,9,5,12,3,10,1,8). A configuration of data CNOTs for a $\mathtt{v}$ generator similar to this setting will be called \emph{sawtooth configuration}. Using this configuration for every $\mathtt{v}$ generator, we find that there exists a CNOT ordering for each generator such that all possible (non-equivalent) $Z$-type errors from all circuits can be distinguished.


An example of the CNOT orderings which give a distinguishable fault set can be represented by the diagram in \cref{fig:diagram_3D}. The diagram looks similar to the 2D color code on the center plane, thus all $\mathtt{f}$ generators are displayed. The meanings of the diagram are as follows:
\begin{enumerate}
	\item Each arrow represents the ordering of data CNOTs for each $\mathtt{f}$ generator: the qubits on which data CNOTs act start from the qubit at the tail of an arrow, then proceed counterclockwise.
	\item The ordering of data CNOTs for each $\mathtt{v}$ generator can be obtained from its corresponding $\mathtt{f}$ generator using the sawtooth configuration.
	\item The ordering of data CNOTs for the $\mathtt{cap}$ generator is in numerical order.
\end{enumerate}
From the diagram, the exact orderings of data CNOTs for $\mathtt{f}$, $\mathtt{v}$, and $\mathtt{cap}$ generators are,
\begin{enumerate}
	\item $\mathtt{f}$ generators: (2,5,3,1), (3,6,4,1), and (4,7,2,1).
	\item $\mathtt{v}$ generators: (2,9,5,12,3,10,1,8), (3,10,6,13,4,11,1, 8), and (4,11,7,14,2,9,1,8).
	\item $\mathtt{cap}$ generator: (0,1,2,3,4,5,6,7).
\end{enumerate}
(Please note that these are not the only CNOT orderings which give a distinguishable fault set.)

\begin{figure}[tbp]
	\centering
	\includegraphics[width=0.18\textwidth]{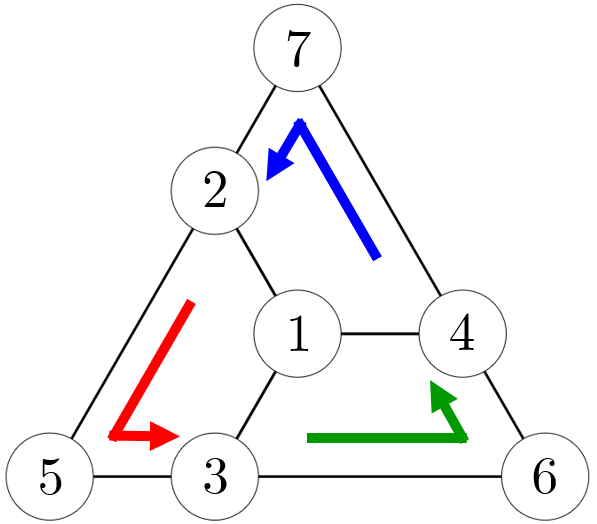}
	\caption{An example of the orderings of CNOT gates for the 3D color code of distance 3 in H form which give a distinguishable fault set $\mathcal{F}_1$. For each $\mathtt{f}$ generator, the qubits on which data CNOT gates act start from the tail of each arrow, then proceed counterclockwise. The ordering of CNOT gates for the $\mathtt{cap}$ generator is determined by the qubit numbering.}
	\label{fig:diagram_3D}
\end{figure}


Possible $Z$-type errors of weight greater than 1 depend heavily on the ordering of CNOT gates in the circuits for measuring $Z$-type generators. The exhaustive list of all possible $Z$-type errors arising from 1 fault and their syndrome corresponding to the CNOT orderings in \cref{fig:diagram_3D} is given in \cref{tab:err_list_d3}. From the list, we find that any pair of possible $Z$-type errors either have different syndromes or are logically equivalent.


Since $X$-type and $Z$-type generators have the same form, this result is also applicable to the case of $X$-type errors. In general, a single fault in any circuit can cause an error of mixed types. However, note that a single fault in a circuit for measuring a $Z$-type generator cannot cause an $X$-type error of weight $>1$ (and vice versa), and $X$-type and $Z$-type errors can be detected and corrected separately. Therefore, our results for $X$-type and $Z$-type errors implies that all fault combinations arising from up to 1 fault satisfy the condition in \cref{def:distinguishable}. This means that $\mathcal{F}_1$ is distinguishable, and the protocols in \cref{sec:FT_protocol} will be applicable. Since the circuits for measuring generators of the 3D color code are non-flag circuits, only one ancilla is required in each protocol (assuming that the qubit preparation and measurement are fast and the ancilla can be reused).


In the next section, we will generalize our technique to families of capped and recursive capped color codes, which have similar properties to the 3D color code of distance 3. Capped and recursive capped color code will be defined in \cref{subsec:CCC_def,subsec:RCCC_def} respectively, and the construction of circuits for measuring the code generators will be discussed in \cref{subsec:CCC_config}.

\begin{table*}[tbp]
	\begin{center}
		\begin{tabular}{| c | c | c | c | c || c | c | c | c | c |}
			\hline
			\multirow{2}{*}{Fault origin} & \multirow{2}{*}{Error} & \multicolumn{3}{| c ||}{Syndrome $(u,\vec{v},\vec{w})$} & \multirow{2}{*}{Fault origin} & \multirow{2}{*}{Error} & \multicolumn{3}{| c |}{Syndrome $(u,\vec{v},\vec{w})$}\\
			\cline{3-5} \cline {8-10}
			& & $u$ & $\vec{v}$ & $\vec{w}$ & & & $u$ & $\vec{v}$ & $\vec{w}$ \\
			\hline
			$\mathtt{q_{0}}$ & $Z_{0}$ & 1 & (0,0,0) & (0,0,0) & \multirow{7}{*}{$v_0^z$} & $Z_{0}$ & 1 & (0,0,0) & (0,0,0) \\
			\cline{1-5} \cline {7-10}
			$\mathtt{q_{1}}$ & $Z_{1}$ & 1 & (1,1,1) & (1,1,1) & & $Z_{0}Z_{1}$ & 0 & (1,1,1) & (1,1,1) \\
			\cline{1-5} \cline {7-10}
			$\mathtt{q_{2}}$ & $Z_{2}$ & 1 & (1,0,1) & (1,0,1) & & $Z_{0}Z_{1}Z_{2}$ & 1 & (0,1,0) & (0,1,0) \\
			\cline{1-5} \cline {7-10}
			$\mathtt{q_{3}}$ & $Z_{3}$ & 1 & (1,1,0) & (1,1,0) & & $Z_{0}Z_{1}Z_{2}Z_{3}$ & 0 & (1,0,0) & (1,0,0) \\
			\cline{1-5} \cline {7-10}
			$\mathtt{q_{4}}$ & $Z_{4}$ & 1 & (0,1,1) & (0,1,1) & & $Z_{5}Z_{6}Z_{7}$ & 1 & (1,1,1) & (1,1,1) \\
			\cline{1-5} \cline {7-10}
			$\mathtt{q_{5}}$ & $Z_{5}$ & 1 & (1,0,0) & (1,0,0) & & $Z_{6}Z_{7}$ & 0 & (0,1,1) & (0,1,1) \\
			\cline{1-5} \cline {7-10}
			$\mathtt{q_{6}}$ & $Z_{6}$ & 1 & (0,1,0) & (0,1,0) & & $Z_{7}$ & 1 & (0,0,1) & (0,0,1) \\
			\cline{1-5} \cline {6-10}
			$\mathtt{q_{7}}$ & $Z_{7}$ & 1 & (0,0,1) & (0,0,1) & \multirow{7}{*}{$v_1^z$} & $Z_{2}$ & 1 & (1,0,1) & (1,0,1) \\
			\cline{1-5} \cline {7-10}
			$\mathtt{q_{8}}$ & $Z_{8}$ & 0 & (0,0,0) & (1,1,1) & & $Z_{2}Z_{9}$ & 1 & (1,0,1) & (0,0,0) \\
			\cline{1-5} \cline {7-10}
			$\mathtt{q_{9}}$ & $Z_{9}$ & 0 & (0,0,0) & (1,0,1) & & $Z_{2}Z_{9}Z_{5}$ & 0 & (0,0,1) & (1,0,0) \\
			\cline{1-5} \cline {7-10}
			$\mathtt{q_{10}}$ & $Z_{10}$ & 0 & (0,0,0) & (1,1,0) & & $Z_{2}Z_{9}Z_{5}Z_{12}$ & 0 & (0,0,1) & (0,0,0) \\
			\cline{1-5} \cline {7-10}
			$\mathtt{q_{11}}$ & $Z_{11}$ & 0 & (0,0,0) & (0,1,1) & & $Z_{10}Z_{1}Z_{8}$ & 1 & (1,1,1) & (1,1,0) \\
			\cline{1-5} \cline {7-10}
			$\mathtt{q_{12}}$ & $Z_{12}$ & 0 & (0,0,0) & (1,0,0) & & $Z_{1}Z_{8}$ & 1 & (1,1,1) & (0,0,0) \\
			\cline{1-5} \cline {7-10}
			$\mathtt{q_{13}}$ & $Z_{13}$ & 0 & (0,0,0) & (0,1,0) & & $Z_{8}$ & 0 & (0,0,0) & (1,1,1) \\
			\cline{1-5} \cline {6-10}
			$\mathtt{q_{14}}$ & $Z_{14}$ & 0 & (0,0,0) & (0,0,1) & \multirow{7}{*}{$v_2^z$} & $Z_{3}$ & 1 & (1,1,0) & (1,1,0) \\
			\cline{1-5} \cline {7-10}
			\multirow{3}{*}{$f_4^z$} & $Z_{2}$ & 1 & (1,0,1) & (1,0,1) & & $Z_{3}Z_{10}$ & 1 & (1,1,0) & (0,0,0) \\
			\cline{2-5} \cline {7-10}
			& $Z_{2}Z_{5}$ & 0 & (0,0,1) & (0,0,1) & & $Z_{3}Z_{10}Z_{6}$ & 0 & (1,0,0) & (0,1,0) \\
			\cline{2-5} \cline {7-10}
			& $Z_{1}$ & 1 & (1,1,1) & (1,1,1) & & $Z_{3}Z_{10}Z_{6}Z_{13}$ & 0 & (1,0,0) & (0,0,0) \\
			\cline{1-5} \cline {7-10}
			\multirow{3}{*}{$f_5^z$} & $Z_{3}$ & 1 & (1,1,0) & (1,1,0) & & $Z_{11}Z_{1}Z_{8}$ & 1 & (1,1,1) & (0,1,1) \\
			\cline{2-5} \cline {7-10}
			& $Z_{3}Z_{6}$ & 0 & (1,0,0) & (1,0,0) & & $Z_{1}Z_{8}$ & 1 & (1,1,1) & (0,0,0) \\
			\cline{2-5} \cline {7-10}
			& $Z_{1}$ & 1 & (1,1,1) & (1,1,1) & & $Z_{8}$ & 0 & (0,0,0) & (1,1,1) \\
			\cline{1-5} \cline {6-10}
			\multirow{3}{*}{$f_6^z$} & $Z_{4}$ & 1 & (0,1,1) & (0,1,1) & \multirow{7}{*}{$v_3^z$} & $Z_{4}$ & 1 & (0,1,1) & (0,1,1) \\
			\cline{2-5} \cline {7-10}
			& $Z_{4}Z_{7}$ & 0 & (0,1,0) & (0,1,0) & & $Z_{4}Z_{11}$ & 1 & (0,1,1) & (0,0,0) \\
			\cline{2-5} \cline {7-10}
			& $Z_{1}$ & 1 & (1,1,1) & (1,1,1) & & $Z_{4}Z_{11}Z_{7}$ & 0 & (0,1,0) & (0,0,1) \\
			\cline{1-5} \cline {7-10}
			\multicolumn{5}{| c ||}{\multirow{2}{*}{}} & & $Z_{4}Z_{11}Z_{7}Z_{14}$ & 0 & (0,1,0) & (0,0,0) \\
			\cline {7-10}
			\multicolumn{5}{| c ||}{} & & $Z_{9}Z_{1}Z_{8}$ & 1 & (1,1,1) & (1,0,1) \\
			\cline {7-10}
			\multicolumn{5}{| c ||}{} & & $Z_{1}Z_{8}$ & 1 & (1,1,1) & (0,0,0) \\
			\cline {7-10}
			\multicolumn{5}{| c ||}{} & & $Z_{8}$ & 0 & (0,0,0) & (1,1,1) \\
			\hline
		\end{tabular}
	\end{center}
	\caption{All possible $Z$-type errors arising from 1 fault and their syndrome corresponding to the CNOT orderings in \cref{fig:diagram_3D}. Any pair of possible $Z$-type errors on the list either have different syndromes or are logically equivalent.}
	\label{tab:err_list_d3}
\end{table*}












\section{Syndrome measurement circuits for a capped color code}
\label{sec:CCC}

In the previous section, we have seen that it is possible to construct circuits for the 3D color code of distance 3 such that the fault set is distinguishable. In this section, we will extend our construction ideas to quantum codes of higher distance. First, we will introduce families of capped and recursive capped color codes, whose properties are similar to those of the 3D color codes, but the structures of the recursive capped color codes of higher distance are more suitable for our construction rather than those of the 3D color codes of higher distance (as defined in \cite{Bombin15}). Afterwards, we will apply the error correction ideas using weight parities from the previous section and develop the main theorem of this work, which can help us find proper CNOT orderings for a capped or a recursive capped color code of any distance. 

\subsection{Capped color codes}
\label{subsec:CCC_def}


We begin by defining some notations for the 2D color codes \cite{BM06} and stating some code properties. A 2D color code of distance $d$ ($d=3,5,7,\dots$) is an \codepar{n_\mathrm{2D},1,d} CSS code where $n_\mathrm{2D} = (3d^2+1)/4$. The number of stabilizer generators of each type is $r = (n_\mathrm{2D}-1)/2$ (note that the total number of generators is $2r$). For any 2D color code, it is possible to choose generators so that those of each type ($X$ or $Z$) are 3-colorable. The three smallest 2D color codes are shown in \cref{fig:2D_code}.

\begin{figure}[tbp]
	\centering
	\includegraphics[width=0.4\textwidth]{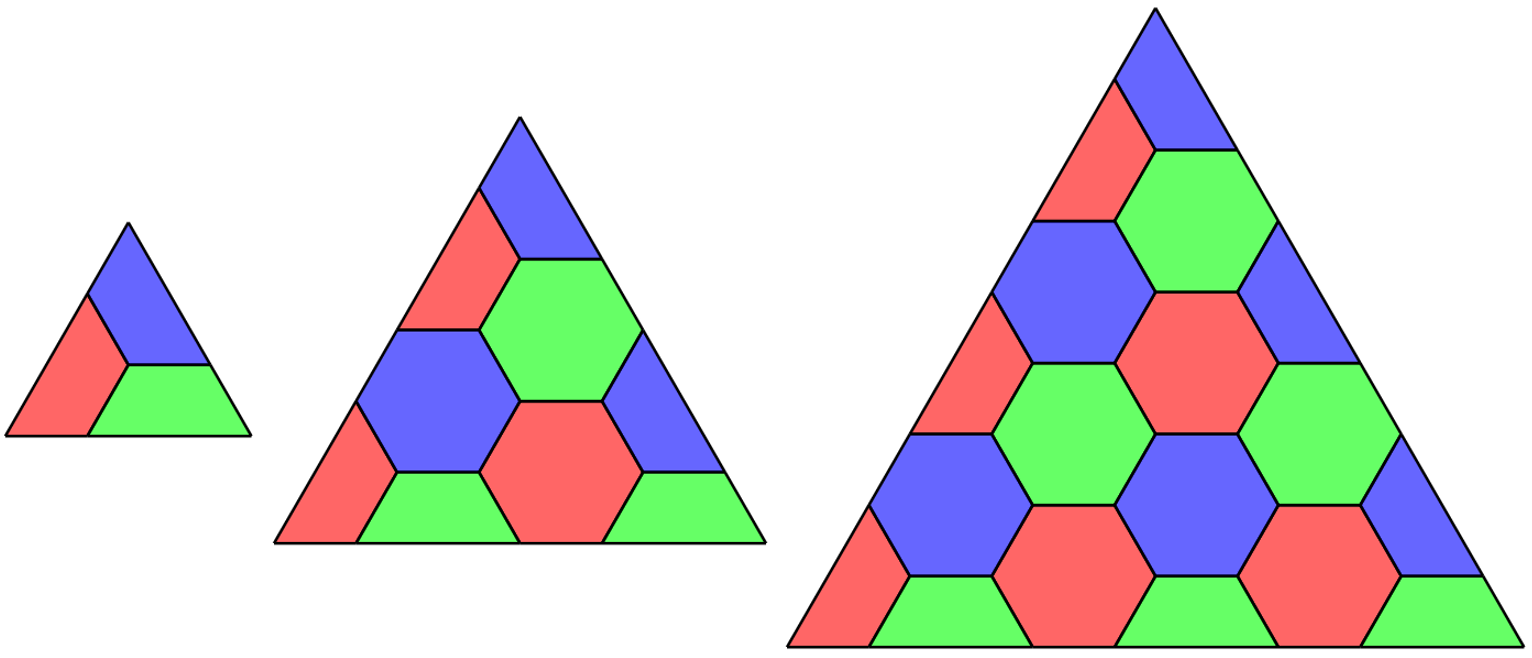}
	\captionsetup{justification=centering}
	\caption{2D color codes of distance 3, 5, and 7.}
	\label{fig:2D_code}%
\end{figure}


A 2D color code of any distance has the following properties \cite{KB15}:
\begin{enumerate}
	\item the number of qubits $n_\mathrm{2D}$ is odd,
	\item every generator has even weight,
	\item the code encodes 1 logical qubit,
	\item logical $X$ and logical $Z$ operators are of the form $X^{\otimes n_\mathrm{2D}}M$ and $Z^{\otimes n_\mathrm{2D}}N$, where $M,N$ are some stabilizers,
	\item the set of physical qubits of a 2D color code is bipartite.
\end{enumerate}
With properties 1-4, we can see that \cref{lem:err_equivalence} is applicable to a 2D color code of any distance.


A \emph{capped color code} $CCC(d)$ is constructed from 2 layers of the 2D color code of distance $d$ plus one qubit. Thus, the number of qubits of $CCC(d)$ is $2n_\mathrm{2D}+1=3(d^2+1)/2$. Examples of capped color codes with $d=5$ and $7$ are displayed in \cref{fig:CCC}\hyperlink{target:CCC}{a}, and their dual lattices are shown in \cref{fig:CCC}\hyperlink{target:CCC}{b}. Let $\mathtt{q}_i$ denote qubit $i$. For convenience, we will divide each code into 3 areas: the \emph{top qubit} (consisting of $\mathtt{q_0}$), the \emph{center plane} (consisting of $\mathtt{q_1}$ to $\mathtt{q}_{n_\mathrm{2D}}$), and the \emph{bottom plane} (consisting of $\mathtt{q}_{n_\mathrm{2D}+1}$ to $\mathtt{q}_{2n_\mathrm{2D}}$). We will primarily use the center plane as a reference, and sometimes call the qubits on the center plane \emph{on-plane qubits} and call the qubits on the bottom plane \emph{off-plane qubits}. Note that the set of physical qubits of $CCC(d)$ is also bipartite (as colored in black and white in \cref{fig:CCC}\hyperlink{target:CCC}{a}) since the set of physical qubits of any 2D color code is bipartite.

\begin{figure}[tbp]
	\centering
	\hypertarget{target:CCC}{}
	\includegraphics[width=0.47\textwidth]{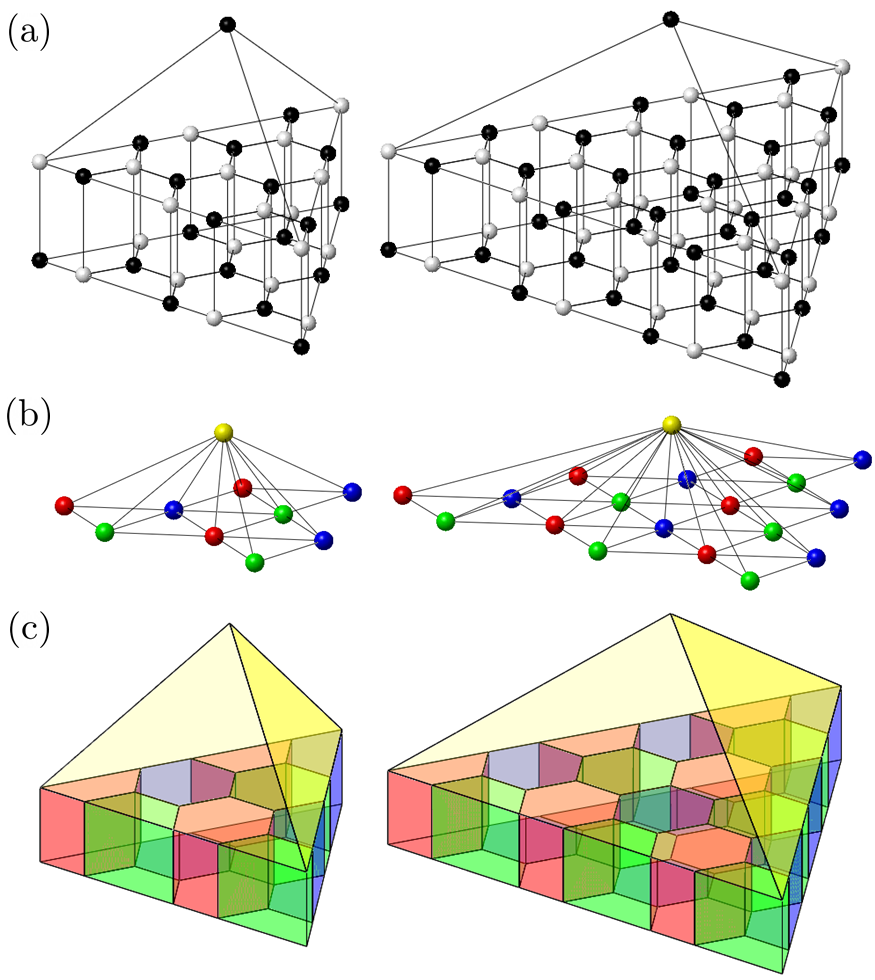}
	\caption{Capped color codes $CCC(d)$ with $d=5$ (left) and $d=7$ (right). (a) The set of qubits of any capped color code is bipartite, as displayed by black and white vertices. (b) The dual lattice of each capped color code. (c) Stabilizer generators of each code can be illustrated by volume operators.}
	\label{fig:CCC}
\end{figure}


A capped color code $CCC(d)$ is a CSS subsystem code \cite{Poulin05,Bacon06}. Its stabilizer generators are volume operators which can be defined as follows:
\begin{enumerate}
	\item $v_0^x$ and $v_0^z$ are $X$-type and $Z$-type operators that cover $\mathtt{q_0}$ and all qubits on the center plane. These operators are called $\mathtt{cap}$ generators; and
	\item $v_1^x,\dots,v_{r}^x$ and $v_1^z,\dots,v_{r}^z$ are $X$-type and $Z$-type operators in which each $v_i^x$ (or $v_i^z$) acts as an $X$-type (or a $Z$-type) generator of the 2D color code on both center and bottom planes. These operators are called $\mathtt{v}$ generators. 
\end{enumerate}
The stabilizer generators of a capped color code are illustrated in \cref{fig:CCC}\hyperlink{target:CCC}{c}. Using these notations, the stabilizer group of the code is
\begin{equation}
	S_\mathrm{CCC} = \langle v_0^x,v_1^x,\dots,v_{r}^x,v_0^z,v_1^z,\dots,v_{r}^z\rangle.
\end{equation}


For each $CCC(d)$, the generators of the gauge group are face operators which can be defined as follows:
\begin{enumerate}
	\item $f_1^x,\dots,f_{r}^x$ are $X$-type operators in which each operator acts as an $X$-type generator of the 2D color code on the center plane.
	\item $f_{r+1}^z,\dots,f_{2r}^z$ are $Z$-type operators in which each operator acts as a $Z$-type generator of the 2D color code on the center plane, and $f_i^x$ and $f_{r+i}^z$ ($i=1,\dots,r$) act on the same set of qubits.
	\item $f_1^z,\dots,f_{r}^z$ and $f_{r+1}^x,\dots,f_{2r}^x$ are $Z$-type and $X$-type operators that satisfy the following conditions:
	\begin{enumerate}
		\item $f_i^x$ and $f_j^z$ anticommute when $i=j$ ($i,j=1,\dots,2r$),
		\item $f_i^x$ and $f_j^z$ commute when $i \neq j$ ($i,j=1,\dots,2r$),
		\item $f_i^z$ and $f_{r+i}^x$ ($i=1,\dots,r$) act on the same set of qubits.
	\end{enumerate}	
\end{enumerate}
With these notations, the gauge group of each $CCC(d)$ is,
\begin{equation}
	G_\mathrm{CCC}=\langle v_i^x,v_i^z,f_j^x,f_j^z\rangle,\;
\end{equation}
where $i=0,1,\dots,r$ and $j=1,\dots,2r$.

Another way to define the gauge group of each $CCC(d)$ is to use gauge generators of weight 4 which are vertical face operators lying between the center and the bottom planes, instead of $f_1^z,\dots,f_{r}^z$ and $f_{r+1}^x,\dots,f_{2r}^x$ defined previously. Let $e_1^z,\dots,e_{r}^z$ and $e_{r+1}^x,\dots,e_{2r}^x$ denote such generators (where $e_i^z$ and $e_{r+i}^x$ act on the same set of qubits). Each pair of $e_i^z$ and $e_{r+i}^x$ can be represented by an edge on a 2D color code. For example, vertical face generators $e_1^z,\dots,e_{r}^z$ and $e_{r+1}^x,\dots,e_{2r}^x$ of capped color codes with $d=5$ and $d=7$ are depicted in \cref{fig:CCC_edge}. Note that $\{f_1^z,\dots,f_{r}^z\}$ and $\{e_1^z,\dots,e_{r}^z\}$ (or $\{f_{r+1}^x,\dots,f_{2r}^x\}$ and $\{e_{r+1}^x,\dots,e_{2r}^x\}$) generate the same group. Therefore, the gauge group of each $CCC(d)$ can also be written as,
\begin{equation}
	G_\mathrm{CCC}=\langle v_i^x,v_i^z,f_{j}^x,e_{r+j}^x,f_{r+j}^z,e_{j}^z\rangle,\;
\end{equation}
where $i=0,1,\dots,r$ and $j=1,\dots,r$.

\begin{figure}[tbp]
	\centering
	\includegraphics[width=0.45\textwidth]{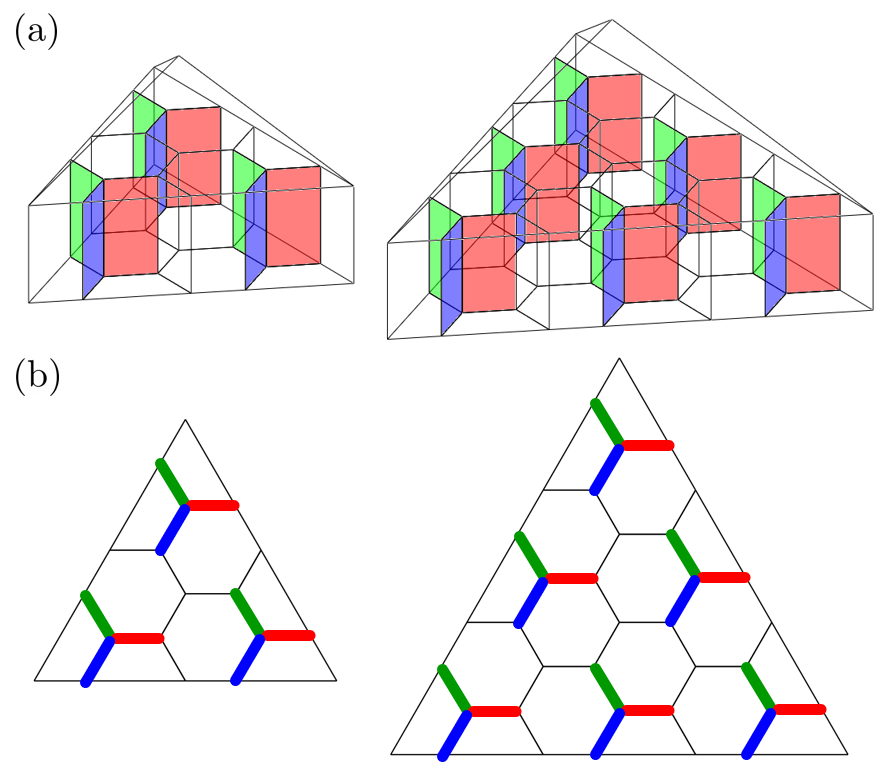}
	\caption{(a) Vertical face generators $e_1^z,\dots,e_{r}^z$ and $e_{r+1}^x,\dots,e_{2r}^x$ of capped color codes $CCC(d)$ with $d=5$ (left) and $d=7$ (right) ($e_i^z$ and $e_{r+i}^x$ act on the same set of qubits). The operators of each code can be represented by edges on a 2D color code as shown in (b).}
	\label{fig:CCC_edge}%
\end{figure}

It should be noted that in this work, the term ``color code'' is used to describe a subsystem code satisfying two conditions proposed in \cite{KB15}. This may be different from common usages in other literature in which the term refers to a stabilizer code. A capped color code is actually a color code in 3 dimensions since the dual lattice of the code (see \cref{fig:CCC}\hyperlink{target:CCC}{b} for examples) is 4-colorable and can be constructed by attaching tetrahedra together (see \cite{KB15} for more details). However, the capped color code and the 3D color code defined in \cite{Bombin15} are different codes.



A capped color code is a subsystem code which encodes 1 logical qubit, meaning that there are $n_\mathrm{2D}$ gauge qubits for each $CCC(d)$. We can clearly see that $CCC(3)$ is exactly the 3D color code of distance 3 discussed in \cref{subsec:3D_code_def}. Similarly, a stabilizer code encoding 1 logical qubit can be obtained from $CCC(d)$ by choosing $n_\mathrm{2D}$ independent, commuting gauge operators and including them in the stabilizer group. This work will discuss two possible ways to do so, and the resulting codes will be called the code in H form and the code in T form (similar to the case of the 3D color code of distance 3).

\pagebreak


\noindent\textbf{Capped color codes in H form}

Observe that the center plane of $CCC(d)$ which covers qubits 1 to $n_\mathrm{2D}$ looks exactly like the 2D color code of distance $d$. The stabilizer group of the 2D color code is $S_\mathrm{2D}=\langle f_1^x,\dots,f_{r}^x,f_{r+1}^z,\dots,f_{2r}^z \rangle$. A capped color code in H form constructed from $CCC(d)$ can be obtained by adding the stabilizer generators of the 2D color code to the original generating set of $CCC(d)$. Thus, the stabilizer group of the code in H form is,
\begin{equation}
	S_\mathrm{H} = \langle v_i^x,v_i^z,f_j^x,f_{r+j}^z\rangle, \label{eq:S_H_CCC}
\end{equation}
where $i=0,1,\dots,r$ and $j=1,2,\dots,r$. Logical $X$ and logical $Z$ operators of this code are of the form $X^{\otimes n}M$ and $Z^{\otimes n}N$, where $M,N$ are some stabilizers in $S_\mathrm{H}$. Note that \cref{lem:err_equivalence} is applicable to the code in H form constructed from any $CCC(d)$.

The code in H form is a code of distance $d$. This can be proved as follows:
\begin{proposition}
	The capped color code in H form constructed from $CCC(d)$ has distance $d$.
	\label{prop:distance_H}
\end{proposition}

\begin{proof}
	In order to prove that the distance of a stabilizer code is $d$, we will show that the weight of a nontrivial logical operator is at least $d$; that is, any Pauli error of weight $<d$ is either a stabilizer or an error with a nontrivial syndrome, and there exists a nontrivial logical operator of weight exactly $d$. Since the capped color code in H form is a CSS code and $X$-type and $Z$-type generators have the same form, we can consider $X$-type and $Z$-type errors separately. For an error occurred on the code in H form, we will represent its weight by a triple $(a,b,c)$ where $a,b,c$ are the weights of the errors occurred on the top qubit, the center plane, and the bottom plane, respectively.
	
	Suppose that a $Z$-type error has weight $k<d$. The weight of such an error will be of the form $(a,b,c)$ with $a=0$ and $b+c=k$, or with $a=1$ and $b+c=k-1$. Observe that the stabilizer generators of the 2D color code on the center plane (which is a subcode of the capped color code in H form) are $f_1^x,\dots,f_{r}^x$ and $f_{r+1}^z,\dots,f_{2r}^z$. Moreover, the 2D color code on the bottom plane is also a subcode of the capped color code in H form, whose stabilizer generators are $f_1^x\cdot v_1^x,\dots,f_{r}^x\cdot v_{r}^x$ and $f_{r+1}^z\cdot v_1^z,\dots,f_{2r}^z\cdot v_{r}^z$ (the syndrome obtained by measuring $\mathtt{v}$ generators is the sum of the syndromes obtained from the 2D color codes on both planes). Since both 2D color codes on the center and the bottom planes have distance $d$, any $Z$-type error of weight $<d$ which occurs solely on the center or the bottom plane either has nontrivial syndrome or acts as a stabilizer on such a plane. From the possible forms of error, a $Z$-type error of weight $<d$ on the capped color code in H form corresponding to the trivial syndrome must act as a stabilizer on both planes and commute with $v_0^x$. Using \cref{lem:err_equivalence}, the weight of such an operator must be in the form $(0,b,c)$ where $b,c$ are even numbers. So the total weight of the error is even, and it cannot be a logical $Z$ operator (by \cref{lem:err_equivalence}). Therefore, any $Z$-type error of weight $<d$ is either a stabilizer or an error with a nontrivial syndrome. The same analysis is applicable to $X$-type errors of weight $<d$.
	
	Next, we will show that there exists a logical $Z$ operator of weight exactly $d$. Consider a $Z$-type operator whose weight is of the form $(0,0,d)$ and acts as a logical $Z$ operator on the 2D color code on bottom plane (the operator exists because the 2D color code has distance $d$). Such an operator commutes with all generators of the capped color code in H form and has odd weight. By \cref{lem:err_equivalence}, this operator is a logical $Z$ operator. The proof is now completed.
\end{proof}

The capped color code in H form constructed from $CCC(d)$ is an \codepar{n,1,d} code where $n=2n_\mathrm{2D}+1$. Similar to the 3D color code of distance 3 in H form, it is not hard to verify that Hadamard, $S$, and CNOT gates are transversal; their logical gates are $\bar{H}=H^{\otimes n}$, $\bar{S} = {(S^\dagger)^{\otimes n}}$, and $\overline{\mathrm{CNOT}}=\mathrm{CNOT}^{\otimes n}$.

It should be noted that there are many other choices of stabilizer generators that can give the same code as what is constructed here. However, different choices of generators can give different fault sets, which may or may not be distinguishable. In \cref{subsec:CCC_config}, we will only discuss circuits for measuring generators corresponding to \cref{eq:S_H_CCC}.
\\


\noindent\textbf{Capped color codes in T form}

A capped color code in T form is constructed from $CCC(d)$ by adding all $Z$-type face generators of weight 4 to the old generating set of $CCC(d)$. That is, the stabilizer group of the code in T form is
\begin{equation}
	S_\mathrm{T} = \langle v_i^x,v_i^z,f_j^z\rangle,
\end{equation}
where $i=0,1,\dots,r$ and $j=1,2,\dots,2r$, or equivalently,
\begin{equation}
	S_\mathrm{T} = \langle v_i^x,v_i^z,e_k^z,f_{r+k}^z\rangle,
\end{equation}
where $i=0,1,\dots,r$ and $k=1,2,\dots,r$. Similar to the code in H form, logical $X$ and logical $Z$ operators of this code are of the form $X^{\otimes n}M$ and $Z^{\otimes n}N$, where $M,N$ are some stabilizers in $S_\mathrm{T}$. Note that \cref{lem:err_equivalence} is also applicable to the code in T form constructed for any $CCC(d)$.

Unlike the code in H form, the capped color code in T form constructed from $CCC(d)$ is a code of distance $3$ regardless of the parameter $d$, i.e., it is an \codepar{n,1,3} code where $n=2n_\mathrm{2D}+1$. The proof of the code distance is as follows:
\begin{proposition}
	The capped color code in T form constructed from $CCC(d)$ has distance 3.
	\label{prop:distance_T}
\end{proposition}
\begin{proof}
	Similar to the proof of \cref{prop:distance_H}, we will show that (1) any Pauli error of weight $<3$ is either a stabilizer or an error with a nontrivial syndrome, and (2) there exists a nontrivial logical operator of weight exactly $3$. However, for the capped color code in T form, $X$-type and $Z$-type generators have different forms, so we have to analyze both types of errors. Observe that all of the $Z$-type generators of the code in H form are also $Z$-type generators of the code in T form, thus we can use the analysis in the proof of \cref{prop:distance_H} to show that any $X$-type error of weight $<d$ is either a stabilizer or an error with a nontrivial syndrome. Thus, we only have to show that any $Z$-type error of weight $<3$ is either a stabilizer or an error with a nontrivial syndrome, and there exists a logical $Z$ operator of weight exactly $3$. Similar to the proof of \cref{prop:distance_H}, we will represent its weight by a triple $(a,b,c)$ where $a,b,c$ are the weights of the errors occurred on the top qubit, the center plane, and the bottom plane, respectively.
	
	The $X$-type generators of the capped color code in T form are $v_0^x,v_1^x,\dots,v_{r}^x$.  First, let us consider any $Z$-type error of weight 1. We can easily verify that the error anticommutes with at least one $X$-type generator, so its syndrome is nontrivial. Next, consider a $Z$-type error of weight 2. The weight of the error will have one of the following forms: $(0,2,0),(0,1,1),(0,0,2),(1,1,0),$ or $(1,0,1)$. We find that (1) a $Z$-type error of the form $(0,1,1)$ or $(1,0,1)$ anticommutes with $v_0^x$, and (2) a $Z$-type errors of the form $(0,2,0),(0,0,2),$ or $(1,1,0)$ anticommutes with at least one $\mathtt{v}$ generator (since $\mathtt{v}$ generators act as generators of the 2D color code on both planes simultaneously, and the 2D color code has distance $d$). Therefore, the syndrome of any $Z$-type error of weight 2 is nontrivial.
	
	Next, we will show that there exists a logical $Z$ operator of distance exactly 3. Consider a $Z$-type operator of weight 3 of the form $Z_{0}Z_{i}Z_{r+i}$, where $i=1,2,\dots,r$. We can verify that such an operator commutes with all $X$-type generators. Since the operator has odd weight, it is a logical $Z$ operator by \cref{lem:err_equivalence}.
\end{proof}


CNOT and $T$ gates are transversal for the code in T form, while Hadamard and $S$ gates are not. In order to prove the transversalily of the $T$ gate, we will use the following lemma \cite{KB15}:

\begin{lemma}
	Let $C$ be an \codepar{n,k,d} CSS subsystem code in which $n$ is odd, $k$ is 1, and $X^{\otimes n}$ and $Z^{\otimes n}$ are bare logical $X$ and $Z$ operators\footnote{A bare logical operator is a logical operator that acts on the logical qubit(s) of a subsystem code and does not affect the gauge qubit(s); see \cite{Poulin05,Bacon06,Bravyi11}.}. Also, let $Q$ be the set of all physical qubits of $C$, and let $p$ be any positive integer. Suppose there exists $V \subset Q$ such that for any $m=1,\dots,p$, for every subset $\{g_1^x,\dots,g_m^x\}$ of the $X$-type gauge generators of the code, the following holds:
	\begin{equation}
		\left| V \cap \bigcap_{i=1}^m G_i \right| = \left| V^c \cap \bigcap_{i=1}^m G_i \right| \mod 2^{p-m+1}, \label{eq:eq_lem2}
	\end{equation}
	where $G_i$ is the set of physical qubits that support $g_i^x$. Then, a logical $R_p$ gate (denoted by $\bar{R}_p$) can be implemented by applying $R_p^q$ to all qubits in $V$ and applying $R_p^{-q}$ to all qubits in $V^c$, where $R_p=\mathrm{diag}\left(1,\mathrm{exp}\left(2\pi i/2^p\right)\right)$, $q$ is a solution to $q(|V|-|V^c|)=1 \mod 2^p$, and $V^c=Q\backslash V$.
	\label{lem:transversal_R}
\end{lemma}

The proof of the transversality of the $T$ gate is as follows:
\begin{proposition}
	A $T$ gate is transversal for the capped color code in T form constructed from any $CCC(d)$.
	\label{prop:transversal_T}
\end{proposition}
\begin{proof}
	Let $C$ be the capped color code in T form constructed from any $CCC(d)$ ($C$ is a stabilizer code, i.e., it is a subsystem code in which the stabilizer group and the gauge group are the same). Note that the $X$-type stabilizer generators of the code are $v_0^x,v_1^x,\dots,v_{r}^x$, which are also the $X$-type gauge generators. Also, let $p=3$ (since $T = R_3$), $q=1$, and let $V$ and $V^c$ be the sets of qubits similar to those represented by black and white vertices in \cref{fig:CCC}\hyperlink{target:CCC}{a} (this kind of representation is always possible for any $CCC(d)$ since the set of physical qubits of $CCC(d)$ is bipartite). We will use \cref{lem:transversal_R} and show that \cref{eq:eq_lem2} is satisfied for $m=1,2,3$.
	
	Let $G_i$ be the set of qubits that support $X$-type generator $g_i^x$. If $m=1$, we can easily verify that $\left| V \cap G_1 \right| = \left| V^c \cap G_1 \right|\mod 8$ for every $g_1^x \in \{v_0^x,v_1^x,\dots,v_{r}^x\}$ since half of supporting qubits of any $X$-type generator is in $V$ and the other half is in $V^c$. 
	
	In the case when $m=2$, let $\{g_1^x,g_2^x\}$ be a subset of $\{v_0^x,v_1^x,\dots,v_{r}^x\}$. If $g_1^x$ is a $\mathtt{cap}$ generator $v_0^x$ and $g_2^x$ is a $\mathtt{v}$ generator $v_i^x, i=1,\dots,r$, then $G_1\cap G_2$ are the qubits that support the face generator $f_i^x$. Since half of qubits in $G_1\cap G_2$ is in $V$ and the other half is in $V^c$, we have that $\left| V \cap G_1 \cap G_2 \right| = \left| V^c \cap G_1 \cap G_2 \right|$ (equal to 2 or 3, depending on $v_i^x$). If $g_1^x$ and $g_2^x$ are adjacent $\mathtt{v}$ generators, then $G_1\cap G_2$ have 4 qubits, two of them are in $V$ and the other two are in $V^c$. So $\left| V \cap G_1 \cap G_2 \right| = \left| V^c \cap G_1 \cap G_2 \right| = 2$. If $g_1^x$ and $g_2^x$ are non-adjacent $\mathtt{v}$ generators, then $\left| V \cap G_1 \cap G_2 \right| = \left| V^c \cap G_1 \cap G_2 \right| = 0$. Therefore, \cref{eq:eq_lem2} is satisfied for any subset $\{g_1^x,g_2^x\}$.
	
	In the case when $m=3$, let $\{g_1^x,g_2^x,g_3^x\}$ be a subset of $\{v_0^x,v_1^x,\dots,v_{r}^x\}$. If $g_1^x$ is a $\mathtt{cap}$ generator $v_0^x$ and $g_2^x,g_3^x$ are adjacent $\mathtt{v}$ generators, or $g_1^x,g_2^x,g_3^x$ are $\mathtt{v}$ generators in which any two of them are adjacent, then $G_1\cap G_2 \cap G_3$ have 2 qubits, one of them is in $V$ and the other one is in $V^c$. Thus, $\left| V \cap G_1 \cap G_2 \cap G_3\right| = \left| V^c \cap G_1 \cap G_2 \cap G_3 \right|=1$. If $g_1^x$ is a $\mathtt{cap}$ generator $v_0^x$ and $g_2^x,g_3^x$ are non-adjacent $\mathtt{v}$ generators, or $g_1^x,g_2^x,g_3^x$ are $\mathtt{v}$ generators in which some pair of them are not adjacent, then $G_1 \cap G_2 \cap G_3$ is the empty set. So $\left| V \cap G_1 \cap G_2 \cap G_3 \right| = \left| V^c \cap G_1 \cap G_2 \cap G_3 \right| = 0$. Therefore, \cref{eq:eq_lem2} is satisfied for any subset $\{g_1^x,g_2^x,g_3^x\}$.
	
	Since the sufficient condition in \cref{lem:transversal_R} is satisfied, a transversal $T$ gate can be implemented by applying $T$ gates to all qubits in $V$ (represented by black vertices) and applying $T^\dagger$ gates to all qubits in $V^c$ (represented by white vertices).
\end{proof}

Incidentally, the capped color codes in T form presented here are similar to some codes that appear in other literature. In fact, the capped color codes in T form is the same as the stacked codes with distance 3 protection defined in \cite{JB16} (where alternative proofs of \cref{prop:distance_T,prop:transversal_T} are also presented). Such a code is the basis for the construction of the $(d-1)+1$ stacked code defined in the same work, whose code distance is $d$ (see also \cite{BC15,JBH16} for other subsystem codes with similar construction).
\\



\noindent\textbf{Code switching}

Similar to the 3D color code of distance 3, one can transform between the capped color code in H form and the code in T form derived from the same $CCC(d)$ using the code switching technique \cite{PR13,ADP14,Bombin15,KB15}. Suppose that we start from the code in H form. The code switching can be done by first measuring $e_1^z,\dots,e_{r}^z$ (vertical face generators of weight 4), then applying an $X$-type Pauli operator that
\begin{enumerate}
	\item commutes with all $v_i^x$'s and $v_i^z$'s ($i=0,1,\dots,r$), and
	\item commutes with $f_{r+1}^z,\dots,f_{2r}^z$, and
	\item for each $j=1,\dots,r$, commutes with $e_j^z$ if the outcome from measuring such an operator is 0 (the eigenvalue is +1) or anticommutes with $e_j^z$ if the outcome is 1 (the eigenvalue is $-1$).
\end{enumerate}
We can use a similar process to switch from the code in T form to the code in H form, except that $f_{1}^x,\dots,f_{r}^x$ will be measured and the operator to be applied is a $Z$-type Pauli operator that commutes or anticommutes with $f_{1}^x,\dots,f_{r}^x$ (depending on the measurement outcomes).


\subsection{Recursive capped color codes}
\label{subsec:RCCC_def}
	
	
One drawback of a capped color code $CCC(d)$ is that the code in T form has only distance 3 regardless of the parameter $d$. This prevents us from performing fault-tolerant $T$ gate implementation through code switching because a few faults that occur to the code in T form can lead to a logical error. In this section, we will introduce a way to construct a code of distance $d$ from capped color codes through a process of recursive encoding. The resulting code will be called \emph{recursive capped color code}.

First, let us consider a capped color code in T form obtained from any (subsystem) capped color code $CCC(d)$. There are many possible errors of weight 3 which are nontrivial logical errors of this code, but all of them have one thing in common: any logical error of weight 3 has support on $\mathtt{q_0}$ (the top qubit of a capped color code). So if we can reduce the error rate on $\mathtt{q_0}$, a logical error of weight 3 on a capped color code in T form will be less likely. In particular, if $\mathtt{q_0}$ is encoded by a code of distance $d-2$, the distance of the resulting code will be $d$.


We define a \emph{recursive capped color code} $RCCC(d)$ ($d=3,5,7,\dots$) to be a subsystem CSS code obtained from the following procedure:
\begin{enumerate}
	\item $RCCC(3)$ and $CCC(3)$ are the same code.
	\item $RCCC(d)$ is obtained by encoding $\mathtt{q_0}$ (the top qubit) of $CCC(d)$ by $RCCC(d-2)$.
\end{enumerate}
Constructing a recursive capped color code is similar to constructing a concatenated code. However, instead of encoding every physical qubit of the original code by another code, here we only encode $\mathtt{q_0}$ of a capped color code by a recursive capped color code with smaller parameter. 



It should be noted that a stacked code of distance $d$ \cite{JB16} can be obtained using a recursive encoding procedure similar to the one presented above. However, in that case, the top qubit of a capped color code $CCC(d)$ is encoded by the same capped color code ($CCC(d)$), and the procedure is repeated $(d-3)/2$ times. The recursive capped color code $RCCC(d)$ with $d=7$ and the stacked code of distance 7 are illustrated in \cref{fig:RCCCvsStacked}.

\begin{figure}[tbp]
	\centering
	\begin{subfigure}[b]{0.22\textwidth}
		\includegraphics[width=\textwidth]{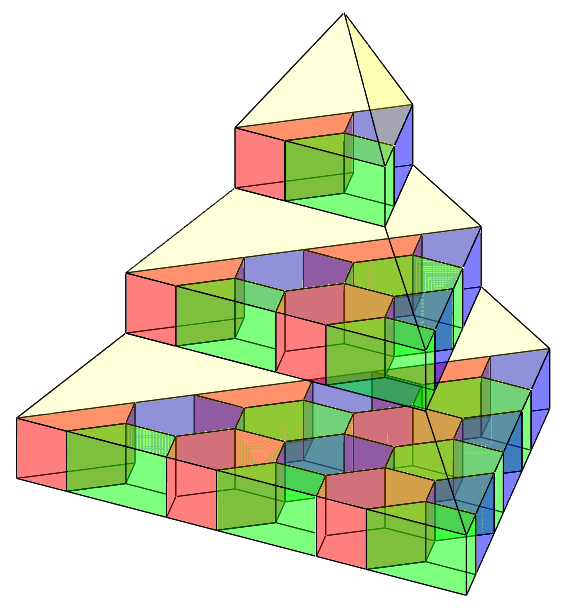}
		\captionsetup{justification=centering}
		\caption{}
		\label{subfig:RCCC_d7}
	\end{subfigure}	
	\begin{subfigure}[b]{0.22\textwidth}
		\includegraphics[width=\textwidth]{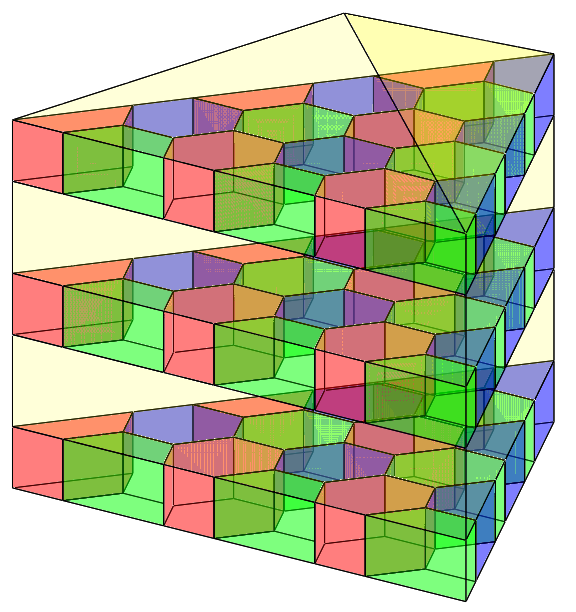}
		\captionsetup{justification=centering}
		\caption{}
		\label{subfig:Stacked_d7}
	\end{subfigure}
	\caption{(a) The recursive capped color code $RCCC(d)$ with $d=7$. (b) The stacked code of distance 7.}
	\label{fig:RCCCvsStacked}%
\end{figure}


The number of qubits of $RCCC(d)$ is $(d^3+3d^2+3d-3)/4$. For convenience, we will divide each $RCCC(d)$ into $d$ layers: 
\begin{enumerate}
	\item The first layer consists of the top qubit $\mathtt{q_0}$.
	\item The $(j-1)$-th layer where $j=3,5,\dots,d$ (which is similar to a 2D color code of distance $j$) will be called the center plane of inner $CCC(j)$. 
	\item The $j$-th layer where $j=3,5,\dots,d$ (which is similar to a 2D color code of distance $j$) will be called the bottom plane of inner $CCC(j)$.
\end{enumerate}
Similar to a capped color code, the stabilizer generators of $RCCC(d)$ are defined by volume operators of $X$ and $Z$ types, and the gauge generators are defined by volume and face generators of $X$ and $Z$ types.\\	
	
\noindent\textbf{Recursive capped color codes in H form}

The stabilizer group of a recursive capped color code in H form can be obtained by adding $X$- and $Z$-type face generators which are generators of 2D color codes on the center planes (layers $2,4,...,d-1$) to the original stabilizer generating set of $RCCC(d)$. Similar to the construction of the subsystem code previously described, the recursive capped color code in H form constructed from $RCCC(d)$ can also be obtained by encoding the top qubit $\mathtt{q_0}$ of the capped color code in H form constructed from $CCC(d)$ by the recursive capped color code in H form constructed from $RCCC(d-2)$.

The recursive capped color code in H form constructed from $RCCC(d)$ has distance $d$. This can be proved as follows:

\begin{proposition}
	The recursive capped color code in H form constructed from $RCCC(d)$ has distance $d$.
	\label{prop:distance_H_recur}
\end{proposition}
\begin{proof}
	Consider a capped color code in H form constructed from $CCC(d)$ which has distance $d$ (see \cref{prop:distance_H}). One example of a logical error of weight $d$ of this code is a logical error of a 2D color code of distance $d$ on the bottom plane. Encoding the top qubit of a capped color code in H form by the recursive capped color code in H form constructed from $RCCC(d-2)$ will not affect the aforementioned logical error, so the distance of the resulting code is still $d$.
\end{proof}

A recursive capped color code in H form constructed from $RCCC(d)$ is an \codepar{n,1,d} code where $n=(d^3+3d^2+3d-3)/4$. Similar to a capped color code in H form, a recursive capped color code in H form also possesses transversal Hadamard, $S$, and CNOT gates, where $\bar{H}=H^{\otimes n}$, $\overline{\mathrm{CNOT}}=\mathrm{CNOT}^{\otimes n}$, $\bar{S} = {(S^\dagger)^{\otimes n}}$ when $d=3,7,11,...$, and $\bar{S} = {(S)^{\otimes n}}$ when $d=5,9,13,...$.\\

\noindent\textbf{Recursive capped color codes in T form}

Consider the $(j-1)$-th and $j$-th layers ($j=3,5,\dots,d$) of a subsystem code $RCCC(d)$, which are similar to 2D color codes of distance $j$. We can define vertical face generators of inner $CCC(j)$ between these two layers similar to the way we define vertical face generators for $CCC(d)$ in \cref{subsec:CCC_def} (see \cref{fig:CCC_edge} for examples). The stabilizer group of a recursive capped color code in T form can be obtained by adding vertical face generators of $Z$ type of all inner $CCC(j)$'s to the original stabilizer generating set of $RCCC(d)$. Also, similar to the construction of the subsystem code $RCCC(d)$, the recursive capped color code in T form constructed from $RCCC(d)$ can be obtained by encoding the top qubit $\mathtt{q_0}$ of the capped color code in T form constructed from $CCC(d)$ by the recursive capped color code in T form constructed from $RCCC(d-2)$.

Unlike the capped color code in T form constructed from $CCC(d)$ whose distance is 3 regardless of the parameter $d$, the recursive capped color code in T form constructed from $RCCC(d)$ has distance $d$. This can be proved as follows:

\begin{proposition}
	The recursive capped color code in T form constructed from $CCC(d)$ has distance $d$.
	\label{prop:distance_T_recur}
\end{proposition}
\begin{proof}
	Consider a capped color code in T form constructed from $CCC(d)$ which has distance 3 (see \cref{prop:distance_T}). We find that any logical error of weight $3$ is of $Z$ type and has support on $\mathtt{q_0}$ (the top qubit of the capped color code). Suppose that $\mathtt{q_0}$ is encoded by a code of distance $d-2$, effectively becoming an inner logical qubit $\bar{\mathtt{q}}_\mathtt{0}$. To create a logical error on the resulting code similar to the logical error of weight 3 on a capped color code in T form, we need an error on $\bar{\mathtt{q}}_\mathtt{0}$ plus errors on two more qubits. Thus, the minimum weight of a logical error of the resulting code is $(d-2)+2=d$. In our case, the code being used to encode $\mathtt{q_0}$ is the recursive capped color code in T form constructed from $RCCC(d-2)$. By induction, the recursive capped color code in T form constructed from $RCCC(d)$ has distance $d$.
\end{proof}

A recursive capped color code in T form constructed from $RCCC(d)$ is an \codepar{n,1,d} code where $n=(d^3+3d^2+3d-3)/4$. Similar to a capped color code in T form, a recursive capped color code in T form also possesses transversal CNOT and $T$ gates. The proof of transversality of the $T$ gate is as follows:
\begin{proposition}
	A $T$ gate is transversal for the recursive capped color code in T form constructed from any $RCCC(d)$.
	\label{prop:transversal_T_recur}
\end{proposition}
\begin{proof}
	Recall that for any capped color code in T form, by \cref{prop:transversal_T}, a logical $T$ gate can be achieved by applying physical $T$ and $T^\dagger$ gates on qubits represented by black and white vertices, respectively (see \cref{fig:CCC}\hyperlink{target:CCC}{a} for examples; the representation can be extended to any $CCC(d)$ since the set of physical qubits of $CCC(d)$ is bipartite). Suppose that the top qubit $\mathtt{q_0}$ of a capped color code in T form constructed from $CCC(d)$ is encoded by the recursive capped color code in T form constructed from $RCCC(d-2)$, becoming an inner logical qubit $\bar{\mathtt{q}}_\mathtt{0}$. A logical $T$ gate of the resulting code is similar to a logical $T$ gate of the capped color code, except that an (inner) logical $T$ gate is applied on $\bar{\mathtt{q}}_\mathtt{0}$. By induction, the $T$ gate for $\bar{\mathtt{q}}_\mathtt{0}$ is transversal, and the $T$ gate for the recursive capped color code in T form constructed from $RCCC(d)$ is also transversal.
\end{proof}

\noindent\textbf{Code switching}

Similar to the capped color codes, the code switching technique can be used to transform between the recursive capped color codes in H and T forms constructed from the same $RCCC(d)$. In particular, we can switch from the code in H form to the code in T form by measuring $Z$-type vertical face generators of all inner $CCC(j)$'s and apply an appropriate Pauli operator depending on the measurement outcome. Switching from the code in T form to the code in H form can be done in a similar fashion, except that $X$-type generators of 2D color codes on the center planes (layers $2,4,...,d-1$) will be measured instead. Please refer to the process of finding a appropriate Pauli operator for code switching in \cref{subsec:CCC_def}. 

We have not yet discussed whether the procedure above is fault tolerant when we switch between the recursive capped color codes in H form and T form. However, the discussion of the fault-tolerant implementation of $T$ gate will be deferred until \cref{subsec:FT_T_gate}.

\subsection{Circuit configuration for capped and recursive capped color codes}
\label{subsec:CCC_config}


One of the main goals of this work is to find circuits for measuring generators of a capped color code in H form in which the corresponding fault set $\mathcal{F}_t$ is distinguishable (where $t=\tau=(d-1)/2$ and $d=3,5,7,...$ is the code distance), and we expect that similar circuits will also work for a recursive capped color code in H form. As discussed before, the CNOT orderings and the number of flag ancillas are crucial for the circuit design. Finding such circuits for a capped color code of any distance using a random approach can be very challenging because of a few reasons: (1) the number of stabilizer generators of a capped color code increases quadratically as the distance increases. This means that the number of possible single faults in the circuits grow quadratically as well. (2) for a code with larger distance, a fault set $\mathcal{F}_t$ with larger $t$ will be considered. Since it concerns all possible fault combinations arising from up to $t$ faults, the size of $\mathcal{F}_t$ grows dramatically (perhaps exponentially) as $t$ and the number of possible single faults increase. For these reasons, verifying whether $\mathcal{F}_t$ is distinguishable using the conditions in \cref{def:distinguishable} requires a lot of computational resources, and exhaustive search for appropriate CNOT orderings may turn intractable.



Fortunately, there is a way to simplify the search for the CNOT orderings. From the structure of the capped color code in H form, it is possible to relate CNOT orderings for the 3D-like generators to those for the 2D-like generators, as we have seen in the circuit construction in \cref{subsec:3D_code_config}. Instead of finding CNOT orderings directly for all generators, we will simplify the problem and develop sufficient conditions for the CNOT orderings of the 2D-like generators which, if satisfied, can guarantee that the fault set $\mathcal{F}_t$ (which concerns both 3D-like and 2D-like generators) is distinguishable. Although we still need to check whether the sufficient conditions are satisfied for given CNOT orderings, the process is much simpler than checking the conditions in \cref{def:distinguishable} directly when the size of $\mathcal{F}_t$ is large.



We begin by dividing the stabilizer generators of the capped color code in H form constructed from $CCC(d)$ into 3 categories (similar to the discussion in \cref{subsec:3D_code_config}):

\begin{enumerate}
	\item $\mathtt{cap}$ generators consisting of $v_0^x$ and $v_0^z$,
	\item $\mathtt{v}$ generators consisting of $v_1^x,\dots,v_{r}^x,$ and $v_1^z,\dots,$ $v_{r}^z$,
	\item $\mathtt{f}$ generators consisting of $f_1^x,\dots,f_{r}^x,$ and $f_{r+1}^z,$ $\dots,f_{2r}^z$.
\end{enumerate}
Here we will only consider fault combinations arising from circuits for measuring $Z$-type generators which can lead to purely $Z$-type data errors of any weight. This is because $i$ faults in circuits for measuring $X$-type generators cannot cause $Z$-type data error of weight greater than $i$ (and vice versa). Similar analysis will be applicable to the case of purely $X$-type errors, and also the case of mixed-type errors. We will first consider a $Z$-type data error and a flag vector arising from each single fault. Afterwards, fault combinations constructed from multiple faults will be considered, where the combined data error and the cumulative flag vector for each fault combination can be calculated using \cref{eq:combined_E,eq:cumulative_f}.

Observe that the center plane of a capped color code behaves like a 2D color code, and the weight of a $Z$-type error that occurred on the center plane can be measured by the $\mathtt{cap}$ generator $v_0^x$. In order to find CNOT orderings for generators of each category, we will use an idea similar to that presented in \cref{subsec:3D_code_config}; we will try to design circuits for measuring $Z$-type generators so that most of possible $Z$-type errors arising from a single fault are on the center plane. In this work, we will start by imposing general configurations of data and flag CNOT gates; these general configurations will facilitate finding CNOT orderings. Then, exact configurations of CNOT gates which can make $\mathcal{F}_t$ distinguishable will be found using the theorem developed later in this section. The general configurations of data CNOT gates, which depend on the category of the generator, are as follows:\\

\noindent\textbf{General configurations of data CNOT gates}

\begin{enumerate}
	\item $\mathtt{f}$ generator: there is no constraint for the ordering of data CNOTs since each $\mathtt{f}$ generator lies on the center plane, but the ordering for $f_{r+i}^z$ (or $f_i^x$) must be related to the ordering for $v_i^z$ (or $v_i^x$) where $i=1,\dots,r$.
	\item $\mathtt{v}$ generator: The \emph{sawtooth configuration} will be used; the qubits on which the data CNOTs act must be alternated between on-plane and off-plane qubits. The ordering of data CNOTs for $v_i^z$ (or $v_i^x$) is referenced by the ordering of data CNOTs for $f_{r+i}^z$ (or $f_i^x$) where $i=1,\dots,r$ (see examples in \cref{fig:flag_config} and \cref{subsec:3D_code_config}).
	\item $\mathtt{cap}$ generator: The first data CNOT must always be the one that couples $\mathtt{q_0}$ with the syndrome ancilla. The ordering of the other data CNOTs has yet to be fixed.
\end{enumerate}

\begin{figure}[tbp]
	\centering
	\begin{subfigure}{0.33\textwidth}
		\includegraphics[width=\textwidth]{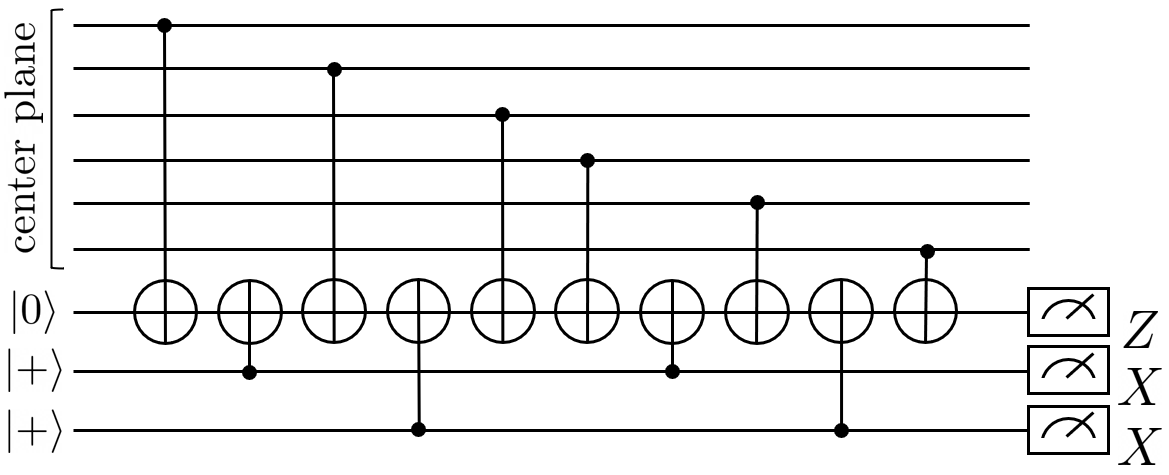}
		\captionsetup{justification=centering}
		\caption{}
		\label{subfig:flag_config_f}
	\end{subfigure}	
	\begin{subfigure}{0.48\textwidth}
		\includegraphics[width=\textwidth]{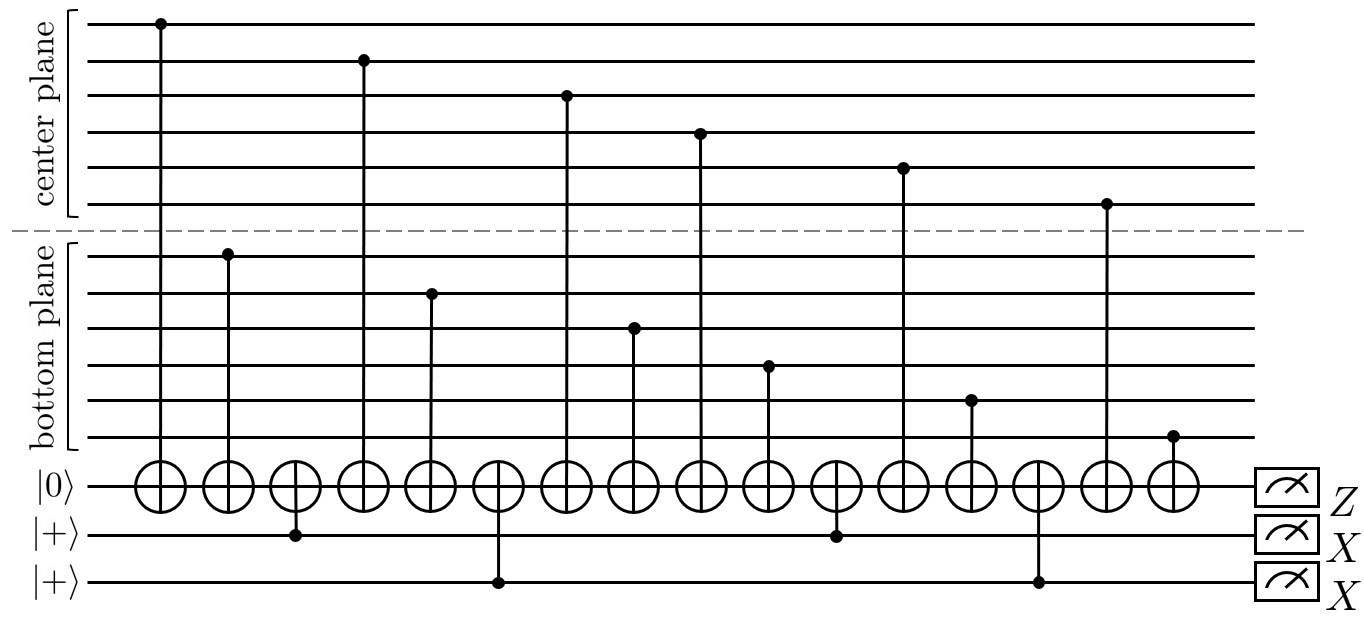}
		\captionsetup{justification=centering}
		\caption{}
		\label{subfig:flag_config_v}
	\end{subfigure}
	\caption{(a) An example of flag circuit for measuring $f$ generator with two flag ancillas. (b) A flag circuit for measuring the corresponding $v$ generator. The circuit is obtained by replacing each data CNOT which couples $\mathtt{q}_j$ with the syndrome ancilla by two data CNOTs which couple $\mathtt{q}_j$ and $\mathtt{q}_{n_\mathrm{2D}+j}$ with the syndrome ancilla.}
	\label{fig:flag_config}
\end{figure}


In general, some flag ancillas may be added to the circuits for measuring a generator to help distinguish some possible errors and make $\mathcal{F}_t$ distinguishable. In that case, the general configurations for data CNOT gates will also be applied to the data CNOTs in each flag circuit. Moreover, additional configurations for flag CNOT gates will be required.\\

\noindent\textbf{General configurations of flag CNOT gates}

\begin{enumerate}
	\item For each flag circuit, the first and the last data CNOTs must not be in between any pair of flag CNOT gates.
	\item The arrangements of flag CNOTs in the circuits for each pair of $\mathtt{f}$ and $\mathtt{v}$ generators must be similar; Suppose that a flag circuit for $f_{r+i}^z$ (or $f_i^x$) where $i=1,\dots,r$ is given. A flag circuit for $v_i^z$ (or $v_i^x$) is obtained by replacing each data CNOT which couples $\mathtt{q}_j$ with the syndrome ancilla ($j=1,\dots,n_\mathrm{2D}$) by two data CNOTs which couple $\mathtt{q}_j$ and $\mathtt{q}_{n_\mathrm{2D}+j}$ with the syndrome ancilla; see an example in \cref{fig:flag_config}.
\end{enumerate}

By imposing the general configurations for data and flag CNOTs, what have yet to be determined before $\mathcal{F}_t$ is specified are the ordering of data CNOTs for each $f$ generator, the ordering of data CNOTs after the first data CNOT for each $\mathtt{cap}$ generator, and the number of flag ancillas and the ordering of their relevant flag CNOTs. (Note that having more flag ancillas can make fault distinguishing become easier, but more resources such as qubits and gates are also required.)

\begin{figure}[tbp]
	\centering
	\includegraphics[width=0.13\textwidth]{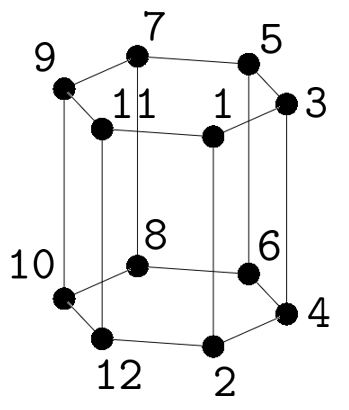}
	\caption{Consider a circuit for measuring a $\mathtt{v}$ generator of $Z$ type in which its supporting qubits are labeled as displayed above and the ordering of data CNOT gates is $(1,2,\dots,12)$. A single fault in the circuit is either $\mathtt{v}$ type or $\mathtt{v^*}$ type, depending on whether the data errors on the center and the bottom planes have the same form. For example, an $IZ$ fault on the 7th data CNOT is a $\mathtt{v^*}$ fault since the data error arising from the fault is $Z_9 Z_{11} \otimes Z_8 Z_{10} Z_{12}$, while an $IZ$ fault on the 8th data CNOT is a $\mathtt{v}$ fault since the data error arising from the fault is $Z_9 Z_{11}\otimes Z_{10} Z_{12}$.}
	\label{fig:v_vstar}
\end{figure}



In this work, possible single faults which can give $Z$-type errors will be divided into 7 types (based on relevant faulty locations) as follows:
\begin{enumerate}
	\item Type $\mathtt{q_0}$: a fault causing a $Z$-type error on $\mathtt{q_0}$ which does not arise from any $Z$-type generator measurement. The total number of $\mathtt{q_0}$ faults is $n_0$ (which is 0 or 1).
	\item Type $\mathtt{q_{on}}$: a fault causing a single-qubit $Z$-type error on the center plane which does not arise from any $Z$-type generator measurement. The syndrome of an error is denoted by $\vec{q}_\mathtt{on}$. The total number of $\mathtt{q_{on}}$ faults is $n_\mathtt{on}$.
	\item Type $\mathtt{q_{off}}$: a fault causing a single-qubit $Z$-type error on the bottom plane which does not arise from any $Z$-type generator measurement. The syndrome of an error is denoted by $\vec{q}_\mathtt{off}$. The total number of $\mathtt{q_{off}}$ faults is $n_\mathtt{off}$.
	\item Type $\mathtt{f}$: a fault occurred during a measurement of a $\mathtt{f}$ generator of $Z$ type. A $Z$-type error from each fault of this type and its syndrome are denoted by $\sigma_\mathtt{f}$ and $\vec{p}_\mathtt{f}$. A flag vector corresponding to each fault of this type is denoted by $\vec{f}_\mathtt{f}$. The total number of $\mathtt{f}$ faults is $n_\mathtt{f}$.
	\item Type $\mathtt{v}$: a fault occurred during a measurement of a $\mathtt{v}$ generator of $Z$ type which give errors of the same form on both center and bottom planes (see an example in \cref{fig:v_vstar}). A part of a $Z$-type error from each fault of this type occurred on the center plane only (or the bottom plane only) and its syndrome are denoted by $\sigma_\mathtt{v}$ and $\vec{p}_\mathtt{v}$. A flag vector corresponding to each fault of this type is denoted by $\vec{f}_\mathtt{v}$. The total number of $\mathtt{v}$ faults is $n_\mathtt{v}$. 
	\item Type $\mathtt{v^*}$: a fault occurred during a measurement of a $\mathtt{v}$ generator of $Z$ type in which an error occurred on the center plane and an error on the bottom plane are different (see an example in \cref{fig:v_vstar}). A part of a $Z$-type error from each fault of this type occurred on the center plane only and its syndrome are denoted by $\sigma_\mathtt{v^*,cen}$ and $\vec{p}_\mathtt{v^*,cen}$. The other part of the $Z$-type error that occurred on the bottom plane only and its syndrome are denoted by $\sigma_\mathtt{v^*,bot}$ and $\vec{p}_\mathtt{v^*,bot}$. A flag vector corresponding to each fault of this type is denoted by $\vec{f}_\mathtt{v^*}$. The total number of $\mathtt{v^*}$ faults is $n_\mathtt{v^*}$.
	\item Type $\mathtt{cap}$: a fault occurred during a measurement of a $\mathtt{cap}$ generator of $Z$ type. A $Z$-type error from each fault of this type and its syndrome are denoted by $\sigma_\mathtt{cap}$ and $\vec{p}_\mathtt{cap}$ ($\sigma_\mathtt{cap}$ is always on the center plane up to a multiplication of the $\mathtt{cap}$ generator being measured). A flag vector corresponding to each fault of this type is denoted by $\vec{f}_\mathtt{cap}$. The total number of $\mathtt{cap}$ faults is $n_\mathtt{cap}$.
\end{enumerate}
Examples of faults of each type on the 3D structure are illustrated in \cref{fig:fault_2D_3D}\hyperlink{target:fault_2D_3D}{a}.
Note that a fault of $\mathtt{q_0}$, $\mathtt{q_{on}}$, or $\mathtt{q_{off}}$ type can be a $Z$-type input error, a single-qubit error from phase flip, or a single fault during any $X$-type generator measurement which gives a $Z$-type error.

\begin{figure}[tbp]
	\centering
	\hypertarget{target:fault_2D_3D}{}
	\includegraphics[width=0.35\textwidth]{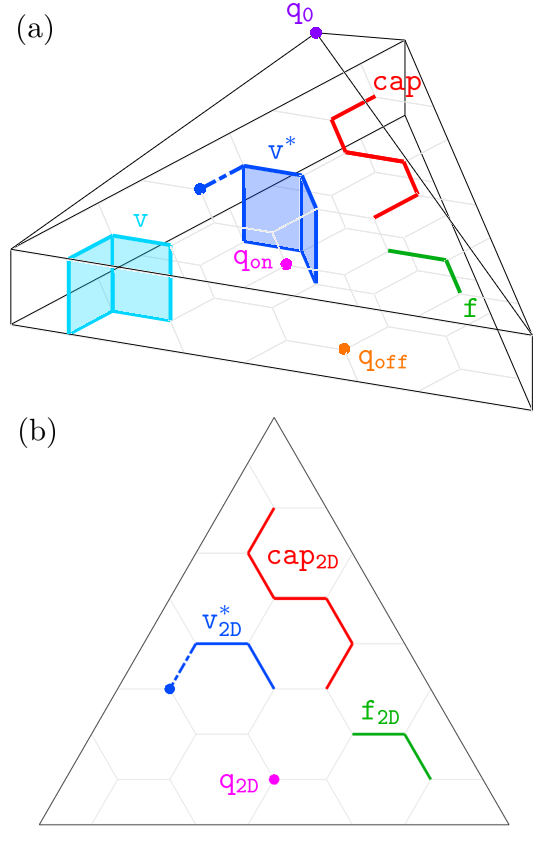}
	\caption{(a) Examples of faults of each type on the 3D structure. (b) Examples of faults of each type on the 2D plane.}
	\label{fig:fault_2D_3D}
\end{figure}

\begin{table*}[tbp]
	\begin{center}
		\begin{tabular}{| c | c | c | c | c | c | c | c | c |}
			\hline
			\multicolumn{2}{| c |}{} & \multicolumn{7}{| c |}{Type of fault}\\
			\cline{3-9}
			\multicolumn{2}{| c |}{} & $\mathtt{q_0}$ & $\mathtt{q_{on}}$ & $\mathtt{q_{off}}$ & $\mathtt{f}$ & $\mathtt{v}$ & $\mathtt{v^*}$ & $\mathtt{cap}$\\
			\hline
			\parbox[t]{3mm}{\multirow{4}{*}{\rotatebox[origin=c]{90}{Syndrome}}} & $s_a$ ($\mathtt{cap}$) & 1 & 1 & 0 & $\mathrm{wp}(\sigma_\mathtt{f})$ & $\mathrm{wp}(\sigma_\mathtt{v})$ & $\mathrm{wp}(\sigma_\mathtt{v^*,cen})$ & $\mathrm{wp}(\sigma_\mathtt{cap})$ \\
			\cline{2-9}
			& $\vec{s}_b$ ($\mathtt{f}$) & 0 & $\vec{q}_\mathtt{on}$ & 0 & $\vec{p}_\mathtt{f}$ & $\vec{p}_\mathtt{v}$ & $\vec{p}_\mathtt{v^*,cen}$ & $\vec{p}_\mathtt{cap}$ \\
			\cline{2-9}
			& \multirow{2}{*}{$\vec{s}_c$ ($\mathtt{v}$)} & \multirow{2}{*}{0} & \multirow{2}{*}{$\vec{q}_\mathtt{on}$} & \multirow{2}{*}{$\vec{q}_\mathtt{off}$} & \multirow{2}{*}{$\vec{p}_\mathtt{f}$} & \multirow{2}{*}{0} & $\vec{p}_\mathtt{v^*,cen}+\vec{p}_\mathtt{v^*,bot}$& \multirow{2}{*}{$\vec{p}_\mathtt{cap}$} \\
			&&&&&&&(or $\vec{q}_\mathtt{v^*}$)& \\
			\hline
			\multicolumn{2}{| c |}{Weight parity} & 1 & 1 & 1 & $\mathrm{wp}(\sigma_\mathtt{f})$ & 0 & 1 & $\mathrm{wp}(\sigma_\mathtt{cap})$ \\
			\hline
			\parbox[t]{3mm}{\multirow{3}{*}{\rotatebox[origin=c]{90}{Flag}}} & $\vec{f}_a$ ($\mathtt{cap}$)& 0 & 0 & 0 & 0 & 0 & 0 & $\vec{f}_\mathtt{cap}$ \\
			\cline{2-9}
			& $\vec{f}_b$ ($\mathtt{f}$)& 0 & 0 & 0 & $\vec{f}_\mathtt{f}$ & 0 & 0 & 0 \\
			\cline{2-9}
			& $\vec{f}_c$ ($\mathtt{v}$)& 0 & 0 & 0 & 0 & $\vec{f}_\mathtt{v}$ & $\vec{f}_\mathtt{v^*}$ & 0 \\
			\hline
		\end{tabular}
	\end{center}
	\caption{Syndrome $\vec{s}=(s_a,\vec{s}_b,\vec{s}_c)$, weight parity, and flag vector $\vec{f}=(\vec{f}_a,\vec{f}_b,\vec{f}_c)$ corresponding to a single fault of each type which leads to a $Z$-type error. $s_a,\vec{s}_b,\vec{s}_c$ are syndromes evaluated by $\mathtt{cap}, \mathtt{f}$ and $\mathtt{v}$ generators of $X$ type, while $\vec{f}_a,\vec{f}_b,\vec{f}_c$ are flag outcomes obtained from circuits for measuring $\mathtt{cap}, \mathtt{f}$ and $\mathtt{v}$ generators of $Z$ type. Note that in some cases, a syndrome bit is equal to the weight parity of an error.}
	\label{tab:fault-syndrome}
\end{table*}

Suppose that a single fault causes a $Z$-type data error $E$ and a flag vector $\vec{f}$. The syndrome of $E$ evaluated by $X$-type generators can be written as $(s_a,\vec{s}_b,\vec{s}_c)$, where $s_a,\vec{s}_b,\vec{s}_c$ are syndromes obtained from measuring $\mathtt{cap},\mathtt{f},$ and $\mathtt{v}$ generators of $X$ type. In addition, the flag vector can be written as $(\vec{f}_a,\vec{f}_b,\vec{f}_c)$, where $\vec{f}_a,\vec{f}_b,\vec{f}_c$ are flag outcomes obtained from circuits for measuring $\mathtt{cap},\mathtt{f},$ and $\mathtt{v}$ generators of $Z$ type, respectively. (The lengths of $s_a,\vec{s}_b,\vec{s}_c$ are equal to the number of generators of each category, while  the lengths of $\vec{f}_a,\vec{f}_b,\vec{f}_c$ are equal to the number of generators of each category times the number of flag ancillas in each flag circuit, assuming that all flag circuits have equal number of flag ancillas.) Let $\mathrm{wp}(\sigma)$ denote the weight parity of error $\sigma$. Due to the general configurations of CNOT gates being used, the weight parity and the syndromes of a $Z$-type error (evaluated by $X$-type generators) and a flag vector arising from each type of faults can be summarized as in \cref{tab:fault-syndrome}. Note that for a $\mathtt{v^*}$ fault, $\sigma_\mathtt{v^*,cen}$ and $\sigma_\mathtt{v^*,bot}$ differ by a $Z$ error on a single qubit; i.e., $\mathrm{wp}(\sigma_\mathtt{v^*,cen})+\mathrm{wp}(\sigma_\mathtt{v^*,bot})=1$. Sometimes we will write $\vec{p}_\mathtt{v^*,cen}+\vec{p}_\mathtt{v^*,bot} = \vec{q}_\mathtt{v^*}$ to emphasize its similarity to the syndrome of a single-qubit error.


Now, let us consider the case that a fault combination arises from multiple faults. The syndrome and the weight parity of the combined error, and the cumulative flag vector of a fault combination can be calculated by adding the syndromes and the flag outcomes of all faults in the fault combination (the addition is modulo 2). For example, suppose that a fault combination consists of 2 faults which are of $\mathtt{q_{on}}$ type and $\mathtt{v}$ type. The syndrome $\vec{s}(\mathbf{E})$ and the weight parity $\mathrm{wp}(\mathbf{E})$ of the combined error $\mathbf{E}$, and the cumulative flag vector $\vec{\mathbf{f}}$ correspond to such a fault combination are,
\begin{align}
	\vec{s}(\mathbf{E}) &= (1+\mathrm{wp}(\sigma_\mathtt{v}),\vec{q}_\mathtt{on}+\vec{p}_\mathtt{v},\vec{q}_\mathtt{on}),\nonumber \\
	\mathrm{wp}(\mathbf{E}) &= 1,\nonumber \\
	\vec{\mathbf{f}} &= (\vec{0},\vec{0},\vec{f}_\mathtt{v}).\nonumber
\end{align}

For a general fault combination composed of multiple faults, the corresponding syndrome, weight parity, and cumulative flag vector can be calculated as follows: let $s_\mathtt{cap},\vec{s}_\mathtt{f},\vec{s}_\mathtt{v}$ denote syndromes of the combined error evaluated by $\mathtt{cap},\mathtt{f},$ and $\mathtt{v}$ generators of $X$ type, let $\mathrm{wp}_\mathrm{tot}$ denote the weight parity, and let $\vec{\mathbf{f}}_\mathtt{cap},\vec{\mathbf{f}}_\mathtt{f},\vec{\mathbf{f}}_\mathtt{v}$ denote parts of the cumulative flag vector obtained from circuits for measuring $\mathtt{cap},\mathtt{f},$ and $\mathtt{v}$ generators of $Z$ type. From \cref{tab:fault-syndrome}, we find that for each fault combination,
\begin{align}
	s_\mathtt{cap} =&\; n_{0}+n_\mathtt{on}+\sum\mathrm{wp}(\sigma_\mathtt{f})+\sum\mathrm{wp}(\sigma_\mathtt{v}) \nonumber \\
	&+\sum\mathrm{wp}(\sigma_\mathtt{v^*,cen})+\sum\mathrm{wp}(\sigma_\mathtt{cap}), \label{eq:main1}\\	
	\vec{s}_\mathtt{f} =& \sum \vec{q}_\mathtt{on} + \sum \vec{p}_\mathtt{f} + \sum \vec{p}_\mathtt{v} + \sum \vec{p}_\mathtt{v^*,cen} \nonumber \\ &+ \sum \vec{p}_\mathtt{cap}, \label{eq:main2}
\end{align}	

\vspace*{-0.8cm}
\begin{align}	
	\vec{s}_\mathtt{v} =& \sum \vec{q}_\mathtt{on}+\sum \vec{q}_\mathtt{off}+\sum \vec{p}_\mathtt{f}+\sum \vec{q}_\mathtt{v^*} \nonumber \\ &+ \sum \vec{p}_\mathtt{cap}, \label{eq:main3}	\\
	\mathrm{wp}_\mathtt{tot} =&\; n_{0}+n_\mathtt{on}+n_\mathtt{off}+\sum \mathrm{wp}(\sigma_\mathtt{f})+n_\mathtt{v^*} \nonumber \\ &+ \sum\mathrm{wp}(\sigma_\mathtt{cap}), \label{eq:main4}\\
	\vec{\mathbf{f}}_\mathtt{cap} =& \sum \vec{f}_\mathtt{cap}, \label{eq:main5}\\
	\vec{\mathbf{f}}_\mathtt{f} =& \sum \vec{f}_\mathtt{f}, \label{eq:main6}\\
	\vec{\mathbf{f}}_\mathtt{v} =& \sum \vec{f}_\mathtt{v}+\sum \vec{f}_\mathtt{v^*}, \label{eq:main7}
\end{align}
where each sum is over the same type of faults (the equations are modulo 2). In addition, adding \cref{eq:main1} to \cref{eq:main4} and adding \cref{eq:main2} to \cref{eq:main3} give the following equations:
\begin{align}
	\mathrm{wp}_\mathtt{bot}=&n_\mathtt{off}+\sum \mathrm{wp}(\sigma_\mathtt{v})+\sum \mathrm{wp}(\sigma_\mathtt{v^*,bot}), \label{eq:main8} \\
	\vec{s}_\mathtt{bot} =& \sum \vec{q}_\mathtt{off}+\sum \vec{p}_\mathtt{v} + \sum \vec{p}_\mathtt{v^*,bot}, \label{eq:main9}
\end{align}
where $\mathrm{wp}_\mathtt{bot}= s_\mathtt{cap}+\mathrm{wp}_\mathtt{tot}$ and $\vec{s}_\mathtt{bot}=\vec{s}_\mathtt{f}+\vec{s}_\mathtt{v}$.


\begin{table*}[tbp]
	\begin{center}
		\begin{tabular}{| c | c | c | c | c | c | c | c |}
			\hline
			\multicolumn{4}{| c |}{2D plane} & \multicolumn{4}{| c |}{3D structure}\\
			\hline
			{\small Fault type} & {\small Syndrome} & {\small Weight parity} & {\small Flag vector} & {\small Fault type} & {\small Syndrome} & {\small Weight parity} & {\small Flag vector} \\
			\hline
			$\mathtt{q_{2D}}$ & $\vec{q}_\mathtt{2D}$ & 1 & - & $\mathtt{q_{on}}$, $\mathtt{q_{off}}$, or $\mathtt{q_{v^*}}$ & $\vec{q}_\mathtt{on}$, $\vec{q}_\mathtt{off}$ or $\vec{q}_\mathtt{v^*}$ & 1 & - \\
			\hline
			\multirow{2}{*}{$\mathtt{f_{2D}}$} & \multirow{2}{*}{$\vec{p}_\mathtt{f_{2D}}$} & \multirow{2}{*}{$\mathrm{wp}(\sigma_\mathtt{f_{2D}})$} & \multirow{2}{*}{$\vec{f}_\mathtt{f_{2D}}$} & \multirow{2}{*}{$\mathtt{f}$, $\mathtt{v}$, or $\mathtt{v^*}$} & $\vec{p}_\mathtt{f}$, $\vec{p}_\mathtt{v}$, $\vec{p}_\mathtt{v^*,cen}$, & $\mathrm{wp}(\sigma_\mathtt{f})$, $\mathrm{wp}(\sigma_\mathtt{v})$, & $\vec{f}_\mathtt{f}$, $\vec{f}_\mathtt{v}$, \\
			& & & & & or $\vec{p}_\mathtt{v^*,bot}$ & $\mathrm{wp}(\sigma_\mathtt{v^*,cen})$, or $\mathrm{wp}(\sigma_\mathtt{v^*,bot})$ & or $\vec{f}_\mathtt{v^*}$ \\
			\hline
			\multirow{2}{*}{$\mathtt{v^*_{2D}}$} & $\vec{p}_\mathtt{v^*_{2D},cen}$ & $\mathrm{wp}(\sigma_\mathtt{v^*_{2D},cen})$ & \multirow{2}{*}{$\vec{f}_\mathtt{v^*_{2D}}$} & \multirow{2}{*}{$\mathtt{v^*}$} & $\vec{p}_\mathtt{v^*,cen}$ & $\mathrm{wp}(\sigma_\mathtt{v^*,cen})$ & \multirow{2}{*}{$\vec{f}_\mathtt{v^*}$} \\
			\cline{2-3} \cline{6-7} 
			& $\vec{p}_\mathtt{v^*_{2D},bot}$ & $\mathrm{wp}(\sigma_\mathtt{v^*_{2D},bot})$ & & & $\vec{p}_\mathtt{v^*,bot}$ & $\mathrm{wp}(\sigma_\mathtt{v^*,bot})$ & \\
			\hline
			$\mathtt{cap_{2D}}$ & $\vec{p}_\mathtt{cap_{2D}}$ & $\mathrm{wp}(\sigma_\mathtt{cap_{2D}})$ & $\vec{f}_\mathtt{cap_{2D}}$ & $\mathtt{cap}$ & $\vec{p}_\mathtt{cap}$ & $\mathrm{wp}(\sigma_\mathtt{cap})$ & $\vec{f}_\mathtt{cap}$ \\
			\hline
		\end{tabular}
	\end{center}
	\caption{The correspondence between the notations for types of faults on the 2D plane and the 3D structure.}
	\label{tab:2D-3D}
\end{table*}

\cref{eq:main1,eq:main2,eq:main3,eq:main4,eq:main5,eq:main6,eq:main7,eq:main8,eq:main9} are the main ingredients for the proof of the main theorem to be developed. One may notice that \cref{eq:main1,eq:main2}, \cref{eq:main3,eq:main4}, and \cref{eq:main8,eq:main9} come in pairs. They have the following physical meanings: suppose that the combined error $\mathbf{E}$ is $\mathbf{E}_\mathtt{0}\cdot\mathbf{E}_\mathtt{on}\cdot\mathbf{E}_\mathtt{off}$ where $\mathbf{E}_\mathtt{0},\mathbf{E}_\mathtt{on},\mathbf{E}_\mathtt{off}$ are the error on $\mathtt{q_0}$, the error on the center plane, and the error on the bottom plane. Then,
\begin{enumerate}
	\item \cref{eq:main2} is the syndrome of $\mathbf{E}_\mathtt{on}$, while \cref{eq:main1} is the weight parity $\mathbf{E}_\mathtt{on}$ plus the weight parity of $\mathbf{E}_\mathtt{0}$.
	\item \cref{eq:main3} is the syndrome of $\mathbf{E}_\mathtt{on}\cdot\mathbf{E}_\mathtt{off}$, while \cref{eq:main4} is the weight parity of $\mathbf{E}_\mathtt{on}\cdot\mathbf{E}_\mathtt{off}$ plus the weight parity of $\mathbf{E}_\mathtt{0}$. (Since $\mathtt{v}$ generators capture errors on both planes simultaneously, $\mathbf{E}_\mathtt{on}\cdot\mathbf{E}_\mathtt{off}$ can be viewed as a remaining error when $\mathbf{E}_\mathtt{on}$ and $\mathbf{E}_\mathtt{off}$ are `projected' on the same plane.)
	\item \cref{eq:main9} is the syndrome of $\mathbf{E}_\mathtt{off}$, while \cref{eq:main8} is the weight parity of $\mathbf{E}_\mathtt{off}$.
\end{enumerate}
From these pairs of equations, and from the fact that now we only have to specify the ordering of data CNOTs for each $\mathtt{f}$ generator, the ordering of data CNOTs after the first gate for the $\mathtt{cap}$ generator, and the ordering of flag CNOTs for each flag circuit, we can now simplify the CNOT ordering finding problem for a 3D structure to the problem of finding CNOT orderings on a 2D plane (which is similar to the 2D color code of distance $d$). In particular, each pair of equations concern errors on a 2D plane (the center, the bottom, or the projected plane). We will try to find conditions for the CNOT orderings on a 2D plane such that if satisfied, a bad case which makes $\mathcal{F}_t$ indistinguishable cannot happen.


Some types of faults on the 3D structure can be considered as the same type of faults when the problem is simplified. The followings are types of possible single faults on the 2D plane and their correspondence on the 3D structure:
\begin{enumerate}
	\item Type $\mathtt{q_{2D}}$: a fault causing a single-qubit $Z$-type error on the 2D plane which does not arise from any $Z$-type generator measurement. The syndrome of an error is denoted by $\vec{q}_\mathtt{2D}$. The total number of $\mathtt{q_{2D}}$ faults is $n_\mathtt{q_{2D}}$. The combined error from only $\mathtt{q_{2D}}$ faults is denoted by $\mathbf{E}_\mathtt{q_{2D}}$. This type of faults corresponds to $\mathtt{q_{on}}$ and $\mathtt{q_{off}}$ faults on the 3D structure.
	\item Type $\mathtt{f_{2D}}$: a fault occurred during a measurement of a $\mathtt{f}$ generator of $Z$ type. A $Z$-type error from each fault of this type and its syndrome are denoted by $\sigma_\mathtt{f_{2D}}$ and $\vec{p}_\mathtt{f_{2D}}$. A flag vector corresponding to each fault of this type is denoted by $\vec{f}_\mathtt{f_{2D}}$. The total number of $\mathtt{f_{2D}}$ faults is $n_\mathtt{f_{2D}}$. The combined error from only $\mathtt{f_{2D}}$ faults is denoted by $\mathbf{E}_\mathtt{f_{2D}}$. This type of faults corresponds to $\mathtt{f}$ and $\mathtt{v}$ faults on the 3D structure (since an error on the center plane and an error on the bottom plane from a $\mathtt{v}$ fault have the same form; see an example in \cref{fig:v_vstar}).
	\item Type $\mathtt{v^*_{2D}}$: a fault occurred during a measurement of a $\mathtt{v}$ generator of $Z$ type in which an error occurred on the center plane and an error on the bottom plane are different (see an example in \cref{fig:v_vstar}). A part of a $Z$-type error from each fault of this type occurred on the center plane only and its syndrome are denoted by $\sigma_\mathtt{v^*_{2D},cen}$ and $\vec{p}_\mathtt{v^*_{2D},cen}$. The other part of the $Z$-type error that occurred on the bottom plane only and its syndrome are denoted by $\sigma_\mathtt{v^*_{2D},bot}$ and $\vec{p}_\mathtt{v^*_{2D},bot}$. A flag vector corresponding to each fault of this type is denoted by $\vec{f}_\mathtt{v^*_{2D}}$. The total number of $\mathtt{v^*_{2D}}$ faults is $n_\mathtt{v^*_{2D}}$. The part of the combined error from only $\mathtt{v^*_{2D}}$ faults on the center plane and the part on the bottom plane are denoted by $\mathbf{E}_\mathtt{v^*_{2D},cen}$ and $\mathbf{E}_\mathtt{v^*_{2D},bot}$. This type of faults corresponds to $\mathtt{v^*}$ faults on the 3D structure. (Note that this is the only type of faults which cannot be represented completely on the 2D plane since the error on the center plane and the error on the bottom plane are different. However, when running a computer simulation, we can treat a fault of $\mathtt{v^*_{2D}}$ type similarly to a fault of $\mathtt{f_{2D}}$ type except that two values of errors will be assigned to each fault.)
	\item Type $\mathtt{cap_{2D}}$: a fault occurred during a measurement of a $\mathtt{cap}$ generator of $Z$ type. A $Z$-type error from each fault of this type and its syndrome are denoted by $\sigma_\mathtt{cap_{2D}}$ and $\vec{p}_\mathtt{cap_{2D}}$ ($\sigma_\mathtt{cap_{2D}}$ is always on the center plane up to a multiplication of the $\mathtt{cap}$ generator being measured). A flag vector corresponding to each fault of this type is denoted by $\vec{f}_\mathtt{cap_{2D}}$. The total number of $\mathtt{cap_{2D}}$ faults is $n_\mathtt{cap_{2D}}$. The combined error from only $\mathtt{cap_{2D}}$ faults is denoted by $\mathbf{E}_\mathtt{cap_{2D}}$. This type of faults corresponds to $\mathtt{cap}$ faults on the 3D structure.
\end{enumerate}
Examples of faults of each type on the 2D plane are illustrated in \cref{fig:fault_2D_3D}\hyperlink{target:fault_2D_3D}{b}. The correspondence between the notations for types of faults on the 2D plane and the 3D structure can be summarized in \cref{tab:2D-3D}.



We can see that possible $Z$-type errors on the 2D plane depend on the CNOT orderings for measuring $\mathtt{f}$ and $\mathtt{cap}$ generators of $Z$ type. Next, we will state the sufficient conditions for the CNOT orderings on the 2D plane which will make $\mathcal{F}_t$ (which concerns fault combinations from the 3D structure) distinguishable. These sufficient conditions are introduced in order to prevent the case that can lead to an `indistinguishable' pair (a pair of fault combinations from the 3D structure which does not satisfy any condition in \cref{def:distinguishable}).


First, we will state a condition which is automatically satisfied if a code being considered on the 2D plane is a code of distance $d$ to which \cref{lem:err_equivalence} is applicable: 
\begin{condition}{0}
	For any fault combination on the 2D plane which satisfies $n_\mathtt{q_{2D}} \leq d-1$, $\mathbf{E}_\mathtt{q_{2D}}$ is not a nontrivial logical operator; equivalently, at least one of the followings is satisfied:
	\begin{enumerate}
		\item $\sum \vec{q}_\mathtt{2D} \neq 0 \mod 2$, or
		\item $n_\mathtt{q_{2D}} \neq 1 \mod 2$.
	\end{enumerate}
	\label{con:con0}%
	\vspace*{-0.3cm}
\end{condition}
Note that a nontrivial logical operator is an error corresponding to the trivial syndrome whose weight parity is odd (from \cref{lem:err_equivalence}). \cref{con:con0} is equivalent to the fact that an error of weight $\leq d-1$ is detectable by a code of distance $d$; i.e., it either has a nontrivial syndrome or is a stabilizer. We state \cref{con:con0} explicitly (although it is automatically satisfied) because the condition in this form looks similar to other conditions, which will simplify the proof of the main theorem. 

Next, we will state five sufficient conditions for the CNOT orderings on the 2D plane which will make $\mathcal{F}_t$ distinguishable. The conditions are as follows: 
\begin{condition}{1}
	For any fault combination on the 2D plane which satisfies $n_\mathtt{f_{2D}} \leq d-2$, $\mathbf{E}_\mathtt{f_{2D}}$ is not a nontrivial logical operator or the cumulative flag vector is not zero; equivalently, at least one of the followings is satisfied:
	\begin{enumerate}
		\item $\sum \vec{p}_\mathtt{f_{2D}} \neq 0 \mod 2$, or
		\item $\sum \mathrm{wp}(\sigma_\mathtt{f_{2D}})\neq 1 \mod 2$, or
		\item $\sum \vec{f}_\mathtt{f_{2D}} \neq 0 \mod 2$.
	\end{enumerate}
	\label{con:con1}%
	\vspace*{-0.3cm}
\end{condition}
\begin{condition}{2}
	For any fault combination on the 2D plane which satisfies $n_\mathtt{q_{2D}}+n_\mathtt{f_{2D}} \leq d-3$, $\mathbf{E}_\mathtt{q_{2D}}\cdot\mathbf{E}_\mathtt{f_{2D}}$ is not a nontrivial logical operator or the cumulative flag vector is not zero; equivalently, at least one of the followings is satisfied:
	\begin{enumerate}
		\item $\sum \vec{q}_\mathtt{2D}+\sum \vec{p}_\mathtt{f_{2D}} \neq 0 \mod 2$, or
		\item $n_\mathtt{q_{2D}}+\sum \mathrm{wp}(\sigma_\mathtt{f_{2D}})\neq 1 \mod 2$, or
		\item $\sum \vec{f}_\mathtt{f_{2D}} \neq 0 \mod 2$.
	\end{enumerate}
	\label{con:con2}%
	\vspace*{-0.3cm}
\end{condition}

\begin{condition}{3}
	For any fault combination on the 2D plane which satisfies $n_\mathtt{f_{2D}}=1$ and $n_\mathtt{q_{2D}}+n_\mathtt{f_{2D}} \leq d-2$, $\mathbf{E}_\mathtt{q_{2D}}\cdot\mathbf{E}_\mathtt{f_{2D}}$ is not a nontrivial logical operator or the cumulative flag vector is not zero; equivalently, at least one of the followings is satisfied:
	\begin{enumerate}
		\item $\sum \vec{q}_\mathtt{2D}+\sum \vec{p}_\mathtt{f_{2D}} \neq 0 \mod 2$, or
		\item $n_\mathtt{q_{2D}}+\sum \mathrm{wp}(\sigma_\mathtt{f_{2D}})\neq 1 \mod 2$, or
		\item $\sum \vec{f}_\mathtt{f_{2D}} \neq 0 \mod 2$.
	\end{enumerate}
	\label{con:con3}%
\end{condition}
\begin{condition}{4}
	For any fault combination on the 2D plane which satisfies $n_\mathtt{f_{2D}}=1$, $n_\mathtt{q_{2D}}\geq1$, $n_\mathtt{v^*_{2D}}\geq 2$, and $n_\mathtt{q_{2D}}+n_\mathtt{f_{2D}}+n_\mathtt{v^*_{2D}} = d-1$, the following does not happen: $\mathbf{E}_\mathtt{f_{2D}}\cdot\mathbf{E}_\mathtt{v^*_{2D},cen}$ is a stabilizer, and $\mathbf{E}_\mathtt{q_{2D}}\cdot\mathbf{E}_\mathtt{v^*_{2D},bot}$ is a nontrivial logical operator, and the cumulative flag vector is zero; equivalently, at least one of the followings is satisfied:
	\begin{enumerate}
		\item $\sum \vec{p}_\mathtt{f_{2D}} + \sum \vec{p}_\mathtt{v^*_{2D},cen} \neq 0 \mod 2$, or
		\item $\sum \mathrm{wp}(\sigma_\mathtt{f_{2D}})+\sum\mathrm{wp}(\sigma_\mathtt{v^*_{2D},cen}) \neq 0 \mod 2$, or
		\item $\sum \vec{q}_\mathtt{2D} + \sum \vec{p}_\mathtt{v^*_{2D},bot} \neq 0 \mod 2$, or
		\item $n_\mathtt{q_{2D}}+\sum\mathrm{wp}(\sigma_\mathtt{v^*_{2D},bot}) \neq 1 \mod 2$, or
		\item $\sum \vec{f}_\mathtt{f_{2D}} \neq 0 \mod 2$, or
		\item $\sum \vec{f}_\mathtt{v^*_{2D}} \neq 0 \mod 2$.
	\end{enumerate}
	\label{con:con4}%
\end{condition}
\begin{condition}{5}
	For any fault combination on the 2D plane which satisfies $n_\mathtt{cap_{2D}}=1$, $n_\mathtt{q_{2D}}\geq1$, $n_\mathtt{f_{2D}}+n_\mathtt{v^*_{2D}}\geq 2$, and $n_\mathtt{q_{2D}}+n_\mathtt{f_{2D}}+n_\mathtt{v^*_{2D}}+n_\mathtt{cap_{2D}} = d-1$, the following does not happen: $\mathbf{E}_\mathtt{f_{2D}}\cdot\mathbf{E}_\mathtt{v^*_{2D},cen}\cdot \mathbf{E}_\mathtt{cap_{2D}}$ is a stabilizer, and $\mathbf{E}_\mathtt{q_{2D}}\cdot\mathbf{E}_\mathtt{f_{2D}}\cdot\mathbf{E}_\mathtt{v^*_{2D},bot}$ is a nontrivial logical operator, and the cumulative flag vector is zero; equivalently, at least one of the followings is satisfied:
	\begin{enumerate}
		\item $\sum \vec{p}_\mathtt{f_{2D}} + \sum \vec{p}_\mathtt{v^*_{2D},cen} + \sum \vec{p}_\mathtt{cap_{2D}} \neq 0 \mod 2$, or
		\item $\sum \mathrm{wp}(\sigma_\mathtt{f_{2D}})+\sum\mathrm{wp}(\sigma_\mathtt{v^*_{2D},cen}) +\sum \mathrm{wp}(\sigma_\mathtt{cap_{2D}}) \neq 0 \mod 2$, or
		\item $\sum \vec{q}_\mathtt{2D} + \sum \vec{p}_\mathtt{f_{2D}} + \sum \vec{p}_\mathtt{v^*_{2D},bot} \neq 0 \mod 2$, or
		\item $n_\mathtt{q_{2D}}+\sum\mathrm{wp}(\sigma_\mathtt{f_{2D}}) +\sum\mathrm{wp}(\sigma_\mathtt{v^*_{2D},bot}) \neq 1 \mod 2$, or
		\item $\sum \vec{f}_\mathtt{f_{2D}}+\vec{f}_\mathtt{v^*_{2D}}\neq 0 \mod 2$, or
		\item $\sum \vec{f}_\mathtt{cap_{2D}} \neq 0 \mod 2$.
	\end{enumerate}
	\label{con:con5}%
\end{condition}


Conditions \ref{con:con1} to \ref{con:con5} prevent fault combinations of some form from occurring on the 2D plane (such fault combinations can lead to an indistinguishable fault set). If we arrange the CNOT gates in the circuits for $\mathtt{f}$ and $\mathtt{cap}$ generators so that all conditions are satisfied, then a fault set $\mathcal{F}_t$ (which considers the 3D structure) will be distinguishable. The main theorem of this work is as follows:
\begin{theorem}
	Let $\mathcal{F}_t$ be the fault set corresponding to circuits for measuring $\mathtt{f}, \mathtt{v}$, and $\mathtt{cap}$ generators of the capped color code in H form constructed from $CCC(d)$ (where $t=(d-1)/2$, $d=3,5,7,...$), and suppose that the general configurations of CNOT gates for $\mathtt{f}$, $\mathtt{v}$, and $\mathtt{cap}$ generators are imposed, and the circuits for each pair of $X$-type and $Z$-type generators use the same CNOT ordering. Let the code on the (simplified) 2D plane be the 2D color code of distance $d$. If all possible fault combinations on the 2D plane arising from the circuits for measuring $\mathtt{f}$ and $\mathtt{cap}$ generators satisfy Conditions \ref{con:con1} to \ref{con:con5}, then $\mathcal{F}_t$ is distinguishable.
	\label{thm:main}
\end{theorem}


\textit{Proof ideas.} \cref{thm:main} is proved in \cref{app:proof_main_thm}. The proof is organized as follows: First, we try to show that if Conditions \ref{con:con1} to \ref{con:con5} are satisfied, then for any fault combination arising from up to $d-1$ faults whose combined error is purely $Z$ type, the fault combination cannot lead to a logical $Z$ operator and the zero cumulative flag vector. The same analysis is also applicable to fault combinations whose combined error is purely $X$ type since the circuits for measuring each pair of $X$-type and $Z$-type generators are of the same form. Afterwards, we use the fact that $i$ faults during the measurements of $Z$-type generators cannot cause an $X$-type error of weight more than $i$ (and vice versa), and show that there is no fault combination arising from up to $d-1$ faults which leads to a nontrivial logical operator and the zero cumulative flag vector. By \cref{prop:2t}, this implies that $\mathcal{F}_t$ is distinguishable.

In order to prove the first part, we will assume that Conditions \ref{con:con1} to \ref{con:con5} are satisfied and there exists a fault combination arising from $<d$ faults whose combined error is a logical $Z$ operator and its cumulative flag vector is zero, then show that some contradiction will happen. From \cref{lem:err_equivalence}, a logical $Z$ operator is a $Z$-type error with trivial syndrome and odd weight parity. Therefore, such a fault combination will give $s_\mathtt{cap}=0$, $\vec{s}_\mathtt{f}=\vec{0}$, $\vec{s}_\mathtt{v}=\vec{0}$, $\mathrm{wp}_\mathtt{tot}=1$, $\vec{\mathbf{f}}_\mathtt{cap}=\vec{0}$, $\vec{\mathbf{f}}_\mathtt{f}=\vec{0}$, $\vec{\mathbf{f}}_\mathtt{v}=\vec{0}$, $\mathrm{wp}_\mathtt{bot}=1$, and $\vec{s}_\mathtt{bot}=0$ in the main equations (\cref{eq:main1,eq:main2,eq:main3,eq:main4,eq:main5,eq:main6,eq:main7,eq:main8,eq:main9}). A proof for this part will be divided into 4 cases: (1) $n_\mathtt{f}=0$ and $n_\mathtt{cap}=0$, (2) $n_\mathtt{f} \geq 1$ and $n_\mathtt{cap}=0$, (3) $n_\mathtt{f} = 0$ and $n_\mathtt{cap}\geq 1$, and (4) $n_\mathtt{f} \geq 1$ and $n_\mathtt{cap}\geq 1$. In each case, the main equations will be simplified by eliminating the terms which are equal to zero. Afterwards, We will consider the following pairs of equations: \cref{eq:main1} and \cref{eq:main2}, \cref{eq:main3} and \cref{eq:main4}, \cref{eq:main8} and \cref{eq:main9}. For each pair, the types of faults on the 3D structure will be translated to their corresponding types of faults on the 2D plane in order to find matching conditions from Conditions \ref{con:con1} to \ref{con:con5}. Note that the total number of faults of each type will also help in finding the matching conditions, and the total number of faults of all types is at most $d-1$. When the matching conditions are found, we will find that some contradictions will happen (assuming that all conditions are satisfied), and this is true for all possible cases. \hfill $\square$

\cref{thm:main} can make the process of finding CNOT orderings which give a distinguishable fault set less laborious; instead of finding all possible fault combinations arising from the circuits for $\mathtt{f}$, $\mathtt{v}$, and $\mathtt{cap}$ generators and check whether any condition in \cref{def:distinguishable} is satisfied, we just have to check whether all possible fault combinations arising from the circuits for $\mathtt{f}$ and $\mathtt{cap}$ generators satisfy Conditions \ref{con:con1} to \ref{con:con5}. Note that number of possible fault combinations of the latter task is much smaller than that of the prior task because the total number of generators involved in the latter calculation roughly decreases by half, and the weight of an $\mathtt{f}$ generator is half of the weight of its corresponding $\mathtt{v}$ generator. After good CNOT orderings for $\mathtt{f}$ and $\mathtt{cap}$ generators are found, we can find the CNOT orderings of $\mathtt{v}$ generators by the constraints imposed by the general configurations for data and flag CNOTs.\\

\begin{figure}[btp]
	\centering
	\hypertarget{target:circuit_CCC}{}
	\includegraphics[width=0.42\textwidth]{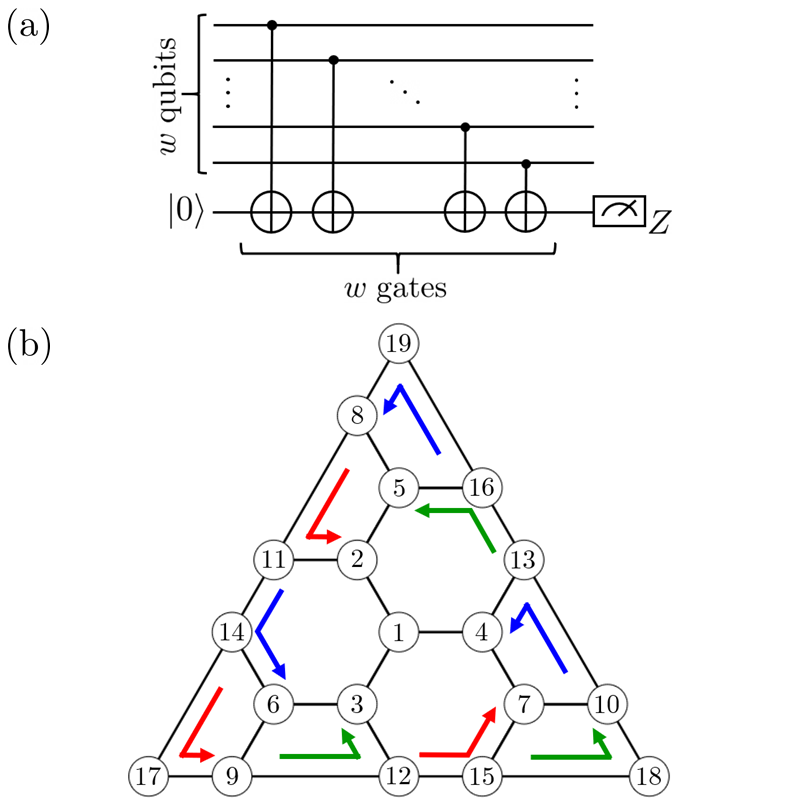}
	\caption{(a) A non-flag circuit for measuring a generator of the capped color code of distance 5 in H form, where $w$ is the weight of the generator. (b) The orderings of data CNOT gates which give a distinguishable fault set $\mathcal{F}_2$.}
	\label{fig:circuit_CCC}
\end{figure}	

\noindent\textbf{Non-flag circuits for measuring generators of capped color codes in H form of distance 3 and 5}



In the case that all circuits for measuring generators are \emph{non-flag circuits}, we can find good CNOT orderings (which give a distinguishable fault set) for the capped color code in H form of distance 3 and 5. The circuits and CNOT orderings for the code of distance 3 (which is the 3D color code of distance 3) are previously described in \cref{subsec:3D_code_config}. The circuit for measuring a generator of weight $w$ of the code of distance 5 is a non-flag circuit as shown in \cref{fig:circuit_CCC}\hyperlink{target:circuit_CCC}{a}, and the orderings of data CNOTs for $\mathtt{f}$ and $\mathtt{cap}$ generators are presented by the diagram in \cref{fig:circuit_CCC}\hyperlink{target:circuit_CCC}{b}. The meanings of the diagram are as follows: for each $\mathtt{f}$ generator, the qubits on which data CNOTs act start from the tail of an arrow then proceed counterclockwise, and the ordering of data CNOTs for the $\mathtt{cap}$ generator is in numerical order, i.e., (0,1,2,...,19), following the qubit labels in the diagram. Meanwhile, the ordering of data CNOTs for each $\mathtt{v}$ generator can be obtained from its corresponding $\mathtt{f}$ generator using the sawtooth configuration (see \cref{subsec:3D_code_config} and \cref{fig:diagram_3D} for more details).


The aforementioned results for the codes of distance 3 and 5 are found by manually picking the CNOT ordering for each $\mathtt{f}$ or $\mathtt{cap}$ generator, then using a computer simulation to verify that Conditions \ref{con:con1} to \ref{con:con5} are satisfied. However, searching for good CNOT orderings using this procedure might not be efficient when $d$ is large. We point out that in the case that all circuits for measuring generators are non-flag circuits, it is still not known whether good CNOT orderings exist for $d \geq 7$. Fortunately, we can prove analytically that if all circuits for measuring generators are \emph{flag circuits} of a particular form, it is always possible to obtain a distinguishable fault set for a capped color code in H form of \emph{any distance}.\\


	

\noindent\textbf{Flag circuits for measuring generators of a capped color code in H form of any distance}

Here we will show that there exist flag circuits for measuring generators of a capped color code in H form of any distance which can give a distinguishable fault set. First, assume that the circuit for measuring an $\mathtt{f}$ and a $\mathtt{cap}$ generator of weight $w$ is a flag circuit with one flag ancilla similar to the circuit in \cref{subfig:circuit_CCC_f}, and the circuit for measuring a $\mathtt{v}$ generator is a flag circuit with one flag ancilla similar to the circuit in \cref{subfig:circuit_CCC_v} (which follows the general configurations of data and flag CNOTs). 

Next, let us consider \cref{eq:main1,eq:main2,eq:main3,eq:main4,eq:main5,eq:main6,eq:main7,eq:main8,eq:main9}. A nontrivial logical operator of a capped color code in H form with trivial flags happens whenever $s_\mathtt{cap}=0$, $\vec{s}_\mathtt{f}=\vec{0}$, $\vec{s}_\mathtt{v}=\vec{0}$, $\mathrm{wp}_\mathtt{tot}=1$, $\vec{\mathbf{f}}_\mathtt{cap}=\vec{0}$, $\vec{\mathbf{f}}_\mathtt{f}=\vec{0}$, $\vec{\mathbf{f}}_\mathtt{v}=\vec{0}$, $\mathrm{wp}_\mathtt{bot}=1$, and $\vec{s}_\mathtt{bot}=0$. This means that a nontrivial logical operator of a capped color code in H form (constructed from $CCC(d)$) occurs if and only if (1) the combined data error on the bottom plane ($\mathbf{E}_\mathtt{off}$) is a nontrivial logical operator of the 2D color code of distance $d$ with trivial flags, and either (2.a) $n_{0}=0$ and the combined data error on the center plane ($\mathbf{E}_\mathtt{on}$) is a stabilizer of the 2D color code of distance $d$ with trivial flags, or (2.b) $n_{0}=1$ and the combined data error on the center plane ($\mathbf{E}_\mathtt{on}$) is a nontrivial logical operator of the 2D color code of distance $d$ with trivial flags. For this reason, if we can show that there is no fault combination from up to $d-1$ faults that can cause a nontrivial logical operator of the 2D color code of distance $d$ with trivial flags on the bottom plane, then a nontrivial logical operator of the capped color code in H form (constructed from $CCC(d)$) with trivial flags cannot happen, meaning that the fault set $\mathcal{F}_t$ is distinguishable.

\begin{figure}[tbp]
	\centering
	\begin{subfigure}{0.33\textwidth}
		\includegraphics[width=\textwidth]{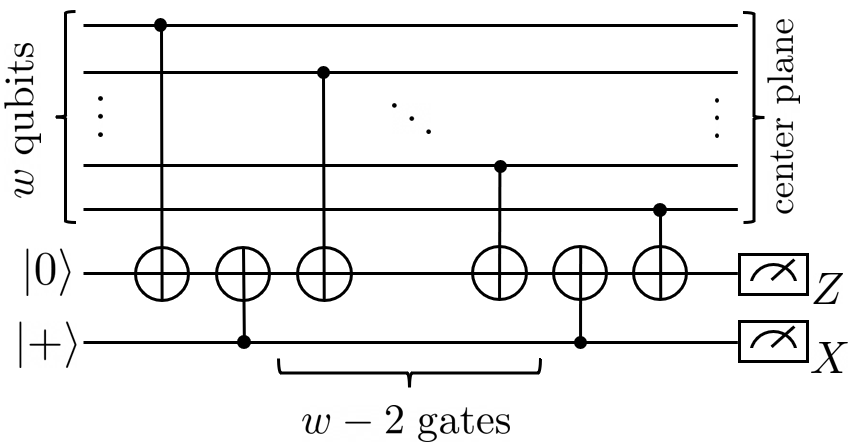}
		\captionsetup{justification=centering}
		\caption{}
		\label{subfig:circuit_CCC_f}
	\end{subfigure}	
	\begin{subfigure}{0.40\textwidth}
		\includegraphics[width=\textwidth]{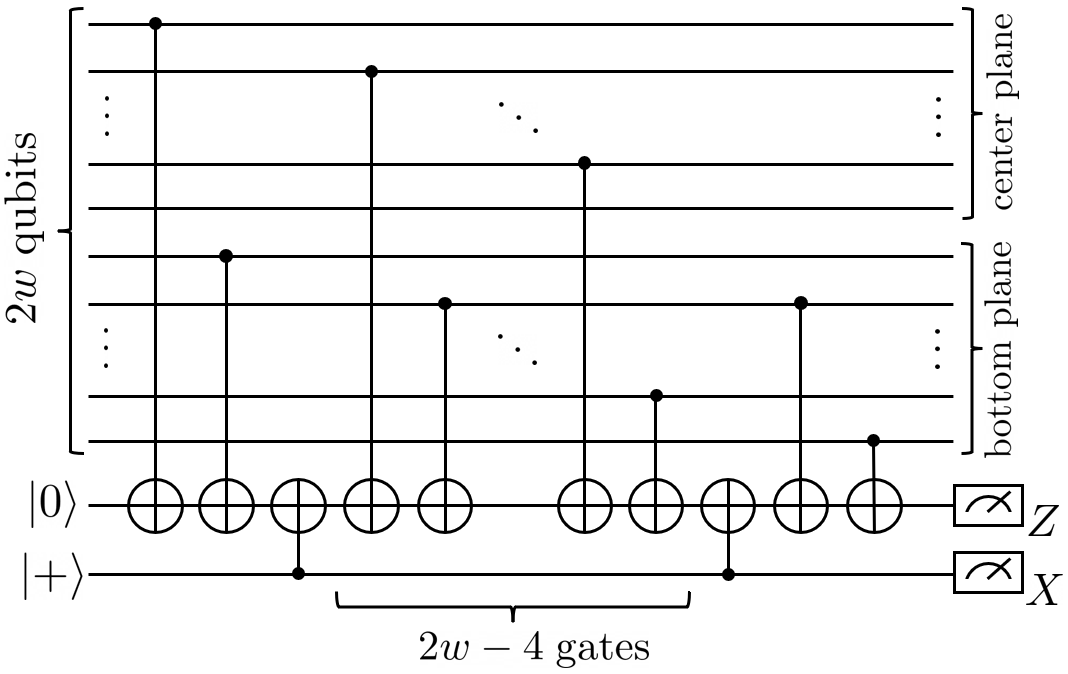}
		\captionsetup{justification=centering}
		\caption{}
		\label{subfig:circuit_CCC_v}
	\end{subfigure}
	\caption{(a) A flag circuit with one flag ancilla for measuring an $\mathtt{f}$ or a $\mathtt{cap}$ generator of weight $w$. (b) A flag circuit with one flag ancilla for measuring a $\mathtt{v}$ generator of weight $2w$.}
	\label{fig:circuit_CCC_flag}
\end{figure}

Observe that faults that can contribute to $\mathbf{E}_\mathtt{off}$ are $\mathtt{q_{off}}$, $\mathtt{v}$, and $\mathtt{v^*}$ faults only. Moreover, from the flag circuit for a $\mathtt{v}$ generator in \cref{subfig:circuit_CCC_v}, a single fault of $\mathtt{v}$ or $\mathtt{v^*}$ type will give a trivial flags only when the part of the corresponding data error on the bottom plane has weight $\leq 1$. This fact leads to the following claim:
\begin{claim}
Suppose that $\mathtt{v}$ generators are measured using flag circuits with one flag ancilla similar to the circuit in \cref{subfig:circuit_CCC_v}.
\begin{enumerate}
\item If there is exactly one fault during a measurement of generator $v_i^z$ and the bit of the flag vector corresponding to $v_i^z$ is zero, then the data error on the bottom plane has weight 0 or 1. In this case, the data error on the bottom plane from one fault of $\mathtt{v}$ (or $\mathtt{v^*}$) type is similar to some data error from 0 or 1 fault of $\mathtt{q_{off}}$ type. 
\item If there are exactly two faults during measurements of the same generator $v_i^z$ (possibly on different rounds) and the bit of the cumulative flag vector corresponding to $v_i^z$ is zero, then the combined data error on the bottom plane has weight 0, 1, 2, or 3 (up to a multiplication of $v_i^z$). The combined data error of weight 0, 1, or 2 on the bottom plane from two faults of $\mathtt{v}$ (or $\mathtt{v^*}$) type on the same generator is similar to some combined data error from 0, 1, or 2 faults of $\mathtt{q_{off}}$ type. The case that the combined data error on the bottom plane of weight 3 arising from two faults of $\mathtt{v}$ (or $\mathtt{v^*}$) type on the same generator is the only case that the weight of the combined data error on the bottom plane is greater than the number of faults.
\item If there are three or more faults during measurements of the same generator $v_i^z$ (possibly from different rounds) and the bit of the cumulative flag vector corresponding to $v_i^z$ is zero, then the combined data error on the bottom plane has weight 0, 1, 2, or 3 (up to a multiplication of $v_i^z$) and is similar to some combined data error from 0, 1, 2, or 3 faults of $\mathtt{q_{off}}$ type.
\end{enumerate}
\label{claim:v_meas_flag}
\end{claim}

\cref{claim:v_meas_flag} will be later used to prove that a nontrivial logical operator of the 2D color code of distance $d$ with trivial flags cannot happen on the bottom plane.
 
Because the ordering of CNOT gates for each $\mathtt{v}$ generator is related to its corresponding $\mathtt{f}$ generator, the problem of finding CNOT orderings for a 3D structure which give a distinguishable fault set can be simplified to the problem of finding CNOT orderings on a 2D plane. In particular, since we are now considering the bottom plane only, $\mathtt{f_{2D}}$ faults on the 2D plane correspond to both $\mathtt{v}$ and $\mathtt{v^*}$ faults on the 3D structure, while $\mathtt{q_{2D}}$ faults on the 2D plane correspond to $\mathtt{q_{off}}$ faults on the 3D structure. 

A fault set $\mathcal{F}_t$ is distinguishable if the following condition is satisfied:
\begin{condition}{6}
	For any fault combination on the 2D plane which satisfies $n_\mathtt{q_{2D}}+n_\mathtt{f_{2D}} \leq d-1$, $\mathbf{E}_\mathtt{q_{2D}}\cdot\mathbf{E}_\mathtt{f_{2D}}$ is not a nontrivial logical operator or the cumulative flag vector is not zero; equivalently, at least one of the followings is satisfied:
	\begin{enumerate}
		\item $\sum \vec{q}_\mathtt{2D}+\sum \vec{p}_\mathtt{f_{2D}} \neq 0 \mod 2$, or
		\item $n_\mathtt{q_{2D}}+\sum \mathrm{wp}(\sigma_\mathtt{f_{2D}})\neq 1 \mod 2$, or
		\item $\sum \vec{f}_\mathtt{f_{2D}} \neq 0 \mod 2$.
	\end{enumerate}
	\label{con:con6}%
\end{condition}

Surprisingly, using the flag circuits with one flag ancilla as shown in \cref{subfig:circuit_CCC_f} and \cref{subfig:circuit_CCC_v} to measure the generators of a capped color code in H form, Condition \ref{con:con6} is satisfied regardless of the orderings of data CNOT gates of $\mathtt{f}$ generators (as long as the CNOT orderings of $\mathtt{v}$ generators follow the general configurations of data and flag CNOTs). And because we are considering faults on the (simplified) 2D plane, the fact that Condition \ref{con:con6} is satisfied regardless of the orderings of data CNOT gates in the flag circuits is also applicable to a 2D color code of any distance as well. This can be restated in the following theorem:

\begin{theorem}
Suppose that the generators of a 2D color code of distance $d$ are measured using the flag circuits with one flag ancilla as displayed in \cref{subfig:circuit_CCC_f}. Then, there is no fault combination arising from $d-1$ faults whose combined data error is a nontrivial logical operator and the cumulative flag vector is zero (i.e., Condition \ref{con:con6} is satisfied), regardless of the orderings of data CNOT gates in the flag circuits.
	\label{thm:main2}
\end{theorem}

\cref{thm:main2} has been proved in \cite{CKYZ20}, where the circuit in \cref{subfig:circuit_CCC_f} is a 1-flag circuit according to the definition in \cite{CB18}. Here we also provide an alternative proof of \cref{thm:main2} which is tailored to the notations being used throughout this work, so that the paper becomes self-contained. We also believe that our proof technique using the relationship between faults and error weights would be useful for finding proper CNOT orderings for other families of codes.

\begin{proof}
Assume by contradiction that Condition \ref{con:con6} is not satisfied; i.e., there exists a fault combination from $d-1$ faults which gives a nontrivial logical operator with trivial flags. For such a fault combination, the syndrome of $\mathbf{E}_\mathtt{q_{2D}}\cdot\mathbf{E}_\mathtt{f_{2D}}$ is zero, the total weight of $\mathbf{E}_\mathtt{q_{2D}}\cdot\mathbf{E}_\mathtt{f_{2D}}$ is odd, and the cumulative flag vector $\sum \vec{f}_\mathtt{f_{2D}}$ is zero. From the structure of the flag circuit in \cref{subfig:circuit_CCC_f}, a single fault of $\mathtt{f_{2D}}$ type will give a trivial flags only when the corresponding data error has weight $\leq 1$. Similar to \cref{claim:v_meas_flag} for faults of $\mathtt{v}$ and $\mathtt{v^*}$ type discussed previously, the only case that faults of $\mathtt{f_{2D}}$ type cannot be considered as faults of $\mathtt{q_{2D}}$ type of the same or smaller number is the case that for each generator $f_i^z$ of the 2D color code, there are exactly two faults during the generator measurements (on the same or different rounds) which lead to the combined data error of weight 3 (up to a multiplication of $f_i^z$). For this reason, we will assume that for each generator $f_i^z$, there are either no faults or exactly two faults during the measurements.

Let $(n_f,n_q)$ denote the case that a fault combination arises from \emph{exactly} $n_f$ faults of $\mathtt{f_{2D}}$ type and \emph{no more than} $n_q$ faults of $\mathtt{q_{2D}}$ type (where $n_f+n_q=d-1$). We will show that in any case with even $n_f$ (i.e., $(0,d-1),(2,d-3),\dots,(d-1,0)$), $\mathbf{E}_\mathtt{q_{2D}}\cdot\mathbf{E}_\mathtt{f_{2D}}$ cannot be a nontrivial logical operator. 
 
\textit{Case $(0,d-1)$}: Because the 2D color code has distance $d$ and the total weight of $\mathbf{E}_\mathtt{q_{2D}}$ is at most $d-1$, $\mathbf{E}_\mathtt{q_{2D}}$ cannot be a nontrivial logical operator.

\textit{Case $(2,d-3)$}: Suppose that a pair of $\mathtt{f_{2D}}$ faults causes a weight-3 error on the supporting qubits of generator $f_i^z$. Consider the following cases:
\begin{enumerate}
	\item If there are even number of $\mathtt{q_{2D}}$ faults on the supporting qubits of $f_i^z$, then the syndrome bit $s_i^x$ corresponding to generator $f_i^x$ is not zero. That is, $\mathbf{E}_\mathtt{q_{2D}}\cdot\mathbf{E}_\mathtt{f_{2D}}$ is not a nontrivial logical operator.
	\item If there are odd number of $\mathtt{q_{2D}}$ faults on the supporting qubits of $f_i^z$, then the total weight of the error on supporting qubits of $f_i^z$ is 0 or 2 (the total weight is even and no more than 3 up to a multiplication of $f_i^z$). Since two $\mathtt{f_{2D}}$ faults and one or more $\mathtt{q_{2D}}$ fault give an error of weight no more than 2, this case is covered by the $(0,d-1)$ case, in which a nontrivial logical operator cannot occur.
\end{enumerate}
Thus, $\mathbf{E}_\mathtt{q_{2D}}\cdot\mathbf{E}_\mathtt{f_{2D}}$ is not a nontrivial logical operator in the $(2,d-3)$ case.

\textit{Case $(n_f,n_q)$ with $n_f\geq 4$ and $n_f+n_q=d-1$}: consider the following cases:
\begin{enumerate}
	\item The case that there are two pairs of $\mathtt{f_{2D}}$ faults that occur on adjacent generators $f_i^z$ and $f_j^z$, and each pair leads to an error of weight 3 on the supporting qubits of each generator. We can always make these two errors of weight 3 overlap by multiplying each error with $f_i^z$ (or $f_j^z$); see examples below. 
	\begin{equation}
		\includegraphics[width=0.25\textwidth]{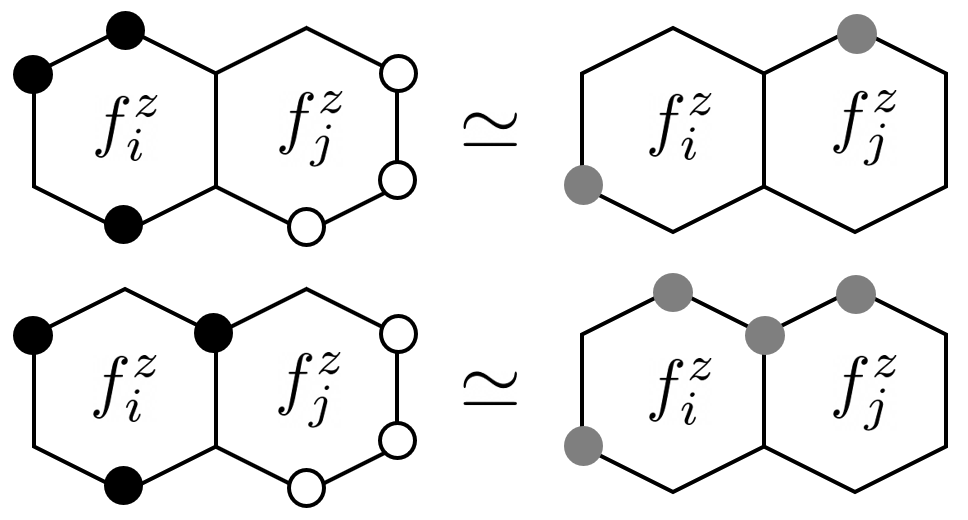} \nonumber
	\end{equation}
	As a result, the total weight of these two errors becomes 2 or 4. Since four $\mathtt{f_{2D}}$ faults give an error of weight no more than 4, this case is covered by the $(n_f-4,n_q+4)$ case. We can repeat this reduction process until there are no pairs of faults that occur on adjacent generators.
	\item The case that there are no pairs of $\mathtt{f_{2D}}$ faults that occur on adjacent generators. Suppose that a single pair of $\mathtt{f_{2D}}$ faults causes a weight-3 error on the supporting qubits of generator $f_i^z$.
	\begin{enumerate}
		\item If there are even number of $\mathtt{q_{2D}}$ faults on the supporting qubits of $f_i^z$, then the syndrome bit $s_i^x$ corresponding to generator $f_i^x$ is not zero. That is, $\mathbf{E}_\mathtt{q_{2D}}\cdot\mathbf{E}_\mathtt{f_{2D}}$ is not a nontrivial logical operator.
		\item If there are odd number of $\mathtt{q_{2D}}$ faults on the supporting qubits of $f_i^z$, then the total weight of the error on supporting qubits of $f_i^z$ is 0 or 2 (up to a multiplication of $f_i^x$). Since two $\mathtt{f_{2D}}$ faults and one or more $\mathtt{q_{2D}}$ fault give an error of weight no more than 2, this case is covered by the $(n_f-2,n_q+2)$ case.
	\end{enumerate}
\end{enumerate}
By induction, a nontrivial logical operator cannot occur in any case with $n_f\geq 4$ and $n_f+n_q=d-1$. 

Therefore, there is no fault combination from $d-1$ faults which gives the zero cumulative flag vector and a nontrivial logical operator on the 2D color code.
\end{proof}

From \cref{thm:main2}, it is always possible to obtain a distinguishable fault set $\mathcal{F}_t$ for a 2D color code of any distance (thus, fault-tolerant protocols for error correction, measurement, and state preparation described in \cref{sec:FT_protocol} are applicable).


Now let us consider the capped color code in H form. Because there is no fault combination from $d-1$ faults that can cause a nontrivial logical operator of the 2D color code with trivial flags on the bottom plane, a nontrivial logical operator of the capped color code in H form with trivial flags cannot occur from $d-1$ faults. By \cref{prop:2t}, this implies that the fault set $\mathcal{F}_t$ is distinguishable. The result can be summarized in the following theorem:

\begin{theorem}
	Let $\mathcal{F}_t$ be the fault set corresponding to circuits for measuring $\mathtt{f}, \mathtt{v}$, and $\mathtt{cap}$ generators of the capped color code in H form constructed from $CCC(d)$ (where $t=(d-1)/2$, $d=3,5,7,...$), and suppose that the general configurations of CNOT gates for $\mathtt{f}$, $\mathtt{v}$, and $\mathtt{cap}$ generators are imposed, and the circuits for each pair of $X$-type and $Z$-type generators use the same CNOT ordering. 
	Also, let circuits for measuring $\mathtt{f}$ and $\mathtt{cap}$ generators be flag circuits with one flag ancilla similar to the circuit in \cref{subfig:circuit_CCC_f}, and let circuits for measuring $\mathtt{v}$ generators be flag circuits with one flag ancilla similar to the circuit in \cref{subfig:circuit_CCC_v}. Then, $\mathcal{F}_t$ is distinguishable.
	\label{thm:main3}
\end{theorem}

(We can also see that whenever Condition \ref{con:con6} is satisfied, Conditions \ref{con:con1} to \ref{con:con5} are also satisfied. This leads to a distinguishable fault set by \cref{thm:main}.)

The fault-tolerant protocols for error correction, measurement, and state preparation in \cref{sec:FT_protocol} are applicable to a capped color code in H form of any distance whenever the fault set is distinguishable. Note that the protocols for the capped color code in H form of distance 3 and 5 need only one ancilla in total, while the protocols for the code of distance 7 or higher need only two ancillas in total (assuming that the ancillas can be reused).

In addition, the CNOT orderings which work for capped color codes in H form will work for recursive capped color codes in H form. That is, for a recursive capped color code in H form of distance $d=2t+1$, the fault set $\mathcal{F}_t$ is distinguishable if the followings are true:
\begin{enumerate}
	\item the $\mathtt{f}$ and $\mathtt{v}$ operators on the $(j-1)$-th and the $j$-th layers of the recursive capped color code are measured using the CNOT orderings for the $\mathtt{f}$ and $\mathtt{v}$ operators of a capped color code of in H form of distance $j$ ($j=3,5,...,d$) which give a distinguishable fault set, and
	\item the $\mathtt{cap}$ operator on the $(j-2)$-th and the $(j-1)$-th layers of the recursive capped color code is measured using the CNOT ordering for the $\mathtt{cap}$ operator of a capped color code in H form of distance $j$ ($j=3,5,...,d$) which give a distinguishable fault set (where an operator on $\mathtt{q_0}$ of the capped color code is replaced by operators on all qubits on the $(j-2)$-th layer of the recursive capped color code).
\end{enumerate}
The orderings above work because the recursive capped color code in H form of distance $d$ is obtained by encoding the top qubit ($\mathtt{q_0}$) of the capped color code in H form of distance $d$ by the recursive capped color code in H form of distance $d-2$. FTEC protocols for a recursive capped color code in H form are similar to conventional FTEC protocols for a concatenated code; we will start from correcting errors on the innermost code then proceed outwards. Other fault-tolerant protocols for a recursive capped color code will also use similar ideas.













\section{Fault-tolerant protocols}
\label{sec:FT_protocol}


So far, we have considered capped and recursive capped color codes in H form, and derived \cref{thm:main,thm:main2,thm:main3} which help us find CNOT orderings for the circuits for measuring the code generators such that the corresponding fault set is distinguishable.
In this section, we will show that whenever the fault set is distinguishable, a fault-tolerant protocol can be constructed. We will first state the definitions of fault-tolerant gadgets in \cref{subsec:FT_def}, which are a bit different from conventional definitions originally proposed by Aliferis, Gottesman, and Preskill in \cite{AGP06}. Afterwards, we will develop several fault-tolerant protocols for a capped or a recursive capped color code whose circuits for measuring generators give a distinguishable fault set, including a fault-tolerant error correction (FTEC) protocol (\cref{subsec:FTEC_ana}), fault-tolerant measurement (FTM) and fault-tolerant state preparation (FTP) protocols (\cref{subsec:FTM_ana}), transversal Clifford gates (\cref{subsec:other_FT_gadgets}), and fault-tolerant protocol for logical $T$-gate implementation (\cref{subsec:FT_T_gate}).

\subsection{Redefining fault tolerance}
\label{subsec:FT_def}

When a fault set $\mathcal{F}_t$ is distinguishable, all possible errors of any weight arising from up to $t$ faults can be accurately identified (up to a multiplication of some stabilizer) using their syndromes and cumulative flag vectors obtained from perfect subsequent syndrome measurements. Therefore, all possible errors arising from up to $t$ faults are correctable. However, one should be aware that faults can happen anywhere in an EC protocol, including the locations in the subsequent syndrome measurements. Our goal is to construct a protocol which is \emph{fault tolerant}; vaguely speaking, if an input state to an EC protocol has some error, we want to make sure that the output state is the same logical state as the input, and if output state has any error, the error must not be `too large'. 


What does it mean for the output error to be not too large? The general idea is that if an output error of a single round of the protocol becomes an input error of the next round of the protocol, the error should still be correctable by the latter round. In \cite{AGP06}, the authors proposed that the weight of the output error from a fault-tolerant protocol should be no more than the number of total faults occurred during the protocol. However, it should be noted that for an \codepar{n,k,d} code which can correct errors up to weight $\tau = \lfloor(d-1)/2\rfloor$ and is not a perfect code (or not a perfect CSS code)\footnote{A perfect code is a quantum code which saturates the quantum Hamming bound; i.e., there is a one-to-one correspondence between correctable errors and all possible syndromes \cite{Gottesman96,Gottesman97}. A perfect CSS code is defined similarly, except that the syndromes of $X$-type and $Z$-type errors are considered separately.}, the idea of correctable errors can be extended to some errors of weight more than $\tau$. For example, if the code being used is a non-perfect code of distance 3, there will be some error $E$ of weight more than 1 whose syndrome $\vec{s}(E)$ is different from those of errors of weight 1. If no other error $E'$ has the same syndrome as $E$ in the set of correctable errors, then in this case, $E$ is also correctable in the sense that we can perform an error correction by applying $E^\dagger$ every time we obtained the syndrome $\vec{s}(E)$. In this section, we will `refine' the idea of high-weight error correction and `redefine' fault tolerance using the notion of distinguishable fault set.



We will start by stating conventional definitions of fault-tolerant gadgets proposed by Aliferis, Gottesman, and Preskill \cite{AGP06}, then we will give the revised version of the same definitions. 
Recall that $\tau$ denotes the weight of errors that a stabilizer code can correct, and $t$ denotes the number of faults. The first two definitions are the definitions of an $r$-filter and an ideal decoder, which are the main tools for describing the properties of fault-tolerant gadgets. The definitions are as follows:

\begin{definition}{$r$-filter (AGP version)}
	
	Let $T(S)$ be the coding subspace defined by the stabilizer group $S$. An $r$-filter is the projector onto the subspace spanned by
	\begin{equation}
		\left\{E\left|\bar{\psi}\right\rangle;\left|\bar{\psi}\right\rangle\in T(S),\;\text{the weight of}\;E\;\text{is at most}\; r\right\}.
	\end{equation}
	An $r$-filter in the circuit form is displayed below:
	\begin{equation}
		\includegraphics[width=0.12\textwidth]{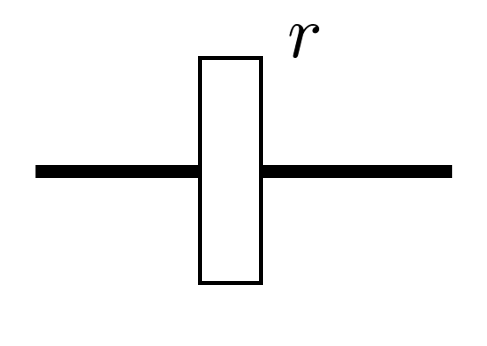} \nonumber
	\end{equation}
	where a thick line represents a block of code.
	\label{def:r_filter_old}%
\end{definition}

\begin{definition}{ideal decoder (AGP version)}
	
	Let $\tau=\lfloor (d-1)/2 \rfloor$ where $d$ is the code distance. An ideal decoder is a gadget which can correct any error of weight up to $\tau$ and map an encoded state $\left|\bar{\psi}\right\rangle$ on a code block to the corresponding (unencoded) state $\left|\psi\right\rangle$ on a single qubit without any fault. An ideal decoder in the circuit form is displayed below:\\
	\begin{equation}
		\includegraphics[width=0.18\textwidth]{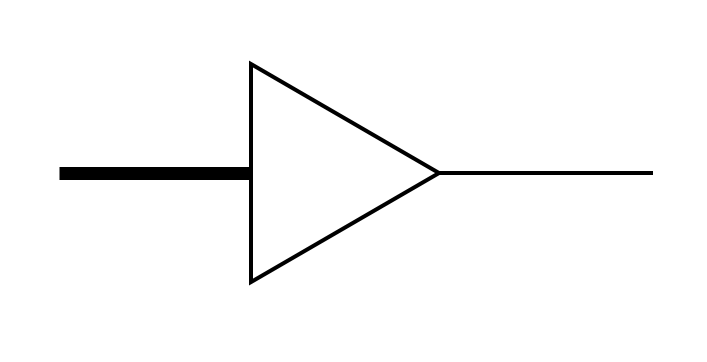} \nonumber
	\end{equation}
	where a thick line represents a block of code, and a thin line represents a single qubit.
	\label{def:ideal_old}%
\end{definition}


The intuition behind the definitions of these two gadgets are as follows: If an input state of an $r$-filter differs from a codeword by an error of weight $\leq r$, then the output state will also differ from the same codeword by an error of weight $\leq r$. However, if the input state has an error of weight $>r$, then the input and output states may correspond to different ideal codewords (i.e., they may be ideally decoded to different unencoded states). An ideal decoder is a gadget which guarantees that the output (unencoded) state and the input (encoded) state will be logically the same whenever the input state has an error of weight no more than $\tau$.

(Note that an $r$-filter is a linear, completely positive map but it is not trace-preserving; an $r$-filter cannot be physically implemented. In the definitions of fault-tolerant gadgets to be described, $r$-filters will be used as mathematical objects to express circuit identities that must be held when the weight of input or output errors and the number of faults are restricted. When each identity holds, both sides of the equation give the same output, including normalization, for the same input state, but the trace of the output might not be one.)

Using the definitions of $r$-filter and ideal decoder, fault-tolerant gate (FTG) gadget and fault-tolerant error correction (FTEC) gadget can be defined as follows:

\begin{definition}{Fault-tolerant gate gadget (AGP version)}
	
	A \emph{gate gadget} with $s$ faults simulating an ideal $m$-qubit gate is represented by the following picture:
	\begin{equation}
		\includegraphics[width=0.11\textwidth]{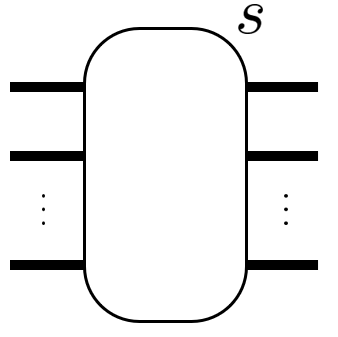} \nonumber
	\end{equation}
	where each thick line represents a block of code. Let $t \leq \lfloor (d-1)/2 \rfloor$. A gate gadget is \emph{$t$-fault tolerant} if it satisfies both of the following properties:
	\begin{enumerate}
		\item Gate correctness property (GCP): whenever $\sum_{i=1}^m r_i+s \leq t$,
		\begin{equation}
			\includegraphics[width=0.43\textwidth]{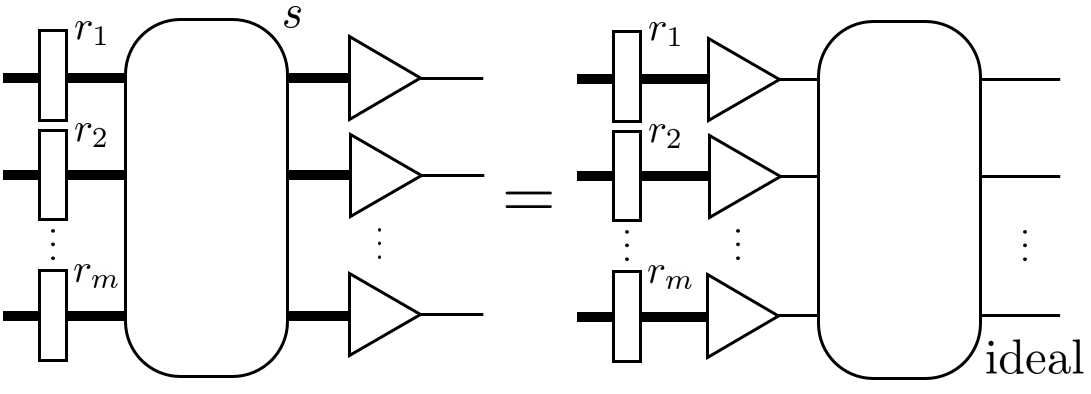} \nonumber
		\end{equation}
		\item Gate error propagation property (GPP): whenever $\sum_{i=1}^m r_i+s \leq t$,
		\begin{equation}
			\includegraphics[width=0.45\textwidth]{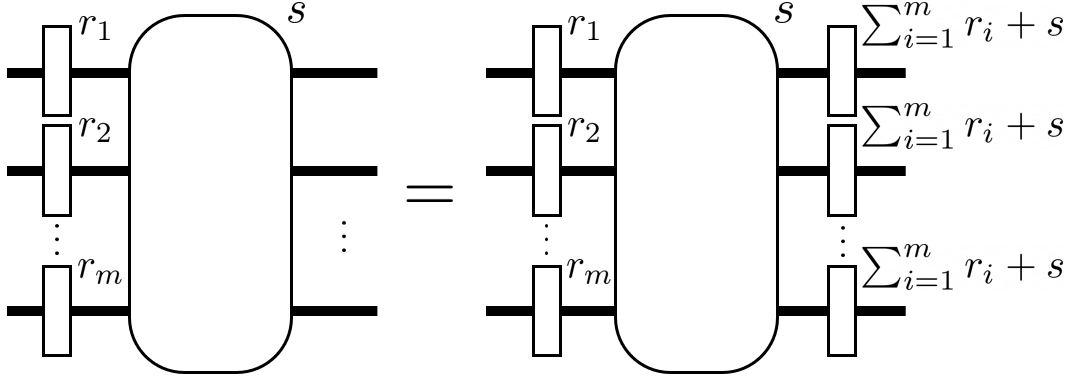} \nonumber
		\end{equation}
	\end{enumerate}
	where the $r$-filter and the ideal decoder are as defined in \cref{def:r_filter_old} and \cref{def:ideal_old}.
	\label{def:FTG_old}%
\end{definition}

\begin{definition}{Fault-tolerant error correction gadget (AGP version)}
	
	An \emph{error correction gadget} with $s$ faults is represented by the following picture:
	\begin{equation}
		\includegraphics[width=0.13\textwidth]{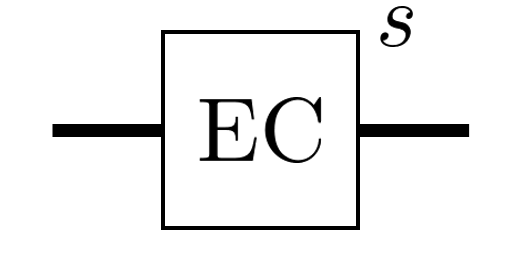} \nonumber
	\end{equation}
	where a thick line represents a block of code. Let $t \leq \lfloor (d-1)/2 \rfloor$. An error correction gadget is \emph{$t$-fault tolerant} if it satisfies both of the following properties:
	\begin{enumerate}
		\item Error correction correctness property (ECCP): whenever $r+s \leq t$,
		\begin{equation}
			\includegraphics[width=0.47\textwidth]{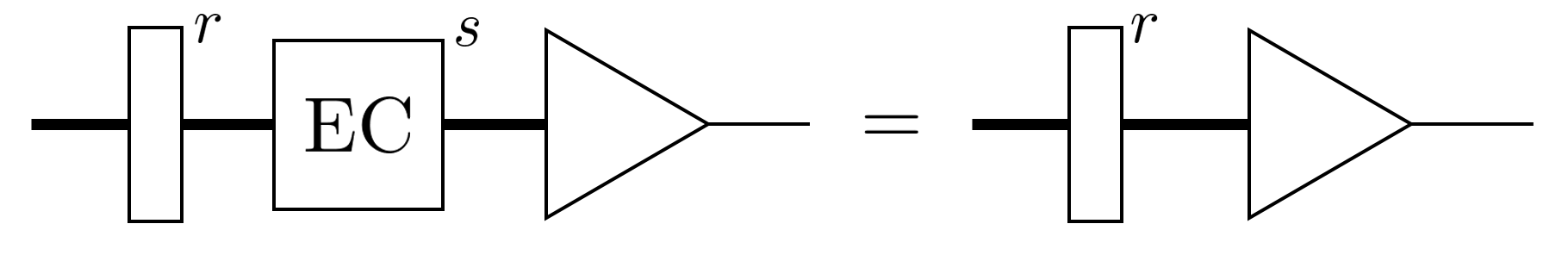} \nonumber
		\end{equation}
		\item Error correction recovery property (ECRP): whenever $s \leq t$,
		\begin{equation}
			\includegraphics[width=0.36\textwidth]{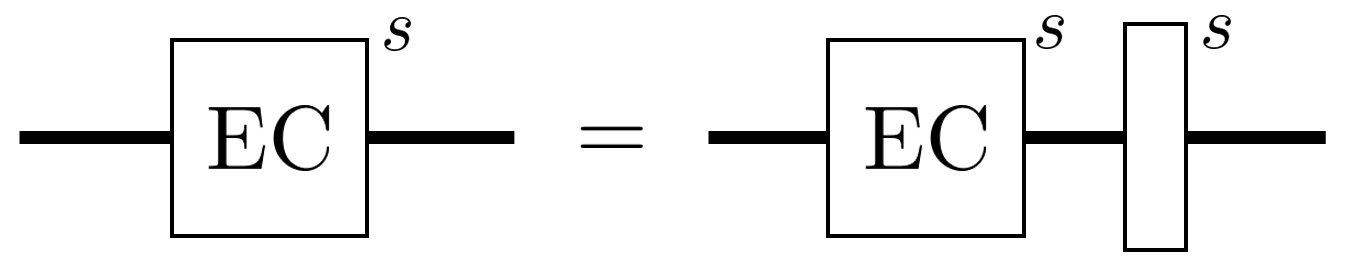} \nonumber
		\end{equation}
	\end{enumerate}
	where the $r$-filter and the ideal decoder are as defined in \cref{def:r_filter_old} and \cref{def:ideal_old}.
	\label{def:FTEC_old}%
\end{definition}


When an FTG gadget satisfies both properties in \cref{def:FTG_old}, it is guaranteed that whenever the weight of the input error plus the number of faults is no more than $t$, (1) the operation of an FTG gadget on an encoded state will be similar to the operation of its corresponding quantum gate on an unencoded state, and (2) an output state of an FTG gadget will have an error of weight no more than $t$ (which is also $\leq \tau$). Meanwhile, the two properties of an FTEC gadget in \cref{def:FTEC_old} guarantee that (1) the output and the input states of an FTEC gadget are logically the same whenever the weight of the input error plus the number of faults is no more than $t$, and (2) the weight of the output error of an FTEC gadget is no more than the number of faults whenever the number of faults is at most $t$, regardless of the weight of the input error.


Fault-tolerant state preparation (FTP) gadget and fault-tolerant (non-destructive) measurement (FTM) gadget, which are special cases of FTG gadget, can be defined as follows:

\begin{definition}{Fault-tolerant state preparation gadget (AGP version)}
	
	A \emph{state preparation gadget} with $s$ faults is represented by the following picture:
	\begin{equation}
		\includegraphics[width=0.11\textwidth]{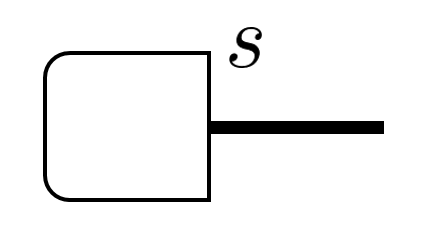} \nonumber
	\end{equation}
	where a thick line represents a block of code. Let $t \leq \lfloor (d-1)/2 \rfloor$. A state preparation gadget is \emph{$t$-fault tolerant} if it satisfies both of the following properties:
	\begin{enumerate}
		\item Preparation correctness property (PCP): whenever $s \leq t$,
		\begin{equation}
			\includegraphics[width=0.35\textwidth]{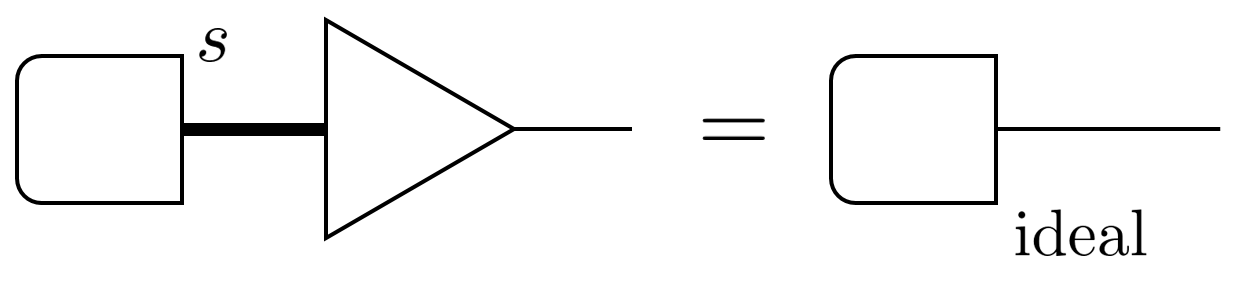} \nonumber
		\end{equation}
		\item Preparation error propagation property (PPP): whenever $s \leq t$,
		\begin{equation}
			\includegraphics[width=0.32\textwidth]{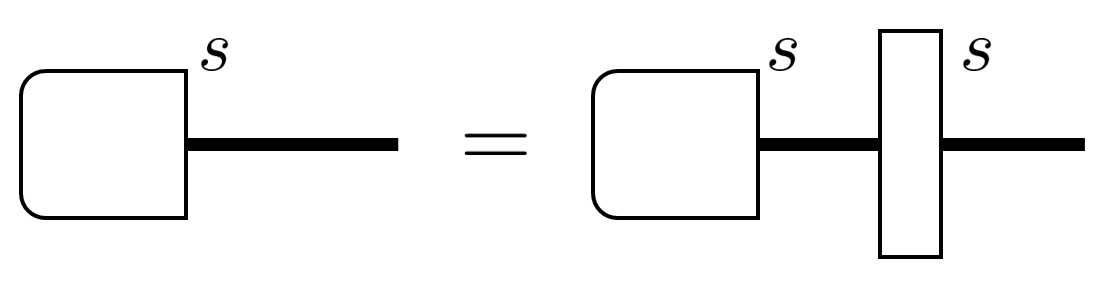} \nonumber
		\end{equation}
	\end{enumerate}
	where the $r$-filter and the ideal decoder are defined as in \cref{def:r_filter_old} and \cref{def:ideal_old}.
	\label{def:FTP_old}%
\end{definition}

\begin{definition}{Fault-tolerant (non-destructive) measurement gadget (AGP version)}
	
	A \emph{(non-destructive) measurement gadget} with $s$ faults is represented by the following picture:
	\begin{equation}
		\includegraphics[width=0.14\textwidth]{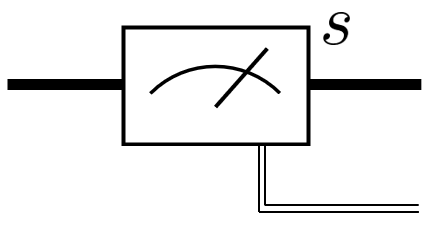} \nonumber
	\end{equation}
	where a thick line represents a block of code. Let $t \leq \lfloor (d-1)/2 \rfloor$. A (non-destructive) measurement gadget is \emph{$t$-fault tolerant} if it satisfies both of the following properties:
	\begin{enumerate}
		\item Measurement correctness property (MCP): whenever $r+s \leq t$,
		\begin{equation}
			\includegraphics[width=0.48\textwidth]{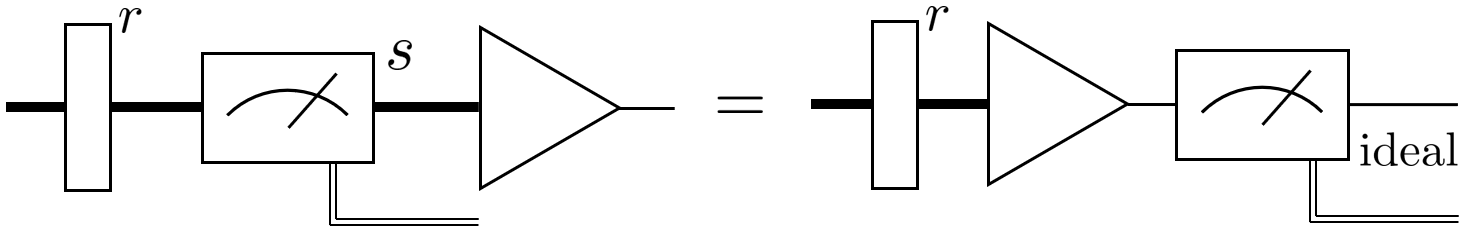} \nonumber
		\end{equation}
		\item Measurement error propagation property (MPP): whenever $r+s \leq t$,
		\begin{equation}
			\includegraphics[width=0.42\textwidth]{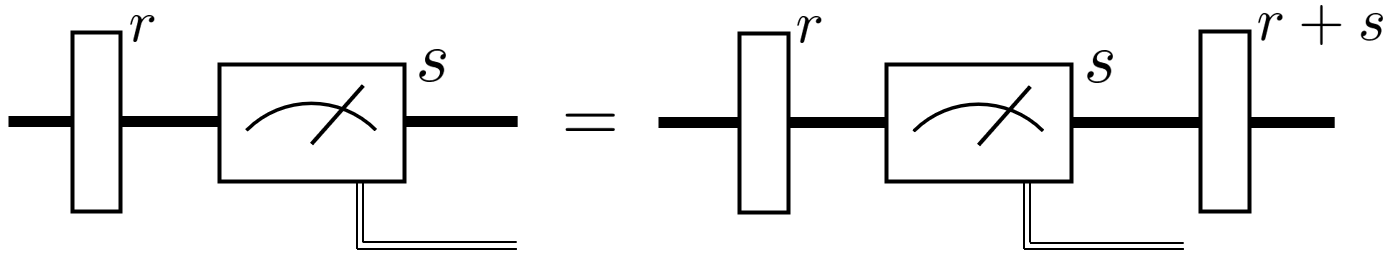} \nonumber
		\end{equation}
	\end{enumerate}
	where the $r$-filter and the ideal decoder are defined as in \cref{def:r_filter_old} and \cref{def:ideal_old}.
	\label{def:FTM_old}%
\end{definition}

The meanings of the properties of FTP and FTM gadgets are similar to the meanings of the properties of an FTG gadget as previously explained. 


From \cref{def:FTG_old,def:FTEC_old,def:FTP_old,def:FTM_old}, we can see that an action of a fault-tolerant gadget is guaranteed in the circumstance that the weight of the input error $r$ and the number of faults occurred in the gadget $s$ satisfy some condition. Now, a question arises: what will happen if the input error has weight greater than $\tau=\lfloor(d-1)/2\rfloor$, which is the weight of errors that a code can correct? By \cref{def:distinguishable}, we know that if a fault set $\mathcal{F}_t$ is distinguishable, possible errors arising from up to $t$ faults in an EC protocol (where $t \leq \lfloor(d-1)/2\rfloor$) can be distinguished using their corresponding syndromes or cumulative flag vectors, regardless of the error weights. Would it be more natural if the definitions of fault-tolerant gadgets depend on the \emph{number of faults} related to an input error, instead of the \emph{weight} of an input error? In this work, we will try to modify the definitions of fault-tolerant gadgets and rewrite them using the notion of distinguishable fault set.


To modify the definitions of fault-tolerant gadgets proposed in \cite{AGP06},
first, let us define distinguishable error set as follows:

\begin{definition}{Distinguishable error set}
	
	Let $\mathcal{F}_r$ be a distinguishable fault set, and let $\mathcal{F}_r|_{\vec{\mathbf{f}}=0}$ be a subset of $\mathcal{F}_r$ defined as follows:
	\begin{equation}
		\mathcal{F}_r|_{\vec{\mathbf{f}}=0} = \{\Lambda\in\mathcal{F}_r;\;\vec{\mathbf{f}}\;\text{of}\;\Lambda\;\text{is zero}\}.
	\end{equation}
	A \emph{distinguishable error set} $\mathcal{E}_r$ corresponding to $\mathcal{F}_r$ is,
	\begin{equation}
		\mathcal{E}_r = \{\mathbf{E}\;\text{of}\;\Lambda\in\mathcal{F}_r|_{\vec{\mathbf{f}}=0}\}.
	\end{equation}
	\label{def:dist_err}%
	\vspace*{-0.6cm}
\end{definition}


If $\mathcal{F}_r$ is distinguishable, $\mathcal{F}_r|_{\vec{\mathbf{f}}=0}$ is also distinguishable since all pairs of fault combinations in $\mathcal{F}_r|_{\vec{\mathbf{f}}=0}$ also satisfy the conditions in \cref{def:distinguishable}. Moreover, because all fault combinations in $\mathcal{F}_r|_{\vec{\mathbf{f}}=0}$ correspond to the zero cumulative flag vector, we find that for any pair of errors in $\mathcal{E}_r$, the errors either have different syndromes or are logically equivalent (up to a multiplication of a stabilizer). For this reason, we can safely say that $\mathcal{E}_r$ is a set of correctable errors.

Because the set of correctable errors is now expanded, the definitions of $r$-filter and ideal decoder can be revised as follows:


\begin{definition}{$r$-filter (revised version)}
	
	Let $T(S)$ be the coding subspace defined by the stabilizer group $S$, and let $\mathcal{E}_r$ be the distinguishable error set corresponding to a distinguishable fault set $\mathcal{F}_r$. An $r$-filter is the projector onto subspace spanned by
	\begin{equation}
		\left\{E\left|\bar{\psi}\right\rangle;\left|\bar{\psi}\right\rangle\in T(S),E\in\mathcal{E}_r\right\}.
	\end{equation}
	An $r$-filter in the circuit form is similar to the one illustrated in \cref{def:r_filter_old}.
	\label{def:r_filter_new}%
\end{definition}

\begin{definition}{ideal decoder (revised version)}
	
	Let $\mathcal{E}_t$ be the distinguishable error set corresponding to a distinguishable fault set $\mathcal{F}_t$, where $t\leq\lfloor (d-1)/2 \rfloor$ and $d$ is the code distance. An ideal decoder is a gadget which can correct any error in $\mathcal{E}_t$ and map an encoded state $\left|\bar{\psi}\right\rangle$ on a code block to the corresponding (unencoded) state $\left|\psi\right\rangle$ on a single qubit without any faults. An ideal decoder in the circuit form is similar to the one illustrated in \cref{def:ideal_old}.
	\label{def:ideal_new}%
\end{definition}


Using the revised definitions of $r$-filter and ideal decoder, fault-tolerant gadgets can be defined as follows:

\begin{definition}{Fault-tolerant gadgets (revised version)}
	
	Let $t \leq \lfloor (d-1)/2 \rfloor$. Fault-tolerant gadgets are defined as follows:
	\begin{enumerate}
		\item A \emph{gate gadget} is \emph{$t$-fault tolerant} if it satisfies both of the properties in \cref{def:FTG_old}, except that $r$-filter and ideal decoder are defined as in \cref{def:r_filter_new} and \cref{def:ideal_new}.
		\item An \emph{error correction gadget} is \emph{$t$-fault tolerant} if it satisfies both of the properties in \cref{def:FTEC_old}, except that $r$-filter and ideal decoder are defined as in \cref{def:r_filter_new} and \cref{def:ideal_new}.
		\item A \emph{state preparation gadget} is \emph{$t$-fault tolerant} if it satisfies both of the properties in \cref{def:FTP_old}, except that $r$-filter and ideal decoder are defined as in \cref{def:r_filter_new} and \cref{def:ideal_new}.
		\item A \emph{(non-destructive) measurement gadget} is \emph{$t$-fault tolerant} if it satisfies both of the properties in \cref{def:FTM_old}, except that $r$-filter and ideal decoder are defined as in \cref{def:r_filter_new} and \cref{def:ideal_new}.
	\end{enumerate}
	\label{def:FT_gadget_new}%
\end{definition}



The revised definitions of fault-tolerant gadgets in the circuit form may look very similar to the old definitions proposed in \cite{AGP06}, but the meanings are different: the conditions in the revised definitions depend on the number of faults which can cause an input or an output error, instead of the weight of an input or an output error. Roughly speaking, this means that (1) a fault-tolerant gadget is allowed to produce an output error of weight greater than $\tau$ (where $\tau=\lfloor(d-1)/2\rfloor$), and (2) a fault-tolerant gadget can work perfectly even though the input error has weight greater than $\tau$, as long as the input or the output error is similar to an error caused by no more than $t\leq \tau$ faults. Because the revised definitions of $r$-filter and ideal decoder are more general than the old definitions, we expect that a gadget that satisfies one of the old definitions of fault-tolerant gadgets (\cref{def:FTG_old,def:FTEC_old,def:FTP_old,def:FTM_old}) will also satisfy the new definitions in \cref{def:FT_gadget_new}. Note that the revised definitions are based on the fact that a fault set relevant to a gadget is distinguishable, that is, whether the gadgets are fault tolerant depends on the way they are designed.


In a special case where the code being used is a CSS code and possible $X$-type and $Z$-type errors have the same form, the definition of distinguishable error set can be further extended as follows:
\begin{definition}{Distinguishable error set (for a special family of CSS codes)}
	
	Let $\mathcal{F}_r$ be a distinguishable fault set, and let $\mathcal{F}_r|_{\vec{\mathbf{f}}=0}$ be a subset of $\mathcal{F}_r$ defined as follows:
	\begin{equation}
		\mathcal{F}_r|_{\vec{\mathbf{f}}=0} = \{\Lambda\in\mathcal{F}_r;\;\vec{\mathbf{f}}\;\text{of}\;\Lambda\;\text{is zero}\}.
	\end{equation}
	A \emph{distinguishable-$X$ error set} $\mathcal{E}_r^x$ and a \emph{distinguishable-$Z$ error set} $\mathcal{E}_r^z$ corresponding to $\mathcal{F}_r$ are,
	\begin{align}
		\mathcal{E}_r^x &= \{\mathbf{E}\;\text{of}\;\Lambda\in\mathcal{F}_r|_{\vec{\mathbf{f}}=0};\mathbf{E}\;\text{is an}\;X\text{-type error}\},\\
		\mathcal{E}_r^z &= \{\mathbf{E}\;\text{of}\;\Lambda\in\mathcal{F}_r|_{\vec{\mathbf{f}}=0};\mathbf{E}\;\text{is a}\;Z\text{-type error}\}.
	\end{align}
	For a CSS code in which the elements of $\mathcal{E}_r^x$ and $\mathcal{E}_r^z$ have a similar form, a \emph{distinguishable error set} $\mathcal{E}_r$ corresponding to $\mathcal{F}_r$ is defined as follows:
	\begin{equation}
		\mathcal{E}_r = \{E_x\cdot E_z; E_x \in \mathcal{E}_r^x, E_z \in \mathcal{E}_r^z\}.
	\end{equation}
	\label{def:dist_err_CSS}%
	\vspace*{-0.6cm}
\end{definition}

Since a CSS code can detect and correct $X$-type and $Z$-type errors separately, here we modify the definition of distinguishable error set for a CSS code in which $\mathcal{E}_r^x$ and $\mathcal{E}_r^z$ are in the same form so that more $Y$-type errors are included in $\mathcal{E}_r$. For example, suppose that $t=2$, each of $XXXX$ and $ZZZZ$ can be caused by 2 faults, and $YYYY$ can be caused by 4 faults. By the old definition (\cref{def:dist_err}), we will say that $XXXX$ and $ZZZZ$ are in $\mathcal{E}_2$, and $YYYY$ is in $\mathcal{E}_4$ but not in $\mathcal{E}_2$. In contrast, by \cref{def:dist_err_CSS}, we will say that $XXXX$, $YYYY$, and $ZZZZ$ are all in $\mathcal{E}_2$. This modification will give more flexibility when developing a fault-tolerant gadget for this special kind of CSS codes, e.g., a transversal $S$ gate which produces an output error $YYYY$ from an input error $XXXX$ still satisfies the properties in \cref{def:FT_gadget_new} when a distinguishable fault set is defined as in \cref{def:dist_err_CSS}.


When performing a fault-tolerant quantum computation, FTEC gadgets will be used repeatedly in order to reduce the error accumulation during the computation. Normally, FTEC gadgets will be placed before and after other gadgets (FTG, FTP, or FTM gadgets). A group of gadgets including an FTG gadget, leading EC gadgets (the FTEC gadgets before the FTG gadget), and trailing EC gadgets (FTEC gadgets after the FTG gadget) as shown below is called an \emph{extended rectangle at level 1} or \emph{1-exRec}: 
\begin{equation}
	\includegraphics[width=0.2\textwidth]{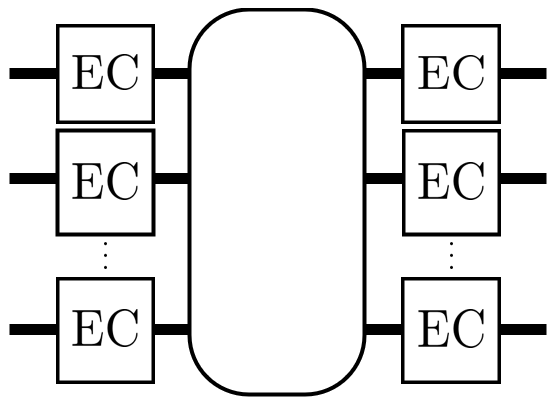} \nonumber
\end{equation}
(A 1-exRec of an FTP or FTM gadget is defined similarly to a 1-exRec of an FTG gadget, except that there is no leading gadget in an FTP gadget.) We say that a 1-exRec is \emph{good} if the total number of faults in a 1-exRec is no more than $t$. Using the revised definitions of fault-tolerant gadgets in \cref{def:FT_gadget_new}, a revised version of the exRec-Cor lemma at level 1, originally proposed in \cite{AGP06}, can be obtained:
\begin{lemma}{ExRec-Cor lemma at level 1 (revised version)}
	
	Suppose that all gadgets are $t$-fault tolerant according to \cref{def:FT_gadget_new}. If a 1-exRec is \emph{good} (i.e., a 1-exRec has no more than $t$ faults), then the 1-exRec is \emph{correct}; that is, the following condition is satisfied:
	\begin{equation}
		\includegraphics[width=0.48\textwidth]{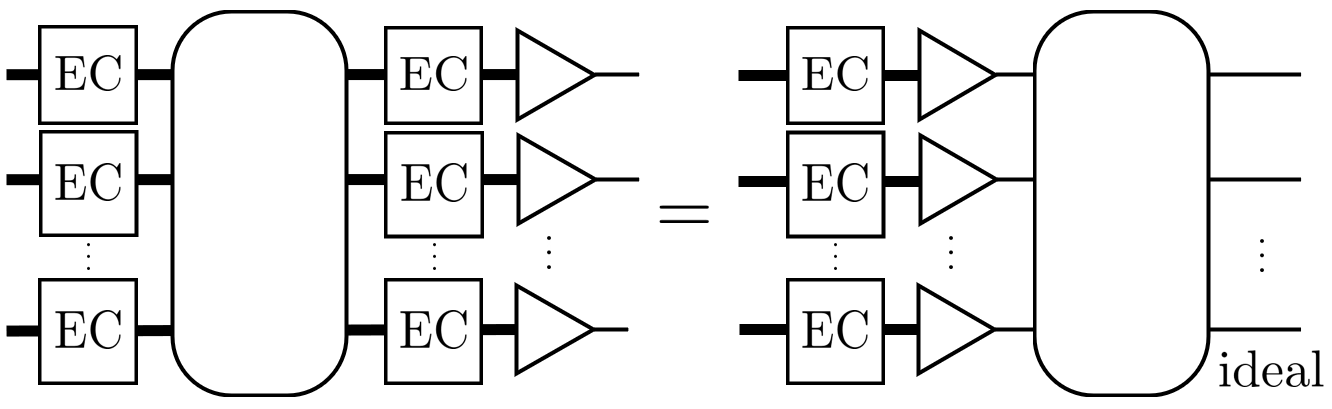} \nonumber
	\end{equation}
	where the $r$-filter and the ideal decoder are defined as in \cref{def:r_filter_new,def:ideal_new}.
	\label{lem:exRec-cor}%
\end{lemma}

\begin{proof}
	Here we will focus only on the case that a gate gadget simulates a single-qubit gate. The proofs for the case of multiple-qubit gate and other gadgets are similar. Suppose that the leading EC gadget, the gate gadget, and the trailing EC gadget in an exRec have $s_1,s_2,$ and $s_3$ faults where $s_1+s_2+s_3 \leq t$. We will show that the following equation holds:
	\begin{equation}
		\includegraphics[width=0.47\textwidth]{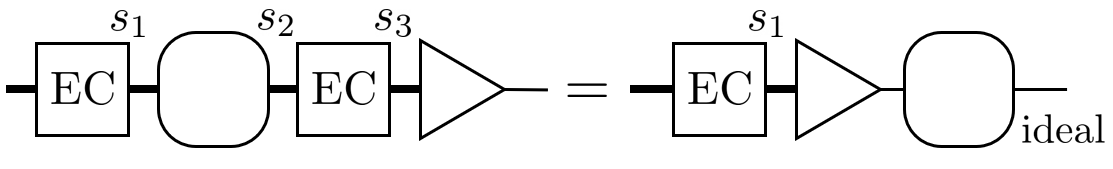} \label{eq:proof_exRec}
	\end{equation}
	Because the gate gadget satisfies GPP and the EC gadgets satisfy ECRP, the left-hand side of \cref{eq:proof_exRec} is
	\begin{equation}
		\includegraphics[width=0.42\textwidth]{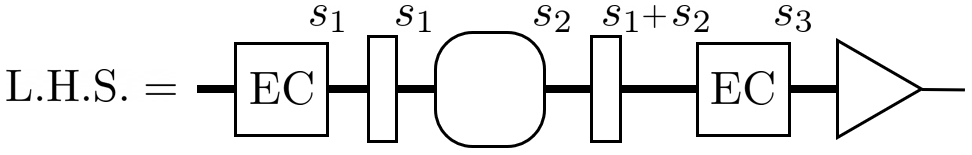} \nonumber
	\end{equation}
	Using GCP, ECCP, and the fact that an ideal decoder can correct any error in $\mathcal{E}_t$, we obtain the following:
	\begin{equation}
		\includegraphics[width=0.37\textwidth]{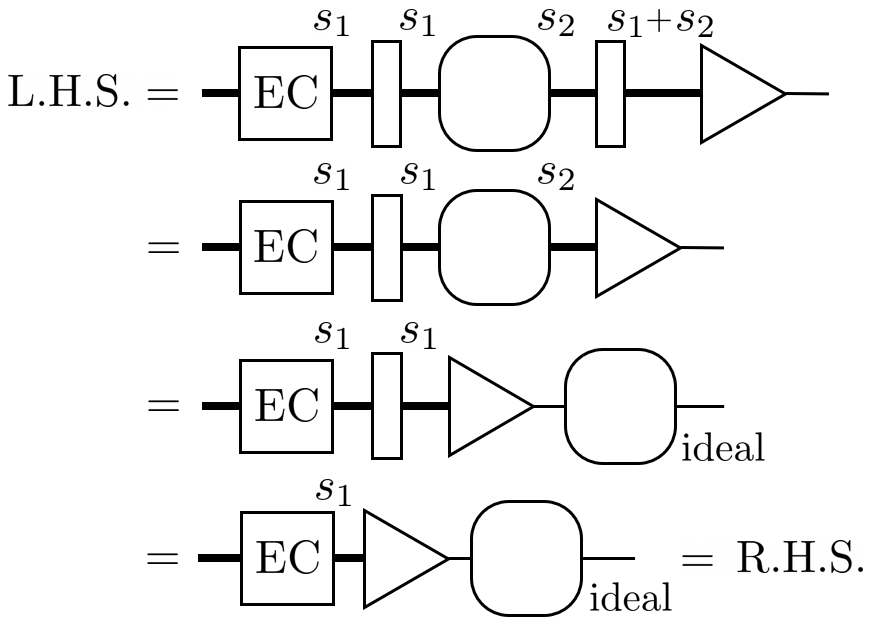} \nonumber
	\end{equation}
	\vspace*{-0.3cm}
\end{proof}

(Note that both sides of the equation in \cref{lem:exRec-cor} are trace-preserving, completely positive maps, even though $r$-filters introduced during the proof are not trace-preserving. This is possible since the total number of faults in a 1-exRec is restricted and all gadgets satisfy \cref{def:FT_gadget_new}.)

The revised version of the exRec-cor lemma developed in this work is very similar to the original version in \cite{AGP06}, even though the $r$-filter, the ideal decoder, and the fault-tolerant gadgets are redefined. 
The exRec-Cor lemma is one of the main ingredients for the proofs of other lemmas and theorems in \cite{AGP06}. As a result, other lemmas and theorems developed in \cite{AGP06} are also applicable to our case, including their version of the \emph{threshold theorem} (the proofs of revised versions of the lemmas and theorems are similar to the proofs presented in \cite{AGP06}, except that \cref{lem:exRec-cor} is used instead of the original exRec-Cor lemma).
This means that fault-tolerant gadgets satisfying \cref{def:FT_gadget_new} can be used to simulate any quantum circuit, and the logical error rate can be made arbitrarily small if the physical error rate is below some constant threshold value. The main advantage of the revised definitions of fault-tolerant gadgets over the conventional definitions is that high-weight errors are allowed as long as they arise from a small number of faults. These revised definitions can give us more flexibility when developing fault-tolerant protocols.

\subsection{Fault-tolerant error correction protocol}
\label{subsec:FTEC_ana}



So far, we have shown that it is possible to redefine $r$-filter and ideal decoder as in \cref{def:r_filter_new,def:ideal_new} using the notions of distinguishable fault set (\cref{def:distinguishable}) and distinguishable error set (\cref{def:dist_err} or \cref{def:dist_err_CSS}), and redefine fault-tolerant gadgets as in \cref{def:FT_gadget_new}. These revised definitions give us more flexibility when designing fault-tolerant protocols, while ensuring that the simulated circuit constructed from these protocols still work fault-tolerantly. In this section, we will construct an FTEC protocol for a capped color code in H form of any distance in which its fault set is distinguishable. Note that having only the FTEC protocol is not enough for general fault-tolerant quantum computation, so we will also construct other fault-tolerant protocols which share the same distinguishable fault set with the FTEC protocol for a particular code in \cref{subsec:FTM_ana,subsec:other_FT_gadgets,subsec:FT_T_gate}.

\begin{figure*}[htbp]
	\centering
	\includegraphics[width=0.85\textwidth]{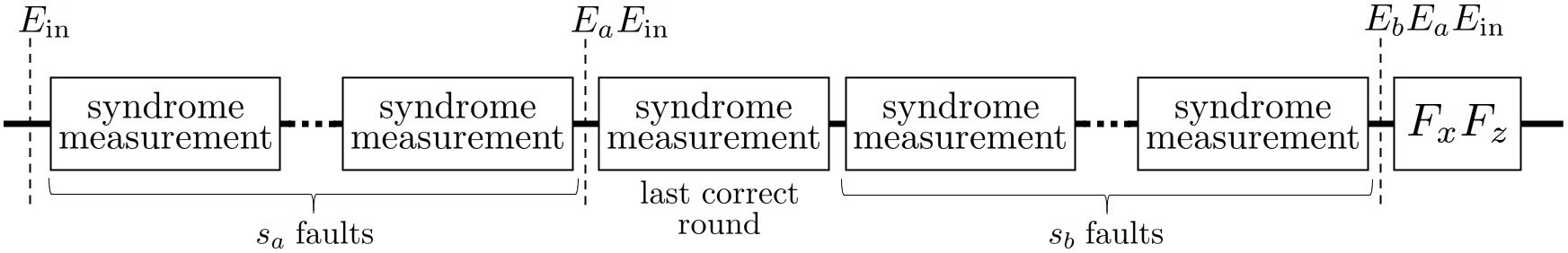}
	\captionsetup{justification=centering}
	\caption{Fault-tolerant error correction protocol for a capped color code.}
	\label{fig:err_in_FTEC}
\end{figure*}


To construct an FTEC protocol for a capped color code in H form obtained from $CCC(d)$, we will first assume that the fault set $\mathcal{F}_t$ (where $t=(d-1)/2$) corresponding to the circuits for measuring the generators of the code is distinguishable, and the orderings of gates in the circuits for each pair of $X$-type and $Z$-type generators are the same. From the fact that $\mathcal{F}_t$ is distinguishable, we can build a list of all possible fault combinations and their corresponding combined error, syndrome of the combined error, and cumulative flag vector. Note that if several fault combinations have the same syndrome and cumulative flag vector, their combined errors are all logically equivalent (from \cref{def:distinguishable}).

Let $\vec{\mathbf{s}}=(\vec{\mathbf{s}}_x|\vec{\mathbf{s}}_z)$ be the syndrome obtained from the measurements of $X$-type and $Z$-type generators, and let $\vec{\mathbf{f}}=(\vec{\mathbf{f}}_x|\vec{\mathbf{f}}_z)$ be the cumulative flag vector corresponding to the flag outcomes from the circuits for measuring $X$-type and $Z$-type generators, where $\vec{\mathbf{f}}$ is accumulated from the first round until the current round. We define the \emph{outcome bundle} $(\vec{\mathbf{s}},\vec{\mathbf{f}})$ to be the collection of $\vec{\mathbf{s}}$ and $\vec{\mathbf{f}}$ obtained during a single round of full syndrome measurement. An FTEC protocol for the capped color code in H form is as follows:
\\

\noindent\textbf{FTEC protocol for a capped color code in H form}

During a single round of full syndrome measurement, measure the generators in the following order: measure $v_i^x$'s, then $f_i^x$'s, then $v_i^z$'s, then $f_i^z$'s. Perform full syndrome measurements until the outcome bundles $(\vec{\mathbf{s}},\vec{\mathbf{f}})$ are repeated $t+1$ times in a row. Afterwards, do the following:

\begin{enumerate}
	\item Determine an EC operator $F_x$ using the list of possible fault combinations as follows: 
	\begin{enumerate}
		\item If there is a fault combination on the list whose syndrome and cumulative flag vector are $(\vec{0}|\vec{\mathbf{s}}_z)$ and $(\vec{\mathbf{f}}_x|\vec{0})$, then $F_x$ is the combined error of such a fault combination. (If there are more than one fault combination corresponding to $(\vec{0}|\vec{\mathbf{s}}_z)$ and $(\vec{\mathbf{f}}_x|\vec{0})$, a combined error of any of such fault combinations will work.) \label{step:EC_1a}
		\item If none of the fault combinations on the list corresponds to $(\vec{0}|\vec{\mathbf{s}}_z)$ and $(\vec{\mathbf{f}}_x|\vec{0})$, then $F_x$ can be any Pauli $X$ operator whose syndrome is $(\vec{0}|\vec{\mathbf{s}}_z)$. \label{step:EC_1b}
	\end{enumerate}
	\item Determine an EC operator $F_z$ using the list of possible fault combinations: 
	\begin{enumerate}
		\item If there is a fault combination on the list whose syndrome and cumulative flag vector are $(\vec{\mathbf{s}}_x|\vec{0})$ and $(\vec{0}|\vec{\mathbf{f}}_z)$, then $F_z$ is the combined error of such a fault combination. (If there are more than one fault combination corresponding to $(\vec{\mathbf{s}}_x|\vec{0})$ and $(\vec{0}|\vec{\mathbf{f}}_z)$, a combined error of any of such fault combinations will work.) \label{step:EC_2a}
		\item If none of the fault combinations on the list corresponds to $(\vec{\mathbf{s}}_x|\vec{0})$ and $(\vec{0}|\vec{\mathbf{f}}_z)$, then $F_z$ can be any Pauli $Z$ operator whose syndrome is $(\vec{\mathbf{s}}_x|\vec{0})$. \label{step:EC_2b}
	\end{enumerate}
	
	\item Apply $F_x\cdot F_z$ to the data qubits to perform error correction. \label{step:EC_3}
\end{enumerate}	


To verify that the above EC protocol is fault tolerant according to the revised definition (\cref{def:FT_gadget_new}), we have to show that the two properties in \cref{def:FTEC_old} are satisfied when the $r$-filter and the ideal decoder are defined as in \cref{def:r_filter_new,def:ideal_new} (instead of \cref{def:r_filter_old,def:ideal_old}) and the distinguishable error set is defined as in \cref{def:dist_err_CSS} (the circuits for $X$-type and $Z$-type generators of the capped color code in H form use similar gate orderings). Here we will assume that there are no more than $t$ faults during the whole protocol. Therefore, the condition that the outcome bundles are repeated $t+1$ times in a row will be satisfied within $(t+1)^2$ rounds. We will divide the analysis into two cases: (1) the case that the last round of the full syndrome measurement has no faults, and (2) the case that the last round has some faults.

(1) Because the outcome bundles are repeated $t+1$ times and the last round of the full syndrome measurement has no faults, we know that the outcome bundle of the last round is correct and corresponds to the data error before the error correction in Step \ref{step:EC_3}. Let $E_\mathrm{in}$ be the input error and $E_a$ be the combined error of a fault combination arising from the $s_a$ faults where $s_a \leq t$. The error on the data qubits before Step \ref{step:EC_3} is $E_a\cdot E_\mathrm{in}$. First, consider the case that $E_\mathrm{in}$ is in $\mathcal{E}_r$ (defined in \cref{def:dist_err_CSS}) where $r+s_a \leq t$. Both $E_\mathrm{in}$ and $E_a$ can be separated into $X$ and $Z$ parts. We find that the $X$ part of $E_\mathrm{in}$ is in  $\mathcal{E}_r^x$ (which is derived from $\mathcal{F}_r|_{\vec{\mathbf{f}}=0}$). Thus, the $X$ part of $E_a\cdot E_\mathrm{in}$ is the combined error of $X$ type of some fault combination in $\mathcal{F}_{r+s_a}$. Similarly, the $Z$ part of $E_\mathrm{in}$ is in  $\mathcal{E}_r^z$, and the $Z$ part of $E_a\cdot E_\mathrm{in}$ is the combined error of $Z$ type of some fault combination in $\mathcal{F}_{r+s_a}$. By picking EC operators $F_x$ and $F_z$ as in Steps \ref{step:EC_1a} and \ref{step:EC_2a}, Step \ref{step:EC_3} can completely remove the data error. Thus, both ECCP and ECRP in \cref{def:FTEC_old} are satisfied. On the other hand, if $E_\mathrm{in}$ is not in $\mathcal{E}_r$ where $r+s_a \leq t$, the $X$ part or the $Z$ part of $E_a\cdot E_\mathrm{in}$ might not correspond to any fault combination in $\mathcal{F}_t$. In this case, $F_x$ or $F_z$ will be picked as in Step \ref{step:EC_1b} or \ref{step:EC_2b}. Because the $X$ part (or the $Z$ part) of $E_a\cdot E_\mathrm{in}$ and $F_x$ (or $F_z$) have the same syndrome no matter how we pick $F_x$ (or $F_z$), the output state after Step \ref{step:EC_3} is a valid codeword, but it may or may not be logically the same as the input state. In any cases, the output state can pass the $s_a$-filter, so the ECRP in \cref{def:FTEC_old} is satisfied.

(2) In the case that the last round of the full syndrome measurement has some faults, the outcome bundle of the last round may not correspond to the data error before the error correction in Step \ref{step:EC_3}. Fortunately, since the outcome bundles are repeated $t+1$ times in a row and there are no more than $t$ faults during the whole protocol, we know that at least one round in the last $t+1$ rounds must be correct, and the outcome bundle of the last round must correspond to the data error right before the last correct round. Let $E_\mathrm{in}$ be the input error, $E_a$ be the combined error arising from $s_a$ faults which happen before the last correct round, and $E_b$ be the combined error arising from $s_b$ faults which happen after the last correct round, where the total number of faults is $s = s_a+s_b \leq t$ (see \cref{fig:err_in_FTEC}). First, consider the case that $E_\mathrm{in}$ is in $\mathcal{E}_r$ where $r+s \leq t$. By an analysis similar to that presented in (1), we find that both $X$ and $Z$ parts of $E_a\cdot E_\mathrm{in}$ are the combined errors of some fault combinations in $\mathcal{F}_{r+s_a}$, and $F_x$ and $F_z$ from Steps \ref{step:EC_1a} and \ref{step:EC_2a} can completely remove $E_a\cdot E_\mathrm{in}$. Thus, the output data error after Step \ref{step:EC_3} is $E_b$. Since $s_b \leq t$ and the cumulative flag vectors do not change after the last correct round, we find that $E_b$ is the combined error of some fault combination arising from $s_b$ faults whose cumulative flag vector is zero; that is, $E_b$ is in $\mathcal{E}_{s_b}$ where $s_b \leq t$. For this reason, $E_b$ can pass the $s$-filter and can be corrected by the ideal decoder, meaning that both ECCP and ECRP in \cref{def:FTEC_old} are satisfied. In contrast, if $E_\mathrm{in}$ is not in $\mathcal{E}_r$ where $r+s \leq t$, $E_\mathrm{in}$ may not correspond to any fault combination in $\mathcal{F}_t$, and $F_x$ or $F_z$ may be picked as in Step \ref{step:EC_1b} or \ref{step:EC_2b}. Similar to the previous analysis, $F_x \cdot F_z$ will have the same syndrome as that of $E_a\cdot E_\mathrm{in}$. By an operation in Step \ref{step:EC_3}, the output state will be a valid codeword with error $E_b$, which can pass the the $s$-filter. Therefore, the ECRP in \cref{def:FTEC_old} is satisfied in this case.

In addition to the capped color code in H form, the FTEC protocol above is also applicable to any CSS code in which $\mathcal{F}_t$ is distinguishable and the possible $X$-type and $Z$-type errors are of the same form (i.e., a code to which \cref{def:dist_err_CSS} is applicable for all $r \in \{1,\dots,t\}$, $t\leq \lfloor (d-1)/2 \rfloor$). Besides this, we can also construct an FTEC protocol for a general stabilizer code whose circuits for the syndrome measurement give a distinguishable fault set (a code in which $\mathcal{E}_r$ is defined by \cref{def:dist_err} instead of \cref{def:dist_err_CSS}) using similar ideas. An FTEC protocol for such a code is provided in \cref{sec:app:protocol_general}.

Because a recursive capped color code in H form of distance $d$ is constructed by recursively encoding the top qubit of the capped color code in H form of distance $d$ using capped color codes of smaller distances, an FTEC protocol for a recursive capped color code in H form can be constructed similarly to an FTEC protocol for a concatenated code. The FTEC protocol is as follows:\\

\noindent\textbf{FTEC protocol for a recursive capped color code in H form}

For $j=3,5,7,\dots,d$, perform error correction on the first $j$ layers of the recursive capped color code of distance $d$ using the FTEC protocol for a capped color code in H form of distance $j$.


\subsection{Fault-tolerant measurement and state preparation protocols}
\label{subsec:FTM_ana}

\begin{figure*}[htbp]
	\centering
	\includegraphics[width=0.85\textwidth]{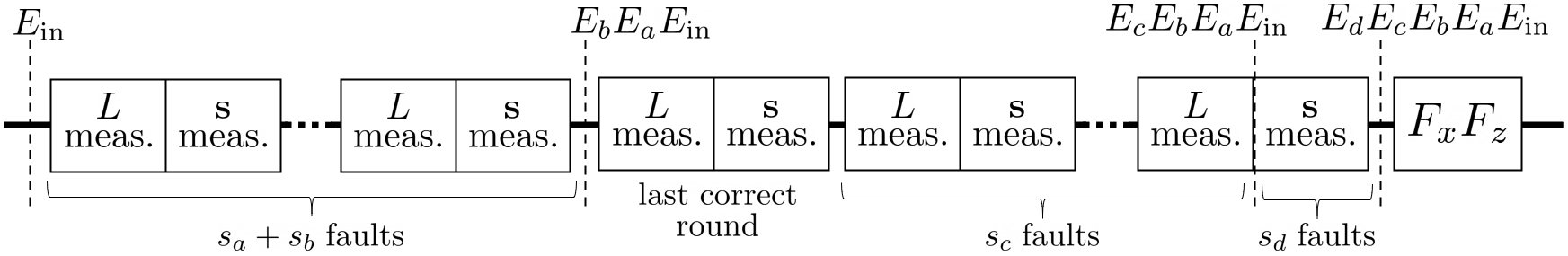}
	\captionsetup{justification=centering}
	\caption{Fault-tolerant measurement protocol for a capped color code.}
	\label{fig:err_in_FTM}
\end{figure*}


Besides FTEC protocols, we also need other gadgets such as FTM, FTP, and FTG gadgets in order to perform fault-tolerant quantum computation. Note that the definitions of the $r$-filter (\cref{def:r_filter_new}) and the ideal decoder (\cref{def:ideal_new}) depend on how the distinguishable error set is defined. Therefore, in order to utilize the new definitions of fault-tolerant gadgets in \cref{def:FT_gadget_new}, all protocols used in the computation must share the same definition of distinguishable error set. In this section, we will construct an FTM protocol for a capped color code in H form, which is also applicable to other CSS codes with similar properties. The distinguishable error set being used in the construction of the FTM protocol will be similar to the distinguishable error set defined for the FTEC protocol for the same code. In addition, an FTP protocol can also be obtained from the FTM protocol.

We will start by constructing an FTM protocol for a capped color code in H form obtained from $CCC(d)$. The FTM protocol discussed below can be used to fault-tolerantly measure any logical $X$ or logical $Z$ operator of the form $X^{\otimes n}M$ or $Z^{\otimes n}N$, where $M,N$ are some stabilizers. Let $L$ be the logical operator being measured. We will assume that the circuits for measuring $X$-type and $Z$-type generators are similar to the ones used in the FTEC protocol for a capped color code, which give a distinguishable fault set $\mathcal{F}_t$ with $t=(d-1)/2$ (the list of possible fault combinations for the FTM protocol is the same as the list used in the FTEC protocol). In addition, we can always use a non-flag circuit with an arbitrary gate ordering for measuring $L$ (since any error arising from the circuit faults can always be corrected as we will see later in the protocol analysis). For the FTM protocol, the outcome bundle will be defined as $(m,\vec{\mathbf{s}},\vec{\mathbf{f}})$, where 
$m$ is the measurement outcome of the logical operator $L$ ($m=0$ and $m=1$ correspond to $+1$ and $-1$ eigenvalues of $L$), and $\vec{\mathbf{s}}=(\vec{\mathbf{s}}_x|\vec{\mathbf{s}}_z)$ and $\vec{\mathbf{f}} = (\vec{\mathbf{f}}_x|\vec{\mathbf{f}}_z)$ are the syndrome and the cumulative flag vector obtained from the measurements of $X$-type and $Z$-type generators ($\vec{\mathbf{f}}$ is accumulated from the first round until the current round). An FTM protocol is as follows:
\\

\noindent\textbf{FTM protocol for a capped color code in H form}

During a single round of logical operator and full syndrome measurements, measure the operators in the following order: measure $L$, then $v_i^x$'s, then $f_i^x$'s, then $v_i^z$'s, then $f_i^z$'s. Perform logical operator and full syndrome measurements until the outcome bundles $(m,\vec{\mathbf{s}},\vec{\mathbf{f}})$ are repeated $t+1$ times in a row. Afterwards, do the following:

\begin{enumerate}
	\item \label{step:M_1} Determine an EC operator $F_x$ using the list of possible fault combinations as follows: 
	\begin{enumerate}
		\item If there is a fault combination on the list whose syndrome and cumulative flag vector are $(\vec{0}|\vec{\mathbf{s}}_z)$ and $(\vec{\mathbf{f}}_x|\vec{0})$, then $F_x$ is the combined error of such a fault combination. (If there are more than one fault combination corresponding to $(\vec{0}|\vec{\mathbf{s}}_z)$ and $(\vec{\mathbf{f}}_x|\vec{0})$, a combined error of any of such fault combinations will work.) \label{step:M_1a}
		\item If none of the fault combinations on the list corresponds to $(\vec{0}|\vec{\mathbf{s}}_z)$ and $(\vec{\mathbf{f}}_x|\vec{0})$, then $F_x$ can be any Pauli $X$ operator whose syndrome is $(\vec{0}|\vec{\mathbf{s}}_z)$. \label{step:M_1b}%
	\end{enumerate}
	\item \label{step:M_2} Determine an EC operator $F_z$ using the list of possible fault combinations as follows: 
	\begin{enumerate}
		\item If there is a fault combination on the list whose syndrome and cumulative flag vector are $(\vec{\mathbf{s}}_x|\vec{0})$ and $(\vec{0}|\vec{\mathbf{f}}_z)$, then $F_z$ is the combined error of such a fault combination. (If there are more than one fault combination corresponding to $(\vec{\mathbf{s}}_x|\vec{0})$ and $(\vec{0}|\vec{\mathbf{f}}_z)$, a combined error of any of such fault combinations will work.) \label{step:M_2a}
		\item If none of the fault combinations on the list corresponds to $(\vec{\mathbf{s}}_x|\vec{0})$ and $(\vec{0}|\vec{\mathbf{f}}_z)$, then $F_z$ can be any Pauli $Z$ operator whose syndrome is $(\vec{\mathbf{s}}_x|\vec{0})$. \label{step:M_2b}%
	\end{enumerate}
	\item Apply $F_x\cdot F_z$ to the data qubits to perform error correction. \label{step:M_3}%
	\item If $L$ and $F_x\cdot F_z$ anticommute, modify $m$ from 0 to 1 or from 1 to 0. If $L$ and $F_x\cdot F_z$ commute, do nothing. \label{step:M_4}%
	\item Output $m$ as the operator measurement outcome, where $m=0$ and $m=1$ correspond to $+1$ and $-1$ eigenvalues of $L$. If $L$ is a logical $Z$ operator, the output state is the logical $|0\rangle$ or logical $|1\rangle$ state for $m=0$ or $1$. If $L$ is a logical $X$ operator, the output state is the logical $|+\rangle$ or logical $|-\rangle$ state for $m=0$ or $1$. \label{step:M_5}%
\end{enumerate}	


To verify that the FTM protocol for a capped color code is fault tolerant according to the revised definition (\cref{def:FT_gadget_new}), we will show that both of the properties in \cref{def:FTM_old} is satisfied when the $r$-filter, the ideal decoder, and the distinguishable error set $\mathcal{E}_r$ are defined as in \cref{def:r_filter_new,def:ideal_new,def:dist_err_CSS}. The distinguishable fault set $\mathcal{F}_t$ for this protocol is the same fault set as the one defined for the FTEC protocol (i.e., $\mathcal{F}_t$ concerns the circuits for measuring $X$-type and $Z$-type generators, and does not concern the circuit for measuring $L$). We will also assume that there are no more than $t$ faults during the whole protocol, so the outcome bundles must be repeated $t+1$ times in a row within $(t+1)^2$ rounds. First, suppose that the operator being measured $L$ is a logical $Z$ operator. The analysis will be divided into two cases: (1) the case that the last round of operator and full syndrome measurements has no faults, and (2) the case that the last round of operator and full syndrome measurements has some faults.

(1) Because the last round is correct and the outcome bundles are repeated $(t+1)$ times in a row, $m, \vec{\mathbf{s}}$, and $\vec{\mathbf{f}}$ exactly correspond to the error on the state before Step \ref{step:M_3}. Let $E_\mathrm{in} \in \mathcal{E}_r$ be the input error, 
$E_a$ be the combined error arising from $s_a$ faults in the circuits for measuring $L$, and $E_b$ be the combined error arising from $s_b$ faults in the syndrome measurement circuits, where $r+s_a+s_b \leq t$. Also, assume that the (uncorrupted) input state is $|\bar{m}_\mathrm{in}\rangle$ where $m_\mathrm{in}=0$ or $1$. The data error on the state before the last round is $E_b E_a E_\mathrm{in}$. Since $L$ is of the form $Z^{\otimes n}N$ where $N$ is some stabilizer, the $X$ part of $E_a$ has weight no more than $s_a$, while the $Z$ part of $E_a$ can be any $Z$-type error. We find that the $X$ part of $E_b E_a E_\mathrm{in}$, denoted as $(E_b E_a E_\mathrm{in})_x$, is similar to a combined error of $X$ type of some fault combination in $\mathcal{F}_{r+s_a+s_b}$. However, the $Z$ part of $E_b E_a E_\mathrm{in}$, denoted as $(E_b E_a E_\mathrm{in})_z$, may or may not correspond to a $Z$-type error of some fault combination in $\mathcal{F}_t$. By picking $F_x$ and $F_z$ as in Steps \ref{step:M_1} and \ref{step:M_2}, $F_x$ is logically equivalent to $(E_b E_a E_\mathrm{in})_x$ and $F_z$ is logically equivalent to $(E_b E_a E_\mathrm{in})_z$ or $(E_b E_a E_\mathrm{in})_z Z^{\otimes n}$. So after the error correction in Step \ref{step:M_3}, the output state is $|\bar{m}_\mathrm{in}\rangle$ or $Z^{\otimes n}|\bar{m}_\mathrm{in}\rangle$. Note that $|\bar{m}_\mathrm{in}\rangle$ and $Z^{\otimes n}|\bar{m}_\mathrm{in}\rangle$ are the same state for both $m_\mathrm{in}=0$ and $m_\mathrm{in}=1$ cases (the $-1$ global phase can be neglected in the case of $m_\mathrm{in}=1$).


Next, let us consider the result $m$ obtained from the last round, which tell us whether the state before the measurement of $L$ during the last round is $+1$ or $-1$ eigenstate of $L$. We find that if $m_\mathrm{in}=0$, $m=0$ whenever $E_b E_a E_\mathrm{in}$ commutes with $L$, and $m=1$ whenever $E_b E_a E_\mathrm{in}$ anticommutes with $L$. On the other hand, if $m_\mathrm{in}=1$, $m=1$ whenever $E_b E_a E_\mathrm{in}$ commutes with $L$, and $m=0$ whenever $E_b E_a E_\mathrm{in}$ anticommutes with $L$. Also, note that $F_x\cdot F_z$ is either $E_b E_a E_\mathrm{in}$ or $E_b E_a E_\mathrm{in} Z^{\otimes n}$ and $L$ is a logical $Z$ operator, so $E_b E_a E_\mathrm{in}$ commutes (or anticommutes) with $L$ if and only if $F_x\cdot F_z$ commutes (or anticommutes) with $L$. Thus, we need to flip the output as in Step \ref{step:M_4} whenever $F_x\cdot F_z$ anticommutes with $L$ so that $m=m_\mathrm{in}$. As a result, the measurement protocol gives an output state $|\bar{m}_\mathrm{in}\rangle$ and its corresponding measurement outcome $m=m_\mathrm{in}$ which reflect the uncorrupted input state. 

Now, let us consider the case that the uncorrupted input state is of the form $\alpha|\bar{0}\rangle+\beta|\bar{1}\rangle$. If there is at least one round before the last correct round in which the measurement of $L$ is correct, then the superposition state collapses and the state before the last correct round is either $E_b E_a E_\mathrm{in}|\bar{0}\rangle$ or $E_b E_a E_\mathrm{in}|\bar{1}\rangle$, so the analysis above is applicable. However, if the measurements of $L$ before the last correct round are all incorrect, it is possible that the superposition state may not collapse and the state before the last correct round is of the form $E_b E_a E_\mathrm{in}(\alpha|\bar{0}\rangle+\beta|\bar{1}\rangle)$. Suppose that the measurement of $L$ in the last correct round gives $m=0$. Then the output state from the last correct round is a $+1$ eigenstate of $L$, which is $E_b E_a E_\mathrm{in}|\bar{0}\rangle$ if $E_b E_a E_\mathrm{in}$ commutes with $L$, or $E_b E_a E_\mathrm{in}|\bar{1}\rangle$ if $E_b E_a E_\mathrm{in}$ anticommutes with $L$. In constrast, if the measurement of $L$ in the last correct round gives $m=1$, then the output state from the last correct round is a $-1$ eigenstate of $L$. This state is $E_b E_a E_\mathrm{in}|\bar{1}\rangle$ if $E_b E_a E_\mathrm{in}$ commutes with $L$, or $E_b E_a E_\mathrm{in}|\bar{0}\rangle$ if $E_b E_a E_\mathrm{in}$ anticommutes with $L$. By applying $F_x\cdot F_z$ as in Step \ref{step:M_3} and modifying $m$ whenever $F_x\cdot F_z$ anticommutes with $L$ as in Step \ref{step:M_4}, the outputs from the protocol are either $m=0$ and $|\bar{0}\rangle$, or $m=1$ and $|\bar{1}\rangle$ (up to some global phase). Therefore, both MCP and MPP in \cref{def:FTM_old} are satisfied.

(2) In the case that the last round has some faults, because the outcome bundles are repeated $(t+1)$ times in a row and there are no more than $t$ faults in the protocol, there must be at least one correct round in the last $t+1$ rounds, and the outcome bundles correspond to the error on the state before the last correct round. Let $E_\mathrm{in} \in \mathcal{E}_r$ be the input error, $E_a$ be the combined error arising from $s_a$ faults in the circuits for measuring $L$ before the last correct round, $E_b$ be the combined error arising from $s_b$ faults in the syndrome measurement circuits before the last correct round, and $E_c$ be the combined error arising from $s_c$ faults in any circuits after the last correct round but before the syndrome measurement circuits of the very last round, and $E_d$ be the combined error arising from $s_d$ faults in the syndrome measurement circuits of the very last round, where $r+s_a+s_b+s_c+s_d \leq t$ (see \cref{fig:err_in_FTM}). By an analysis similar to (1), we find that $F_x$ from Step \ref{step:M_1} is logically equivalent to $(E_b E_a E_\mathrm{in})_x$, and $F_z$ from Step \ref{step:M_2} is logically equivalent to $(E_b E_a E_\mathrm{in})_z$ or $(E_b E_a E_\mathrm{in})_z Z^{\otimes n}$. 

Now, let us consider $E_c$ which can arise from the circuits for measuring $L$ or the syndrome measurement circuits, and $E_d$ which can arise from the syndrome measurement circuits. Because the syndromes and the cumulative flag vectors do not change after the last correct round, and because $i$ faults in the circuits for measuring $L$ cannot cause $X$-type error of weight more than $i$, the $X$ part of $E_c$ (denoted as $(E_c)_x$) is similar to the combined error of $X$ type of a fault combination arising from $s_c$ faults whose cumulative flag vector is zero, i.e., $(E_c)_x$ is an error in $\mathcal{E}_{s_c}^x$. In contrast, because the circuits for measuring $L$ can cause $Z$-type error of any weight but the syndromes and the cumulative flag vectors do not change after the last correct round, the $Z$ part of $E_c$ (denoted as $(E_c)_z$) can be written as $(\tilde{E}_c)_z$ or $(\tilde{E}_c)_z Z^{\otimes n}$, where $(\tilde{E}_c)_z \in \mathcal{E}_{s_c}^z$. That is, $E_c$ is either $\tilde{E}_c$ or $\tilde{E}_c Z^{\otimes n}$ where $\tilde{E}_c \in \mathcal{E}_{s_c}$. For $E_d$ which arising from $s_d$ the syndrome measurement circuits in the very last round, we find that it is an error in $\mathcal{E}_{s_d}$ since the cumulative flag vector from the very last round remains the same. 

Let the (uncorrupted) input state be of the form $\alpha|\bar{0}\rangle+\beta|\bar{1}\rangle$. Suppose that the measurement outcome of $L$ from the last correct round is $m=0$. From the argument on a superposition state in (1), we find that the output state from the last correct round is $E_b E_a E_\mathrm{in}|\bar{0}\rangle$ if $E_b E_a E_\mathrm{in}$ commutes with $L$, or $E_b E_a E_\mathrm{in}|\bar{1}\rangle$ if $E_b E_a E_\mathrm{in}$ anticommutes with $L$. Thus, the state before Step \ref{step:M_3} is $E_d \tilde{E}_c E_b E_a E_\mathrm{in}|\bar{0}\rangle$ or $E_d \tilde{E}_c Z^{\otimes n} E_b E_a E_\mathrm{in}|\bar{0}\rangle$ if $E_b E_a E_\mathrm{in}$ commutes with $L$, or $E_d \tilde{E}_c E_b E_a E_\mathrm{in}|\bar{1}\rangle$ or $E_d \tilde{E}_c Z^{\otimes n} E_b E_a E_\mathrm{in}|\bar{1}\rangle$ if $E_b E_a E_\mathrm{in}$ anticommutes with $L$. Recall that $F_x\cdot F_z$ is either $E_b E_a E_\mathrm{in}$ or $E_b E_a E_\mathrm{in} Z^{\otimes n}$, and $E_b E_a E_\mathrm{in}$ commutes (or anticommutes) with $L$ if and only if $F_x\cdot F_z$ commutes (or anticommutes) with $L$. By applying $F_x\cdot F_z$ as in Step \ref{step:M_3} and modifying $m$ whenever $F_x\cdot F_z$ anticommutes with $L$ as in Step \ref{step:M_4}, the protocol either outputs $m=0$ with the output state $E_d \tilde{E}_c|\bar{0}\rangle$ (up to some global phase), or outputs $m=1$ with the output state $E_d \tilde{E}_c|\bar{1}\rangle$ (up to some global phase). Similar results will be obtained in the case that the measurement outcome of $L$ from the last correct round is $m=1$.

Since $E_d \tilde{E}_c \in \mathcal{E}_s$ where $s=s_a+s_b+s_c+s_d$ and $r+s\leq t$, the output bit corresponds to the logical qubit of the output state in every case, and the output bit is 0 (or 1) if the (uncorrupted) input state is $|\bar{0}\rangle$ (or $|\bar{1}\rangle$), both of MCP and MPP in \cref{def:FTM_old} are satisfied. Similar analysis can be made for the case that $L$ is a logical $X$ operator. In that case, we will let $m=0$ and $m=1$ correspond to $|\bar{+}\rangle$ and $|\bar{-}\rangle$, and the analysis similar to (1) and (2) can be applied.

In addition, it is possible to construct an FTP protocol from the FTM protocol described above. For example, if we want to prepare the state $|\bar{0}\rangle$, we can do so by applying the FTM protocol for a logical $Z$ operator to any state, then applying a logical $X$ operator on the output state if $m=1$ or do nothing if $m=0$.

The FTM and the FTP protocols presented in this section is also applicable to any CSS code in which the number of encoded qubit is 1, $\mathcal{F}_t$ is distinguishable (where $\mathcal{F}_t$ corresponds to the circuits for measuring code generators), and the errors in $\mathcal{E}_r^x$ and $\mathcal{E}_r^z$ have the same form for all $r=1,\dots,t$, $t \leq \lfloor(d-1)/2\rfloor$.

Similar to the FTEC protocol for a recursive capped color code, we can construct an FTM protocol for a recursive capped color code similarly to an FTM protocol for a concatenated code. The FTM protocol is as follows:\\

\noindent\textbf{FTM protocol for a recursive capped color code in H form}

Let $L^{(j)}$ be a logical $Z$ (or logical $X$) operator of a recursive capped color code of distance $j$. The following procedure can fault-tolerantly measure $L^{(d)}$ on a recursive capped color code of distance $d$: for $j=3,5,7,\dots,d$, perform $L^{(j)}$ measurement on the first $j$ layers of the recursive capped color code of distance $d$ using the FTM protocol for a capped color code in H form of distance $j$.\\

An FTP protocol for a recursive capped color code is similar to the FTM protocol a recursive capped color code, except that some logical operator will be applied to the output state depending on the measurement outcome so that the desired logical state can be obtained.

\subsection{Transversal Clifford gates}
\label{subsec:other_FT_gadgets}


From the properties of a capped color code in H form discussed in \cref{subsec:CCC_def}, we know that $H$, $S$, and CNOT gates are transversal. These gates can play an important role in fault-tolerant quantum computation because transversal gates satisfy both properties of fault-tolerant gate gadgets originally proposed in \cite{AGP06} (\cref{def:FTG_old}). However, since the definition of fault-tolerant gadgets being used in this work is revised as in \cref{def:FT_gadget_new}, transversal gates which satisfy the old definition may or may not satisfy the new one. In this section, we will show that transversal $H$, $S$, and CNOT gates are still fault tolerant according to the new definition of fault-tolerant gadgets when the distinguishable error set $\mathcal{E}_r$ of a capped (or a recursive capped) color code in H form is defined as in \cref{def:dist_err_CSS}. 

We start by observing the operations of $H$, $S$, and CNOT gates. These gates can transform Pauli operators as follows:
\begingroup
\setlength\arraycolsep{1pt}	
\begin{equation}
	\begin{matrix}
		H: \quad &X &\mapsto &Z, \quad &Y &\mapsto &-Y, \quad &Z &\mapsto &X, \\
		S: \quad &X &\mapsto &Y, \quad &Y &\mapsto &-X, \quad &Z &\mapsto &Z, \\
		\mathrm{CNOT}: \quad &XI &\mapsto &XX, \quad &ZI &\mapsto &ZI, \\
		&IX &\mapsto &IX, \quad &IZ &\mapsto &ZZ.
	\end{matrix} \nonumber
\end{equation}%
\endgroup 
Meanwhile, the transversal $H$, $S$, and CNOT gates can map logical operators $\bar{X}=X^{\otimes n}$ and $\bar{Z}=Z^{\otimes n}$ as follows:
\begingroup
\setlength\arraycolsep{0.7pt}	
\begin{equation}
	\begin{matrix}
		H^{\otimes n}: \quad &\bar{X} &\mapsto &\bar{Z}, \quad &\bar{Z} &\mapsto &\bar{X}, \\
		S^{\otimes n}: \quad &\bar{X} &\mapsto &-\bar{Y}, \quad &\bar{Z} &\mapsto &\bar{Z}, \\
		\mathrm{CNOT}^{\otimes n}: \quad &\bar{X}\otimes \bar{I} &\mapsto &\bar{X} \otimes \bar{X}, \quad &\bar{Z} \otimes \bar{I} &\mapsto &\bar{Z} \otimes \bar{I}, \\
		&\bar{I} \otimes \bar{X} &\mapsto &\bar{I} \otimes \bar{X}, \quad &\bar{I} \otimes \bar{Z} &\mapsto & \bar{Z} \otimes \bar{Z},
	\end{matrix} \nonumber
\end{equation}%
\endgroup 
where $\bar{I}=I^{\otimes n}$, $\bar{Y}=i\bar{X}\bar{Z}=-Y^{\otimes n}$, and $n=3(d^2+1)/2$ is the total number of qubits for each $CCC(d)$ (since $d=3,5,7,...$, we find that $n=3\;(\text{mod}\;4)$ and $\bar{Y}=-Y^{\otimes n}$ for any $CCC(d)$). In addition, the coding subspace is preserved under the operation of $H^{\otimes n}$, $S^{\otimes n}$, or $\mathrm{CNOT}^{\otimes n}$ (i.e., each stabilizer is mapped to another stabilizer). Therefore, $H^{\otimes n}$, $S^{\otimes n}$, and $\mathrm{CNOT}^{\otimes n}$ are logical $H$, logical $S^\dagger$, and logical CNOT gates, respectively.

For an \codepar{n,1,d} recursive capped color code in H form in which $n=(d^3+3d^2+3d-3)/4$, we find that $n=3\;(\text{mod}\;4)$ when $d=3,7,11,\dots$, and $n=1\;(\text{mod}\;4)$ when $d=5,9,13,\dots$. That is, $S^{\otimes n}$ is a logical $S^\dagger$ gate when $d=3,7,11,\dots$, and $S^{\otimes n}$ is a logical $S$ gate when $d=5,9,13,\dots$. $H^{\otimes n}$ and $\mathrm{CNOT}^{\otimes n}$ are logical $H$ and logical CNOT gates for a recursive capped color code in H form of any distance.


Next, we will verify whether the new definition of fault-tolerant gate gadgets in \cref{def:FT_gadget_new} is satisfied. We will start by considering logical $H$ and CNOT gates. Let the distinguishable error set $\mathcal{E}_r$ ($r=1,\dots,t$) be defined as in \cref{def:dist_err_CSS}, where the distinguishable fault set $\mathcal{F}_t$ is the same fault set as the one defined for the FTEC protocol for a capped color code in H form. Suppose that the operation of $H^{\otimes n}$ 
or $\mathrm{CNOT}^{\otimes n}$ has $s$ faults, the input error of $H^{\otimes n}$ 
is an error in $\mathcal{E}_r$ where $r+s \leq t$, and the input error of $\mathrm{CNOT}^{\otimes n}$ is an error in $\mathcal{E}_{r_1}\times \mathcal{E}_{r_2}$ where $r_1+r_2+s \leq t$. The input error for $H^{\otimes n}$ 
can be written as $E_1^x\cdot E_2^z$ where $E_1^x \in \mathcal{E}_r^x$ and $E_2^z \in \mathcal{E}_r^z$, and the input error for $\mathrm{CNOT}^{\otimes n}$ can be written as $(E_3^x\otimes E_4^x) \cdot (E_5^z \otimes E_6^z)$ where $E_3^x \in \mathcal{E}_{r_1}^x$, $E_4^x \in \mathcal{E}_{r_2}^x$, $E_5^z \in \mathcal{E}_{r_1}^z$, $E_6^z \in \mathcal{E}_{r_2}^z$. Let $E_i^x$ and $E_i^z$ be $X$-type and $Z$-type operators which act on the same qubits. We find that,
\begin{enumerate}
	\item $H^{\otimes n}$ maps $E_1^x\cdot E_2^z$ to $E_1^z\cdot E_2^x$, which is an error in $\mathcal{E}_r$.
	\item $\mathrm{CNOT}^{\otimes n}$ maps $(E_3^x\otimes E_4^x) \cdot (E_5^z \otimes E_6^z)$ to $(E_3^x\otimes E_3^xE_4^x)\cdot (E_5^z E_6^z \otimes E_6^z)$, which is an error in $\mathcal{E}_{{r_1}+{r_2}} \times \mathcal{E}_{{r_1}+{r_2}}$.
\end{enumerate}
The operation of a logical $S$ gate can be tricky to analyze since it can map $X$-type errors to a product of $X$- and $Z$-type errors (up to some phase factor). Let us consider a single-qubit error $P \in \{I,X,Y,Z\}$, an error from a single CNOT fault during the measurement of an $X$-type generator which is of the form $P\otimes X^{\otimes m}$, and an error from a single CNOT fault during the measurement of a $Z$-type generator which is of the form $P\otimes Z^{\otimes m}$ (where $m \geq 0$). The operation of $S^{\otimes n}$ will transform such errors as follows (up to some phase factor):
\begin{equation}
	\begin{matrix}
		I\otimes X^{\otimes m}  &\mapsto &(I \otimes X^{\otimes m})\cdot (I \otimes Z^{\otimes m}) \\
		X\otimes X^{\otimes m}  &\mapsto &(X \otimes X^{\otimes m})\cdot (Z \otimes Z^{\otimes m}) \\
		Y\otimes X^{\otimes m}  &\mapsto &(X \otimes X^{\otimes m})\cdot (I \otimes Z^{\otimes m}) \\
		Z\otimes X^{\otimes m}  &\mapsto &(I \otimes X^{\otimes m})\cdot (Z \otimes Z^{\otimes m}) \\
		I\otimes Z^{\otimes m}  &\mapsto &(I \otimes I^{\otimes m})\cdot (I \otimes Z^{\otimes m}) \\
		X\otimes Z^{\otimes m}  &\mapsto &(X \otimes I^{\otimes m})\cdot (Z \otimes Z^{\otimes m}) \\
		Y\otimes Z^{\otimes m}  &\mapsto &(X \otimes I^{\otimes m})\cdot (I \otimes Z^{\otimes m}) \\
		Z\otimes Z^{\otimes m}  &\mapsto &(I \otimes I^{\otimes m})\cdot (Z \otimes Z^{\otimes m})
	\end{matrix} \nonumber
\end{equation}%
We can see that any error from a single fault will be transformed to an error of the form $E_x\cdot E_z$ where $E_x$ and $E_z$ are $X$- and $Z$-type errors from a single fault. For this reason, a combined error from $r$ faults $\mathbf{E}=E_1 \cdots E_r$ will be transformed to $(\bar{S}E_1\bar{S}^\dagger) \cdots (\bar{S}E_r\bar{S}^\dagger)$, which is of the form $\mathbf{E}_x\cdot \mathbf{E}_z$ where $\mathbf{E}_x$ and $\mathbf{E}_z$ are $X$- and $Z$-type errors from $r$ faults. That is, the error after the transformation of $S^{\otimes n}$ is an error in $\mathcal{E}_r$.

In addition, $s$ faults during the application of $H^{\otimes n}$ or $S^{\otimes n}$ can cause an error in $\mathcal{E}_s$, and $s$ faults during the application of $\mathrm{CNOT}^{\otimes n}$ can cause an error in $\mathcal{E}_s \times \mathcal{E}_s$. Combining the input error and the error from faults, we find that an output error from $H^{\otimes n}$ or $S^{\otimes n}$ is an error in $\mathcal{E}_{r+s}$, while an output error from $\mathrm{CNOT}^{\otimes n}$ is an error in $\mathcal{E}_{{r_1}+{r_2}+s} \times \mathcal{E}_{{r_1}+{r_2}+s}$. As a result, $H^{\otimes n}$, $S^{\otimes n}$, and $\mathrm{CNOT}^{\otimes n}$ satisfy both GCP and GPP in \cref{def:FTG_old} when the $r$-filter, the ideal decoder, and the distinguishable error set are defined in \cref{def:r_filter_new,def:ideal_new,def:dist_err_CSS}. That is, transversal $H$, $S$, and CNOT gates are fault tolerant according to the revised definition. Similar analysis is also applicable to a recursive capped color code in H form. Since the Clifford group can be generated by $H$, $S$, and CNOT \cite{CRSS97,Gottesman98b}, any Clifford gate can be fault-tolerantly implemented on a capped (or a recursive capped) color code in H form using transversal $H$, $S$, and CNOT gates.

(Note that whether a transversal gate satisfies the revised definition of fault-tolerant gate gadgets in \cref{def:FT_gadget_new} depends on how the distinguishable error set is defined (as in either \cref{def:dist_err} or \cref{def:dist_err_CSS}). For example, if the input error $E_\textrm{in}$ can arise from $t$ faults ($E_\textrm{in}$ is in $\mathcal{E}_t$) and a transversal gate transforms such an error to another error $E_\textrm{out}$ which cannot arise from $\leq t$ faults ($E_\textrm{out}$ is not in $\mathcal{E}_t$), then this transversal gate is not considered fault tolerant.)

\subsection{Fault-tolerant implementation of a logical $T$ gate via code switching}
\label{subsec:FT_T_gate}






In order to achieve a universal set of quantum gates, we also need a fault-tolerant implementation of some gate outside the Clifford group \cite{NRS01}. One possible way to implement a non-Clifford gate on the capped color code in H form is to use magic state distillation \cite{BK05}, but large overhead might be required \cite{FMMC12}. Another possible way is to perform code switching; since the code in H form possesses transversal $H$, $S$, and CNOT gates, and the code in T form possesses a transversal $T$ gate, we can apply transversal $H$, $S$, or CNOT gates and perform FTEC on the code in H form, and switch to code in T form to apply a transversal $T$ gate when necessary. However, logical $T$ gate implementation via code switching on a capped color code might not be fault tolerant since the code in T form constructed from $CCC(d)$ has distance 3 regardless of the parameter $d$, and a few faults occurred to the code in T form can cause a logical error. Fortunately, for a recursive capped color code, both distances of the code in H form and the code in T form constructed from $RCCC(d)$ are $d$. Thus, fault-tolerant $T$ gate implementation via code switching is possible. The fault-tolerant protocol for logical $T$ gate implementation on a recursive capped color code will be developed in this section.

First, let us assume that the $T$-gate implementation protocol is performed after the FTEC protocol for a recursive capped color code in H form developed in \cref{subsec:FTEC_ana} and the following CNOT orderings are being used:
\begin{enumerate}
	\item In the preceding FTEC protocol, the $\mathtt{f}$, $\mathtt{v}$, and $\mathtt{cap}$ operators on the $(j-2)$-th, $(j-1)$-th, and $j$-th layers of the recursive capped color code are measured using the CNOT orderings for the $\mathtt{f}$, $\mathtt{v}$, $\mathtt{cap}$ operators of a capped color code of in H form of distance $j$ ($j=3,5,...,d$) which give a distinguishable fault set (where an operator on $\mathtt{q_0}$ of a capped color code is replaced by operators on all qubits on the $(j-2)$-th layer of a recursive capped color code).
	\item During the switching from the code in H form to T form, all $Z$-type vertical face generators $e^z_i$ are measured using flag circuits with one flag ancilla similar to the circuit in \cref{fig:vertical_op} (see the definition of vertical face generators in \cref{subsec:CCC_def,subsec:RCCC_def}).
	\item During the switching from the code in T form to H form, all $X$-type generators of 2D color codes on layers $2,4,...,d-1$ of the code are measured using circuits similar to those being used in the preceding FTEC protocol.
\end{enumerate}

\begin{figure}[tbp]
	\centering
	\begin{subfigure}[b]{0.30\textwidth}
		\includegraphics[width=\textwidth]{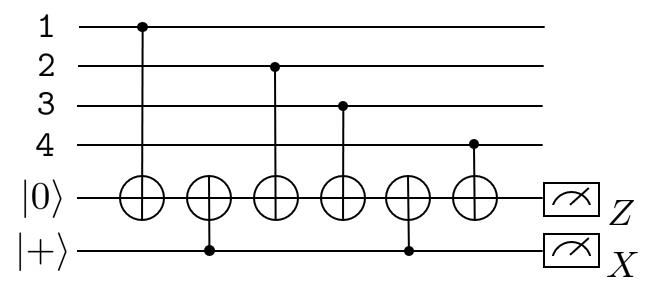}
		\captionsetup{justification=centering}
		\caption{}
		\label{subfig:vertical_circuit}
	\end{subfigure}	
	\begin{subfigure}[b]{0.12\textwidth}
		\includegraphics[width=\textwidth]{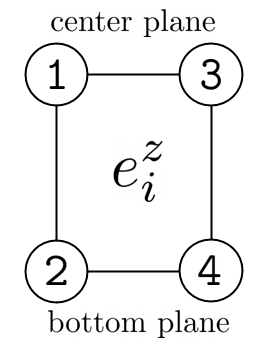}
		\captionsetup{justification=centering}
		\caption{}
		\label{subfig:vertical_ordering}
	\end{subfigure}
	\caption{(a) A flag circuit for measuring a vertical face generator $e_i^z$. (b) The ordering of data CNOTs in the circuit for each $e_i^z$.}
	\label{fig:vertical_op}
\end{figure}

The logical $T$-gate implementation protocol will use the following ideas: we will start from the recursive capped color code in H form, switch to the code in T form, apply a transversal $T$ gate, switch back to the code in H form, then perform error correction using an FTEC protocol similar to the FTEC protocol for a recursive capped color code in H form, except that possible faults from $e_i^z$ measurements are also included in the distinguishable fault set (note that we will never perform error correction on the code in T form). The full procedure of the $T$-gate implementation protocol is as follows:\\

\noindent\textbf{Fault-tolerant $T$-gate implementation protocol for a recursive capped color code in H form}

\begin{enumerate}
	\item During a single round of operator measurements, measure all $Z$-type vertical face generators. Perform measurements until the outcomes are repeated $t+1$ times in a row. After that, apply a Pauli operator corresponding to the repeated measurement outcome (see also the code switching procedure in \cref{subsec:CCC_def,subsec:RCCC_def}). \label{step:T_1}
	\item Perform a logical $T$ operation by applying physical $T$ and $T^\dagger$ gates on qubits represented by black and white vertices, respectively (see also \cref{prop:transversal_T_recur}). \label{step:T_2}
	\item During a single round of operator measurements, measure all $X$-type generators of 2D color codes on layers $2,4,...,d-1$ of the code. Perform measurements until the outcomes are repeated $t+1$ times in a row. After that, apply a Pauli operator corresponding to the repeated measurement outcome (see also the code switching procedure in \cref{subsec:CCC_def,subsec:RCCC_def}). \label{step:T_3}
	\item Perform error correction using an FTEC protocol similar to the FTEC protocol for a recursive capped color code in H form described in \cref{subsec:FTEC_ana}, except that possible faults from vertical face generator measurements are also included in the distinguishable fault set. \label{step:T_4}
\end{enumerate}

We can show that the protocol described above is fault tolerant using the following facts:
\begin{enumerate}
	\item Both codes in H form and T form have distance $d$ (in fact, the distance of $RCCC(d)$ does not depend on the gauge choice).
	\item An input error to the logical $T$-gate implementation protocol is an error $E_\mathrm{in}$ in the distinguishable set $\mathcal{E}_r$, where $r$ is the number of faults in the preceding FTEC protocol.
	\item During the switching from the code in H form to the code in T form (Step \ref{step:T_1}), the flag outcome is not zero whenever a single fault that leads to a data error of weight 2 occurs. That is, when the flag is zero, $s_1$ faults will lead to an error $E_1$ of weight $\leq s_1$.
	\item A logical $T$ gate is transversal, so $s_2$ faults during Step \ref{step:T_2} will lead to an error $E_2$ of weight $\leq s_2$.
	\item Any fault that can occur during the switching from the code in T form to the code in H form (Step \ref{step:T_3}) will lead to an error on layer 2,4,..., or $d-1$ (a center plane of inner $CCC(j)$, $j=3,5,7,...$).
	\item The gauge measurements and Pauli operation during the code switching correct the part of the data error that acts on the gauge qubits being measured. The code switching does not affect the part of the data error that acts on the logical qubit.
\end{enumerate}

Consider the data error $E_3E_2\bar{T}E_1E_\mathrm{in}\bar{T}^\dagger$ (the total error on the desired state $\bar{T}|\bar{\psi}_\mathrm{in}\rangle$). We can show that when $r+s_1+s_2+s_3\leq t$, $E_3E_2\bar{T}E_1E_\mathrm{in}\bar{T}^\dagger$ is correctable by the FTEC protocol in Step \ref{step:T_4}; this is equivalent to showing that $(\bar{T}E'^\dagger_\mathrm{in}E'^\dagger_1\bar{T}^\dagger E'^\dagger_2E'^\dagger_{3})(E_3E_2\bar{T}E_1E_\mathrm{in}\bar{T}^\dagger)$ is not a logical operator with zero cumulative flag vector when $r+r'+s_1+s'_1+s_2+s'_2+s_3+s'_3 \leq 2t$ (using a technique similar to the proof of \cref{thm:main2}). In addition, we know from the analysis of the FTEC protocol in \cref{subsec:FTEC_ana} that if the FTEC protocol in Step \ref{step:T_4} can correct any possible error after Step \ref{step:T_3} whenever $s_4=0$, then in case that $s_4 \leq t$, the output error will be an error in $\mathcal{E}_{s_4}$.

We point out that the protocol described in this section works for a recursive capped color code in H form of any distance given that \emph{flag circuits} are used in the gauge operator measurements during the code switching. Note that for the recursive capped color codes in H form for distance 3 and 5, it is possible to obtain a distinguishable fault set when the circuits for generator measurements are \emph{non-flag circuits} (thus, FTEC, FTP, FTM, and fault-tolerant Clifford computation with one ancilla are possible). In that case, however, an additional ancilla is required if one wants to perform logical $T$ gate implementation via code switching using the fault-tolerant protocol provided in this section.

\section{Discussion and conclusions}
\label{sec:discussions}


In this work, we observe that errors arising from a few faults depend on the structure of the circuits chosen for syndrome measurement, and develop an FTEC protocol accordingly. A fault set which includes all possible fault combinations arising from at most a certain number of faults is said to be distinguishable if any pair of fault combinations in the set either lead to logically equivalent data errors, or lead to different syndromes or cumulative flag vectors (as defined in \cref{def:distinguishable}). Distinguishability may depend on the number of flag ancillas being used in the circuits, the ordering of gates in the circuits, and the choice of stabilizer generators being measured. If we can find a set of circuits for a stabilizer code which leads to a distinguishable fault set, we can construct an FTEC protocol, as shown in \cref{subsec:FTEC_ana}.


We prove in \cref{lem:err_equivalence} that if an \codepar{n,k,d} CSS code has odd $n$, $k=1$, even weight stabilizer generators, and logical $X$ and $Z$ being $X^{\otimes n}$ and $Z^{\otimes n}$, then two Pauli errors of $X$ type (or $Z$ type) with the same syndrome are logically equivalent if and only if they have the same weight parity. 
One may notice that the weight parity of a Pauli operator and the anticommutation between the Pauli operator and a logical operator are closely related. In fact, for a given stabilizer code, the normalizer group can be generated by the stabilizer generators of the code and all independent logical Pauli operators; for example, the normalizer group of the Steane code is $N(S)=\langle g_i^x,g_i^z,X^{\otimes 7},Z^{\otimes 7}\rangle_{i=1,2,3}$. If the anticommutation between a Pauli error $E$ and each of the generators of $N(S)$ can be found, then a Pauli error logically equivalent to $E$ can be determined with certainty. The EC techniques presented in \cite{TL20} and this work use the fact that the weight parity of an error on a smaller code (or the anticommutation between the error and a logical operator of a smaller code) can be inferred by the measurement results of the stabilizer generators of a bigger code. We are hopeful that the relationship between the weight parity and the anticommutation can lead to EC techniques similar to the weight parity technique for a general stabilizer code in which the number of logical qubits can be greater than 1.

With \cref{lem:err_equivalence} in mind, we present the 3D color code of distance 3 in \cref{sec:3D_code} and construct a family of capped color codes in \cref{sec:CCC}, which are good candidates for our protocol construction (the 3D color code of distance 3 is the smallest capped color code). A capped color code is a subsystem code; it can be transformed to stabilizer codes, namely capped color codes in H form and T form, by the gauge fixing method. The code in H form has transversal $H$, $S$, and CNOT gates, while the code in T form has transversal CNOT and $T$ gates.
One interesting property of a capped color code in H form is that the code contains a 2D color code as a subcode lying on the center plane. Since a $\mathtt{cap}$ generator of $X$ type (or $Z$ type) has support on all qubits on the center plane, the weight parity of an error of $Z$ type (or $X$ type) occurred on the center plane can be obtained from the measurement result of the $\mathtt{cap}$ generator. The syndrome of the error on the center plane corresponding to the measurements of the 2D code generators together with the error weight parity can lead to an EC operator for such an error by \cref{lem:err_equivalence}. Exploiting these facts, we design circuits for measuring generators of a capped color code such that most of the possible errors are on the center plane. We prove in \cref{thm:main} that if the circuits satisfy some conditions, the fault set corresponding to all possible fault combinations arising from up to $t=(d-1)/2$ faults is distinguishable, where $d=3,5,7,...$ is the distance of the code. 
Furthermore, we prove in \cref{thm:main2,thm:main3} that a distinguishable fault set for a capped color code in H form of \emph{any distance} can be obtained, given that circuits for measuring code generators are flag circuits with one flag ancilla of a particular form. We also show that for the codes of distance 3 and 5, it is possible to obtain a distinguishable fault set using non-flag circuits with specific CNOT orderings. However, whether such non-flag circuits exist for the code of distance 7 or higher is still not known.

Besides capped color codes, we also construct a family of recursive capped color codes in \cref{sec:CCC}. A recursive capped color code $RCCC(d)$ can be obtained by recursively encoding the top qubit of a capped color code $CCC(d)$ by capped color codes of smaller distances. Similar to a capped color code, stabilizer codes namely recursive capped color codes in H form and T form can be obtained by gauge fixing method. Circuits for measuring code generators which work for capped color codes are also applicable to recursive capped color codes. The main advantage of a recursive capped color code is that both codes in H form and T form have the same distance, allowing us to perform fault-tolerant logical $T$ gate implementation via code switching.

In \cref{sec:FT_protocol}, we construct several fault-tolerant protocols using the fact that the fault set corresponding to the protocols being used is distinguishable. Our definitions of fault-tolerant gadgets in \cref{def:FT_gadget_new} also take the fact that some errors can be distinguished by their relevant flag information, so they can be viewed as a generalization of the definitions of fault-tolerant gadgets proposed in \cite{AGP06} (\cref{def:FTG_old,def:FTEC_old,def:FTP_old,def:FTM_old}). Our protocols are not limited to the capped or the recursive capped color codes; some of the protocols are also applicable to other families of stabilizer codes if their syndrome measurement circuits give a distinguishable fault set. Since possible errors depend on every fault-tolerant gadget being used, all protocols for quantum computation (including error correction, gate, measurement, and state preparation gadgets) must be designed in tandem in order to achieve fault tolerance. 

In our development, the ideal decoder and the $r$-filter (which define fault-tolerant gadgets) are defined by a distinguishable error set in which errors correspond to fault combinations with zero cumulative flag vector (see  \cref{def:dist_err,def:r_filter_new,def:ideal_new,def:FT_gadget_new}). The intuition behind the definitions with zero cumulative flag vector is that in a general flag FTEC protocol, we normally repeat the measurements until the outcomes (syndromes and flag vectors) are repeated $t+1$ times in a row. Thus, undetectable faults at the very end of the protocol which give repeated outcomes must correspond to the zero cumulative flag vector (see the analysis of the FTEC protocol in \cref{subsec:FTEC_ana} for more details). Note that nontrivial cumulative flag vectors are used to distinguish possible fault combinations arising during the FTEC protocol only; the flag information is used locally in each FTEC gadget and is not passed on to other gadgets. One interesting future direction would be studying how fault-tolerant protocols can be further improved by exploiting the flag information outside of the FTEC protocol. For example, we may define both ideal decoder and $r$-filter using fault combinations with trivial or nontrivial cumulative flag vectors. However, when an FTEC protocol is allowed to output nontrivial flag information, we have to make sure that
other subsequent 
fault-tolerant gadgets (such as FTG gadgets) must be able to process the flag information in the way that their possible output errors are still distinguishable. This study is beyond the scope of this work.


\begin{table*}[htbp]
	\begin{center}
		\begin{tabular}{| c | c | c | c |}
			\hline
			 & Number of data & Number of data & Number of data \\
			Code family & qubits only & and ancilla qubits & and ancilla qubits \\
			& $n(d)$ & assuming (1) and (2.a) & assuming (1) and (2.b) \\
			\hline
			2D color code \cite{BM06} & $3(d^2-1)/4+1$ & $3(d^2-1)/4+3$ & $3(d^2-1)/2+1$ \\
			\hline
			Capped color code & $3(d^2-1)/2+3$ & $3(d^2-1)/2+5$ & $9(d^2-1)/4+5$ \\
			\hline
			Recursive capped color code & $(d^3+3d^2+3d-3)/4$ & $(d^3+3d^2+3d+5)/4$  & $(3d^3+9d^2+13d-17)/8$  \\
			\hline
			3D color code \cite{Bombin15} & $(d^3+d)/2$ & $(d^3+d+2)/2$* & $(7d^3+3d^2+5d-3)/12$* \\
			\hline
			Stacked code \cite{JB16} & $(3d^3-3d^2+d+3)/4$ & $(3d^3-3d^2+d+7)/4$* & $(15d^3-15d^2+9d+7)/16$* \\
			\hline
		\end{tabular}
	\end{center}
	\caption{Comparison between the numbers of required qubits for a 2D color code, a capped color code (in H form), a recursive capped color code, a traditional 3D color code, and a stacked code of distance $d$. The assumptions being used in the third and the fourth columns are (1) qubit preparation and qubit measurement are fast, and (2.a) all-to-all connectivity between data and ancilla qubits are allowed or (2.b) there are dedicated syndrome and flag ancillas for each generator measurement. *We still do not know the actual minimum number of required ancillas for a 3D color code and a stacked code to achieve fault tolerance. The numbers for these codes in the table are for the case that only one ancilla per generator is required.}
	\label{tab:code_compare}
\end{table*}

One should note that it is possible to use fault-tolerant protocols satisfying the old definitions of fault-tolerant gadgets (\cref{def:FTG_old,def:FTEC_old,def:FTP_old,def:FTM_old}) in conjunction with fault-tolerant protocols satisfying our definitions of fault-tolerant gadgets (\cref{def:FT_gadget_new}). In particular, observe that for any Pauli error of weight $w\leq t$, we can always find a fault combination arising from $w$ faults whose combined error is such an error; i.e., any Pauli error of weight up to $t$ is contained in a distinguishable fault set $\mathcal{F}_t$. Therefore, an FTEC protocol satisfying \cref{def:FT_gadget_new} can be used to correct an output error of any fault-tolerant protocol satisfying one of the old definitions (assuming that both protocols can tolerate the same number of faults). However, the converse might not be true since an FTEC protocol satisfying the old definition of FTEC gadget might not be able to correct errors of high weight arising from a small number of faults in the protocol satisfying \cref{def:FT_gadget_new}.


In this work, we show that universal quantum computation can be performed fault-tolerantly on a recursive capped color code in H form of any distance; First, we provide FTEC, FTM, and FTP protocols for a capped color code in H form which are applicable to the code of any distance as long as the fault set is distinguishable (see \cref{subsec:FTEC_ana,subsec:FTM_ana}). From the aforementioned protocols, we can construct FTEC, FTM, and FTP protocols for a recursive capped color code in H form similarly to conventional fault-tolerant protocols for a concatenated code.
Second, we show that for a capped color code, transversal $H$, $S$, and CNOT gates are fault tolerant according to our revised definitions of fault-tolerant gadgets in \cref{def:FT_gadget_new} (see \cref{subsec:other_FT_gadgets}), and similar analysis is also applicable to a recursive capped color code. Last, we provide a fault-tolerant protocol for implementing a logical $T$ gate on a recursive capped color code in H form via code switching, which is applicable to the code of any distance given that circuits for measuring gauge operators are flag circuits of a particular form (see \cref{subsec:FT_T_gate}).

Compared with other codes with the same distance, capped and recursive capped color codes may not have the fewest number of data qubits. Nevertheless, these codes have some special properties which may be useful for fault-tolerant quantum computation. The numbers of data qubits $n$ (as functions of the code distance $d$) for the families of 2D color codes \cite{BM06}, capped color codes, recursive capped color codes, traditional 3D color codes \cite{Bombin15}, and stacked codes \cite{JB16} are provided in the second column of \cref{tab:code_compare}. We can observe the followings:
\begin{enumerate}
	\item The number of data qubits required for a capped color code in H form is about twice of that of a 2D color code of the same distance (both numbers are $O(d^2)$). One advantage that capped color codes have over 2D color codes is that a logical $T$ gate can be implemented on the capped color codes via code switching. Although the process might not be fully fault tolerant (because the codes in T form has distance 3 regardless of $d$), code switching uses fewer ancillas compared to magic state distillation and may be beneficial if the error rate is low enough.
	\item When $d$ is large, the number of data qubits required for a recursive capped color code is about two times smaller than that of a 3D color code, and about three times smaller than that of a stacked code of the same distance (all numbers are $O(d^3)$). For these three families of codes, a logical $T$ gate can be fault-tolerantly implemented via code switching since the code distance does not depend on the gauge choice.
\end{enumerate}


Using our fault-tolerant protocols, Clifford computation on a capped color code in H form of any distance can be achieved using only 2 ancillas, while universal quantum computation on a recursive capped color code in H form of any distance can be achieved using only 2 ancillas. This is equal to the number of ancillas required for fault-tolerant protocols for Clifford computation on a 2D color code of any distance (by \cref{thm:main2}). It should be noted that the aforementioned results on the number of ancillas are under the assumption that (1) qubit preparation and qubit measurement are fast enough so that the ancillas can be reused, and (2.a) all-to-all connectivity between data and ancilla qubits are allowed. In practice, attaining the minimum number of ancillas can be challenging because the qubit connectivity is restricted to the nearest neighbor interactions in most architectures. A more practical assumption is (2.b) having dedicated syndrome and flag ancillas for each stabilizer generator measurement. For a 2D color code, syndrome and flag ancillas can be shared between $X$-type and $Z$-type generators acting on the same set of qubits, so the number of required ancillas is the number of stabilizer generators, which is $3(d^2-1)/4$. For a capped (or recursive capped) color code in H form, syndrome and flag ancillas can be shared between $X$-type and $Z$-type volume ($\mathtt{v}$) generators acting on the same set of qubits, and face ($\mathtt{f}$) generators can share ancillas with their corresponding volume generators. Thus, the number of required ancillas is equal to the number of stabilizer generators of the subsystem code $CCC(d)$ (or $RCCC(d)$). 
(Note that on color codes, generators of the same color can be measured in parallel.)

The total numbers of data and ancilla qubits required for 2D color codes, capped color codes, recursive capped color codes, traditional 3D color codes, and stacked codes under assumptions (1) and (2.b) are displayed in the fourth column of \cref{tab:code_compare}. (Note that we still do not know the actual minimum number of required ancillas for a 3D color code and a stacked code to achieve fault tolerance. The numbers for these two codes in the table are for the case that only one ancilla per generator is required.) We find the followings: 
\begin{enumerate}
	\item In exchange for having $T$-gate implementation via code switching available (although the process is not fully fault tolerant), protocols for a capped color code in H form require about 50 percent more qubits than those for a 2D color code of the same distance.
	\item When $d=5$, the recursive capped color code outperforms the stacked code and is comparable to the 3D color code. When $d\geq 7$, the recursive capped color code outperforms both 3D color code and stacked code. (These three codes are the same code when $d=3$.)
\end{enumerate}

Recently, Beverland, Kubica, and Svore \cite{BKS21} compare the overhead required for $T$ gate implementation with two methods: using a 2D color code via magic state distillation versus using a (traditional) 3D color code via code switching. They found that magic state distillation outperforms code switching except at some low physical error rate and when certain fault-tolerant schemes are used in the simulation.
%
Since our protocols require only a few ancillas per generator and the data block of a recursive capped color code is smaller than that of a 3D color code of the same distance, we are hopeful that the range of physical error rate in which code switching beats magic state distillation could be improved by our protocols. A careful simulation on the overhead is required, thus we leave this for future work.




Last, we point out that our fault-tolerant protocols using the flag and the weight parity techniques are specially designed for the \emph{circuit-level noise} so that all possible data errors arising from a few faults (including any 1- and 2-qubit gate faults, faults during the ancilla preparation and measurement, and faults during wait time) can be corrected. However, our protocols require repeated syndrome measurements in order to avoid syndrome bit flips which may occur during the protocols, and the processes can increase the number of gate operations. The single-shot error correction \cite{Bombin15b} is one technique that can deal with the syndrome bit flips without using repeated syndrome measurements. We hope that the flag, the weight parity, and the single-shot error correction techniques could be used together to build fault-tolerant protocols which can protect the data against the circuit-level noise and require only small numbers of gates and ancillas.

\section{Acknowledgements}
\label{sec:acknowledgement}
We thank Christopher Chamberland and Michael Vasmer for the suggestion of considering weight parity error correction on color codes. We also thank Ken Brown, Robert Calderbank, Arun Aloshious, Rui Chao, Shilin Huang, Eric Sabo, and other members of Duke Quantum Center for helpful discussion on similarities between capped color codes and stacked codes, the alternative proof of \cref{thm:main2} presented in this work, and possible future directions. We would like to extend our gratitude to Dan Browne, Andrew Cross, David Gosset, Raymond Laflamme, Ted Yoder, and Beni Yoshida for their helpful comments and suggestions. T.T. acknowledges the support of The Queen Sirikit Scholarship under The Royal Patronage of Her Majesty Queen Sirikit of Thailand. D.L. is supported by an NSERC Discovery grant. Perimeter Institute is supported in part by the Government of Canada and the Province of Ontario.

\bibliographystyle{ieeetr}
\bibliography{bibtex_FT_Thesis}

\newpage

\appendix

\section{Proof of Theorem 1}
\label{app:proof_main_thm}%

In the first part of the proof, we will assume that data errors arising from all faults are purely $Z$ type, and show that if Conditions \ref{con:con1} to \ref{con:con5} are satisfied, then there is no fault combination arising from up to $d-1$ faults whose combined error is a logical $Z$ operator and its cumulative flag vector is zero. Because $i$ faults during the measurements of $X$-type generators cannot cause a $Z$-type error of weight more than $i$, we can assume that each fault is either a qubit fault causing a $Z$-type error (which is $\mathtt{q_0}$, $\mathtt{q_{on}}$, or $\mathtt{q_{off}}$ fault), or a fault during a measurement of some $Z$-type generator (which is $\mathtt{f}$, $\mathtt{v}$, $\mathtt{v^*}$, or $\mathtt{cap}$ fault).

First, recall the main equations (in mod 2):
\begin{align}
	s_\mathtt{cap} =& n_{0}+n_\mathtt{on}+\sum\mathrm{wp}(\sigma_\mathtt{f})+\sum\mathrm{wp}(\sigma_\mathtt{v}) \nonumber \\
	&+\sum\mathrm{wp}(\sigma_\mathtt{v^*,cen})+\sum\mathrm{wp}(\sigma_\mathtt{cap}), \label{eq:mainA1}\\
	\vec{s}_\mathtt{f} =& \sum \vec{q}_\mathtt{on} + \sum \vec{p}_\mathtt{f} + \sum \vec{p}_\mathtt{v} + \sum \vec{p}_\mathtt{v^*,cen} \nonumber \\ &+ \sum \vec{p}_\mathtt{cap}, \label{eq:mainA2}\\
	\vec{s}_\mathtt{v} =& \sum \vec{q}_\mathtt{on}+\sum \vec{q}_\mathtt{off}+\sum \vec{p}_\mathtt{f}+\sum \vec{q}_\mathtt{v^*} \nonumber \\ &+ \sum \vec{p}_\mathtt{cap}, \label{eq:mainA3}\\
	\mathrm{wp}_\mathtt{tot} =& n_{0}+n_\mathtt{on}+n_\mathtt{off}+\sum \mathrm{wp}(\sigma_\mathtt{f})+n_\mathtt{v^*} \nonumber \\ &+ \sum\mathrm{wp}(\sigma_\mathtt{cap}), \label{eq:mainA4}\\
	\vec{\mathbf{f}}_\mathtt{cap} =& \sum \vec{f}_\mathtt{cap}, \label{eq:mainA5}\\
	\vec{\mathbf{f}}_\mathtt{f} =& \sum \vec{f}_\mathtt{f}, \label{eq:mainA6}\\
	\vec{\mathbf{f}}_\mathtt{v} =& \sum \vec{f}_\mathtt{v}+\sum \vec{f}_\mathtt{v^*}, \label{eq:mainA7}\\
	\mathrm{wp}_\mathtt{bot}=&n_\mathtt{off}+\sum \mathrm{wp}(\sigma_\mathtt{v})+\sum \mathrm{wp}(\sigma_\mathtt{v^*,bot}), \label{eq:mainA8}\\
	\vec{s}_\mathtt{bot} =& \sum \vec{q}_\mathtt{off}+\sum \vec{p}_\mathtt{v} + \sum \vec{p}_\mathtt{v^*,bot}. \label{eq:mainA9}
\end{align}

Note that the types of faults involved in the main equations and the types of faults involved in the conditions are related by the correspondence in \cref{tab:2D-3D}. Here we will show that if Conditions \ref{con:con1} to \ref{con:con5} are satisfied and there exists a fault combination arising from up to $d-1$ faults which corresponds to a logical $Z$ operator and the zero cumulative flag vector, some contradictions will happen (also note that Condition \ref{con:con0} is automatically satisfied). By \cref{lem:err_equivalence}, a fault combination corresponding to a logical $Z$ operator and the zero cumulative flag vector gives $s_\mathtt{cap}=0$, $\vec{s}_\mathtt{f}=\vec{0}$, $\vec{s}_\mathtt{v}=\vec{0}$, $\mathrm{wp}_\mathtt{tot}=1$, $\vec{\mathbf{f}}_\mathtt{cap}=\vec{0}$, $\vec{\mathbf{f}}_\mathtt{f}=\vec{0}$, $\vec{\mathbf{f}}_\mathtt{v}=\vec{0}$, $\mathrm{wp}_\mathtt{bot}=1$, and $\vec{s}_\mathtt{bot}=0$. We will divide the proof into 4 cases: (1) $n_\mathtt{f}=0$ and $n_\mathtt{cap}=0$, (2) $n_\mathtt{f} \geq 1$ and $n_\mathtt{cap}=0$, (3) $n_\mathtt{f} = 0$ and $n_\mathtt{cap}\geq 1$, and (4) $n_\mathtt{f} \geq 1$ and $n_\mathtt{cap}\geq 1$. 


\pagebreak

\textit{Case 1}: $n_\mathtt{f}=0$ and $n_\mathtt{cap}=0$. The main equations can be simplified as follows (trivial equations are neglected):
\begin{align}
	0 =& n_{0}+n_\mathtt{on}+\sum\mathrm{wp}(\sigma_\mathtt{v}) +\sum\mathrm{wp}(\sigma_\mathtt{v^*,cen}), \tag{A1}\\
	\vec{0} =& \sum \vec{q}_\mathtt{on} + \sum \vec{p}_\mathtt{v} + \sum \vec{p}_\mathtt{v^*,cen}, \tag{A2} \\
	\vec{0} =& \sum \vec{q}_\mathtt{on}+\sum \vec{q}_\mathtt{off}+\sum \vec{q}_\mathtt{v^*}, \tag{A3}\\
	1 =& n_{0}+n_\mathtt{on}+n_\mathtt{off}+n_\mathtt{v^*}, \tag{A4}\\
	\vec{0} =& \sum \vec{f}_\mathtt{v}+\sum \vec{f}_\mathtt{v^*}, \tag{A7}\\
	1=&n_\mathtt{off}+\sum \mathrm{wp}(\sigma_\mathtt{v})+\sum \mathrm{wp}(\sigma_\mathtt{v^*,bot}), \tag{A8} \\
	\vec{0} =& \sum \vec{q}_\mathtt{off}+\sum \vec{p}_\mathtt{v} + \sum \vec{p}_\mathtt{v^*,bot}. \tag{A9}
\end{align}

All faults involved in \cref{eq:mainA3,eq:mainA4} correspond to $\mathtt{q_{2D}}$ faults on the 2D code and the total number of faults are at most $d-1$. Because Condition \ref{con:con0} is satisfied, from \cref{eq:mainA3,eq:mainA4}, we must have that $n_\mathtt{on}+n_\mathtt{off}+n_\mathtt{v^*}=0$ (mod 2) which implies that $n_{0}=1$. Thus, \cref{eq:mainA1} becomes,
\begin{equation}
	1 = n_\mathtt{on}+\sum\mathrm{wp}(\sigma_\mathtt{v}) +\sum\mathrm{wp}(\sigma_\mathtt{v^*,cen}). \tag{A1}\\
\end{equation}

Since the total number of faults are $n_{0}+n_\mathtt{on}+n_\mathtt{off}+n_\mathtt{v}+n_\mathtt{v^*} \leq d-1$, we find that $n_\mathtt{on}+n_\mathtt{off}+n_\mathtt{v}+n_\mathtt{v^*} \leq d-2$. Let us consider the following cases:

(1.a) If $n_\mathtt{off}=0$, we have $n_\mathtt{v}+n_\mathtt{v^*} \leq d-2-n_\mathtt{on} \leq d-2$. In this case, \cref{eq:mainA7,eq:mainA8,eq:mainA9} contradict Condition \ref{con:con1} (where $\mathtt{v}$ and $\mathtt{v^*}$ faults correspond to $\mathtt{f_{2D}}$ fault).

(1.b) If $n_\mathtt{off}\geq 1$, we have $n_\mathtt{on}+n_\mathtt{v}+n_\mathtt{v^*} \leq d-2-n_\mathtt{off} \leq d-3$. In this case, \cref{eq:mainA1,eq:mainA2,eq:mainA7} contradict Condition \ref{con:con2} (where $\mathtt{q_{on}}$ fault corresponds to $\mathtt{q_{2D}}$ fault, and $\mathtt{v}$ and $\mathtt{v^*}$ faults correspond to $\mathtt{f_{2D}}$ fault).\\

\textit{Case 2}: $n_\mathtt{f} \geq 1$ and $n_\mathtt{cap}=0$. The main equations can be simplified as follows:
\begin{align}
	0 =& n_{0}+n_\mathtt{on}+\sum\mathrm{wp}(\sigma_\mathtt{f})+\sum\mathrm{wp}(\sigma_\mathtt{v}) \nonumber \\
	&+\sum\mathrm{wp}(\sigma_\mathtt{v^*,cen}), \tag{A1}\\
	\vec{0} =& \sum \vec{q}_\mathtt{on} + \sum \vec{p}_\mathtt{f} + \sum \vec{p}_\mathtt{v} + \sum \vec{p}_\mathtt{v^*,cen}, \tag{A2}
\end{align}	
\vspace*{-1.2cm}

\begin{align}	
	\vec{0} =& \sum \vec{q}_\mathtt{on}+\sum \vec{q}_\mathtt{off}+\sum \vec{p}_\mathtt{f}+\sum \vec{q}_\mathtt{v^*}, \tag{A3}\\
	1 =& n_{0}+n_\mathtt{on}+n_\mathtt{off}+\sum \mathrm{wp}(\sigma_\mathtt{f})+n_\mathtt{v^*}, \tag{A4}\\
	\vec{0} =& \sum \vec{f}_\mathtt{f}, \tag{A6}\\
	\vec{0} =& \sum \vec{f}_\mathtt{v}+\sum \vec{f}_\mathtt{v^*}, \tag{A7}\\
	1=&n_\mathtt{off}+\sum \mathrm{wp}(\sigma_\mathtt{v})+\sum \mathrm{wp}(\sigma_\mathtt{v^*,bot}), \tag{A8} \\
	\vec{0} =& \sum \vec{q}_\mathtt{off}+\sum \vec{p}_\mathtt{v} + \sum \vec{p}_\mathtt{v^*,bot}. \tag{A9}
\end{align}
The total number of faults are $n_{0}+n_\mathtt{on}+n_\mathtt{off}+n_\mathtt{f}+n_\mathtt{v}+n_\mathtt{v^*} \leq d-1$, which means that  $n_\mathtt{off}+n_\mathtt{v}+n_\mathtt{v^*} \leq d-1-n_{0}-n_\mathtt{on}-n_\mathtt{f}$ (where $n_\mathtt{f}\geq 1$). Consider the following cases:


(2.a) If $n_{0}=1$ or $n_\mathtt{on}\geq 1$ or $n_\mathtt{f}\geq 2$, we have $n_\mathtt{off}+n_\mathtt{v}+n_\mathtt{v^*}\leq d-3$. In this case, \cref{eq:mainA7,eq:mainA8,eq:mainA9} contradict Condition \ref{con:con2} (where $\mathtt{q_{off}}$ fault corresponds to $\mathtt{q_{2D}}$ fault, and $\mathtt{v}$ and $\mathtt{v^*}$ faults correspond to $\mathtt{f_{2D}}$ fault).

(2.b) If $n_{0}=0,n_\mathtt{on}=0,$ and $n_\mathtt{f}=1$, we find that $n_\mathtt{off}+n_\mathtt{f}+n_\mathtt{v}+n_\mathtt{v^*}\leq d-1$ and $n_\mathtt{off}+n_\mathtt{v}+n_\mathtt{v^*}\leq d-2$. Let us divide this case into the following subcases (where some subcases may overlap):
\begin{enumerate}[label=(\roman*)]
	\item If $n_\mathtt{v} \geq 1$, then $n_\mathtt{off}+n_\mathtt{f}+n_\mathtt{v^*}\leq d-2$. In this case, \cref{eq:mainA3,eq:mainA4,eq:mainA6} contradict Condition \ref{con:con3} (where $\mathtt{q_{off}}$ and $\mathtt{q_{v^*}}$ faults correspond to $\mathtt{q_{2D}}$ fault, and $\mathtt{f}$ fault corresponds to $\mathtt{f_{2D}}$ fault).
	\item If $n_\mathtt{v}=0$ and $n_\mathtt{v^*}=0$, then \cref{eq:mainA8,eq:mainA9} contradict Condition \ref{con:con0} (where $\mathtt{q_{off}}$ fault corresponds to $\mathtt{q_{2D}}$ fault).
	\item If $n_\mathtt{off}=0$, then $n_\mathtt{v}+n_\mathtt{v^*}\leq d-2$ and \cref{eq:mainA7,eq:mainA8,eq:mainA9} contradict Condition \ref{con:con1} (where $\mathtt{v}$ and $\mathtt{v^*}$ faults correspond to $\mathtt{f_{2D}}$ fault).
	\item If $n_\mathtt{off}\geq 1$, $n_\mathtt{v}=0$, and $n_\mathtt{v^*}=1$, then $n_\mathtt{off}+n_\mathtt{v^*}\leq d-2$ and \cref{eq:mainA7,eq:mainA8,eq:mainA9} contradict Condition \ref{con:con3} (where $\mathtt{q_{off}}$ fault correspond to $\mathtt{q_{2D}}$ fault, and $\mathtt{v^*}$ fault corresponds to $\mathtt{f_{2D}}$ fault).
	\item If $n_\mathtt{off}\geq 1$, $n_\mathtt{v}=0$, $n_\mathtt{v^*}\geq 2$, and $n_\mathtt{off}+n_\mathtt{f}+n_\mathtt{v^*} \leq d-2$, then \cref{eq:mainA3,eq:mainA4,eq:mainA6} contradict Condition \ref{con:con3} (where $\mathtt{q_{off}}$ and $\mathtt{q_{v^*}}$ faults correspond to $\mathtt{q_{2D}}$ fault, and $\mathtt{f}$ fault corresponds to $\mathtt{f_{2D}}$ fault).
	\item If $n_\mathtt{off}\geq 1$, $n_\mathtt{v}=0$, $n_\mathtt{v^*}\geq 2$, and $n_\mathtt{off}+n_\mathtt{f}+n_\mathtt{v^*}= d-1$, then \cref{eq:mainA1,eq:mainA2,eq:mainA6,eq:mainA7,eq:mainA8,eq:mainA9} contradict Condition \ref{con:con4} (where $\mathtt{q_{off}}$, $\mathtt{q_{f}}$, and $\mathtt{q_{v^*}}$ faults correspond to $\mathtt{q_{2D}}$, $\mathtt{f_{2D}}$, and $\mathtt{v^*_{2D}}$ faults, respectively).
\end{enumerate}

\textit{Case 3}: $n_\mathtt{f} = 0$ and $n_\mathtt{cap}\geq 1$. The main equations can be simplified as follows:
\begin{align}
	0 =& n_{0}+n_\mathtt{on}+\sum\mathrm{wp}(\sigma_\mathtt{v})+\sum\mathrm{wp}(\sigma_\mathtt{v^*,cen}) \nonumber \\
	&+\sum\mathrm{wp}(\sigma_\mathtt{cap}), \tag{A1}\\
	\vec{0} =& \sum \vec{q}_\mathtt{on} + \sum \vec{p}_\mathtt{v} + \sum \vec{p}_\mathtt{v^*,cen} + \sum \vec{p}_\mathtt{cap}, \tag{A2} \\
	\vec{0} =& \sum \vec{q}_\mathtt{on}+\sum \vec{q}_\mathtt{off}+\sum \vec{q}_\mathtt{v^*} + \sum \vec{p}_\mathtt{cap}, \tag{A3}\\
	1 =& n_{0}+n_\mathtt{on}+n_\mathtt{off}+n_\mathtt{v^*} + \sum\mathrm{wp}(\sigma_\mathtt{cap}), \tag{A4}\\
	\vec{0} =& \sum \vec{f}_\mathtt{cap}, \tag{A5}\\
	\vec{0} =& \sum \vec{f}_\mathtt{v}+\sum \vec{f}_\mathtt{v^*}, \tag{A7}\\
	1=&n_\mathtt{off}+\sum \mathrm{wp}(\sigma_\mathtt{v})+\sum \mathrm{wp}(\sigma_\mathtt{v^*,bot}), \tag{A8} \\
	\vec{0} =& \sum \vec{q}_\mathtt{off}+\sum \vec{p}_\mathtt{v} + \sum \vec{p}_\mathtt{v^*,bot}. \tag{A9}
\end{align}
The total number of faults are $n_\mathtt{0}+n_\mathtt{on}+n_\mathtt{off}+n_\mathtt{v}+n_\mathtt{v^*}+n_\mathtt{cap}\leq d-1$, which means that $n_\mathtt{off}+n_\mathtt{v}+n_\mathtt{v^*}\leq d-1-n_{0}-n_\mathtt{on}-n_\mathtt{cap}$ (where $n_\mathtt{cap} \geq 1$). Consider the following cases:

(3.a) If $n_\mathtt{0}\geq 1$ or $n_\mathtt{on}\geq 1$ or $n_\mathtt{cap}\geq 2$, then $n_\mathtt{off}+n_\mathtt{v}+n_\mathtt{v^*}\leq d-3$. In this case, \cref{eq:mainA7,eq:mainA8,eq:mainA9} contradict Condition \ref{con:con2} (where $\mathtt{q_{off}}$ fault corresponds to $\mathtt{q_{2D}}$ fault, and $\mathtt{v}$ and $\mathtt{v^*}$ faults correspond to $\mathtt{f_{2D}}$ fault).

(3.b) If $n_\mathtt{0}=0, n_\mathtt{on}=0,$ and $n_\mathtt{cap}=1$, we find that $n_\mathtt{off}+n_\mathtt{v}+n_\mathtt{v^*}+n_\mathtt{cap}\leq d-1$ and $n_\mathtt{off}+n_\mathtt{v}+n_\mathtt{v^*}\leq d-2$. Let us divide the proof into the following subcases (where some subcases may overlap):
\begin{enumerate}[label=(\roman*)]
	\item If $n_\mathtt{v}+n_\mathtt{v^*}=0$, then \cref{eq:mainA8,eq:mainA9} contradict Condition \ref{con:con0} (where $\mathtt{q_{off}}$ fault corresponds to $\mathtt{q_{2D}}$ fault).
	\item If $n_\mathtt{v}+n_\mathtt{v^*}=1$, then \cref{eq:mainA7,eq:mainA8,eq:mainA9} contradict Condition \ref{con:con3} (where $\mathtt{q_{off}}$ fault corresponds to $\mathtt{q_{2D}}$ fault, and $\mathtt{v}$ and $\mathtt{v^*}$ faults correspond to $\mathtt{f_{2D}}$ fault).
	\item If $n_\mathtt{off}=0$, then $n_\mathtt{v}+n_\mathtt{v^*}\leq d-2$. In this case, \cref{eq:mainA7,eq:mainA8,eq:mainA9} contradict Condition \ref{con:con1} (where $\mathtt{v}$ and $\mathtt{v^*}$ faults correspond to $\mathtt{f_{2D}}$ fault).
	\item If $n_\mathtt{off}+n_\mathtt{v}+n_\mathtt{v^*}+n_\mathtt{cap}\leq d-2$ (or equivalently, $n_\mathtt{off}+n_\mathtt{v}+n_\mathtt{v^*}\leq d-3$), then \cref{eq:mainA7,eq:mainA8,eq:mainA9} contradict Condition \ref{con:con2} (where $\mathtt{q_{off}}$ fault corresponds to $\mathtt{q_{2D}}$ fault, and $\mathtt{v}$ and $\mathtt{v^*}$ faults correspond to $\mathtt{f_{2D}}$ fault).
	\item If $n_\mathtt{off}\geq1$, $n_\mathtt{v}+n_\mathtt{v^*}\geq2$, and $n_\mathtt{off}+n_\mathtt{v}+n_\mathtt{v^*}+n_\mathtt{cap}= d-1$, then \cref{eq:mainA1,eq:mainA2,eq:mainA5,eq:mainA7,eq:mainA8,eq:mainA9} contradict Condition \ref{con:con5} (where $\mathtt{q_{off}}$, $\mathtt{v}$, $\mathtt{v^*}$, $\mathtt{cap}$ faults correspond to $\mathtt{q_{2D}}$, $\mathtt{f_{2D}}$, $\mathtt{v^*_{2D}}$, and $\mathtt{cap_{2D}}$ faults, respectively).
\end{enumerate}

\textit{Case 4}: $n_\mathtt{f} \geq 1$ and $n_\mathtt{cap}\geq 1$ (the main equations cannot be simplified in this case). From the fact that the total number of faults is at most $d-1$, we have $n_\mathtt{off}+n_\mathtt{v}+n_\mathtt{v^*} \leq d-3$. In this case, we find that \cref{eq:mainA7,eq:mainA8,eq:mainA9} contradict Condition \ref{con:con2} (where $\mathtt{q_{off}}$ fault corresponds to $\mathtt{q_{2D}}$ fault, and $\mathtt{v}$ and $\mathtt{v^*}$ faults correspond to $\mathtt{f_{2D}}$ fault).\\

So far, we have shown that if Conditions \ref{con:con1} to \ref{con:con5} are satisfied and all faults give rise to purely $Z$-type errors, then there is no fault combination arising from up to $d-1$ faults whose combined error is a logical $Z$ operator and its cumulative flag vector is zero. Because the circuits for each pair of $X$-type and $Z$-type generators use the same CNOT ordering, the same analysis is also applicable to the case of purely $X$-type errors; i.e., if Conditions \ref{con:con1} to \ref{con:con5} are satisfied and all faults give rise to purely $X$-type errors, then there is no fault combination arising from up to $d-1$ faults whose combined error is a logical $X$ operator and its cumulative flag vector is zero. In the next part of the proof, we will use these results to show that $\mathcal{F}_t$ is distinguishable.

Let us consider a fault combination whose combined error is of mixed type. Let $t_x$ and $t_z$ denote the total number of faults during the measurements of $X$-type and $Z$-type generators, and let $u_x$, $u_y$, $u_z$ denote the number of qubit faults which give $X$-type, $Y$-type, and $Z$-type errors, respectively. Suppose that the fault combination arises from no more than $d-1$ faults, we have $t_x+t_z+u_x+u_y+u_z \leq d-1$. Next, observe that $t_x$ faults during the measurement of $X$-type generators cannot cause a $Z$-type error of weight more than $t_x$, and $t_z$ faults during the measurement of $Z$-type generators cannot cause a $X$-type error of weight more than $t_z$. Thus, the $Z$ part of the combined error and the cumulative flag vector corresponding to $Z$-type generators can be considered as an error and a cumulative flag vector arising from $t_z+t_x+u_z+u_y \leq d-1$ faults which give rise to purely $Z$-type errors. Similarly, the $X$ part of the combined error and the cumulative flag vector corresponding to $X$-type generators can be considered as an error and a cumulative flag vector arising from $t_x+t_z+u_x+u_y \leq d-1$ faults which give rise to purely $X$-type errors. Recall that there is no fault combination arising from up to $d-1$ faults whose combined error is a logical $X$ (or a logical $Z$) operator and its cumulative flag vector is zero when all faults give rise to purely $X$-type (or purely $Z$-type) errors. Using this, we find that for any fault combination arising from $d-1$ faults, it cannot correspond to a nontrivial logical operator and the zero cumulative flag vector. That is, there is no fault combination corresponding to a nontrivial logical operator and the zero cumulative flag vector in $\mathcal{F}_{2t}$ where $2t=d-1$. By \cref{prop:2t}, this implies that $\mathcal{F}_t$ is distinguishable. 

\section{Fault-tolerant error correction protocol for a general stabilizer code}
\label{sec:app:protocol_general}

In \cref{subsec:FTEC_ana}, we construct an FTEC protocol for a capped color code in H form of any distance in which its fault set is distinguishable. We also show that such a protocol is fault tolerant when the $r$-filter, the ideal decoder, and the distinguishable error set are defined as in \cref{def:r_filter_new,def:ideal_new,def:dist_err_CSS}. Using similar ideas, we can also construct an FTEC protocol for a general stabilizer code whose circuits for the syndrome measurement give a distinguishable fault set $\mathcal{F}_t$, i.e., a code in which $\mathcal{E}_r$ is defined by \cref{def:dist_err} instead of \cref{def:dist_err_CSS}. The outcome bundle defined for the protocol in this section is similar to the outcome bundle defined for the FTEC protocol for a capped color code, except that the syndrome $\vec{\mathbf{s}}$ and the cumulative flag vector $\vec{\mathbf{f}}$ are not separated into $X$ and $Z$ parts. We can also build a list of all possible fault combinations and their corresponding combined error and cumulative vector from the distinguishable fault set $\mathcal{F}_t$. The FTEC protocol for a general stabilizer code is as follows:
\\

\noindent\textbf{FTEC protocol for a stabilizer code whose syndrome measurement circuits give a distinguishable fault set}

During a single round of full syndrome measurement, measure the all generators in any order. Perform full syndrome measurements until the outcome bundles $(\vec{\mathbf{s}},\vec{\mathbf{f}})$ are repeated $t+1$ times in a row. Afterwards, do the following:

\begin{enumerate}
	\item Determine an EC operator $F$ using the list of possible fault combinations as follows:
	\begin{enumerate}
		\item If there is a fault combination on the list whose syndrome and cumulative flag vector are $\vec{\mathbf{s}}$ and $\vec{\mathbf{f}}$, then $F$ is the combined error of such a fault combination. (If there are more than one fault combination corresponding to $\vec{\mathbf{s}}$ and $\vec{\mathbf{f}}$, a combined error of any of such fault combinations will work since they are logically equivalent.)
		\item If none of the fault combinations on the list corresponds to $\vec{\mathbf{s}}$ and $\vec{\mathbf{f}}$, then $F$ can be any Pauli operator whose syndrome is $\vec{\mathbf{s}}$.
	\end{enumerate}
	\item Apply $F$ to the data qubits to perform error correction.
\end{enumerate}

To verify that the FTEC protocol for a general stabilizer code satisfies both properties of an FTEC gadget according to the revised definition (\cref{def:FT_gadget_new}), we can use an analysis similar to that presented in \cref{subsec:FTEC_ana}, except that $\mathcal{E}_r$ is defined by \cref{def:dist_err} instead of \cref{def:dist_err_CSS} and the errors in the analysis ($E_\mathrm{in}, E_a,$ and $E_b$) need not be separated into $X$ and $Z$ parts.

\end{document}